\documentclass[11pt, letterpaper]{article}

\usepackage{graphicx} 
\usepackage{amsmath, amsfonts, amssymb, amsthm}
\usepackage{bbm}
\usepackage{fullpage}
\usepackage{algorithmic}
\usepackage{algorithm}
\usepackage{hyperref}
\usepackage{enumitem}
\usepackage[capitalise,nameinlink,noabbrev]{cleveref}
\usepackage{xcolor}
 \makeatletter
  \renewcommand{\ALG@name}{IOP}
  \makeatother
\usepackage{tabularx}

\theoremstyle{plain} 

   \newtheorem{thm}{Theorem}[section]
   
   \newtheorem{fact}[thm]{Fact}
   \newtheorem{lemma}[thm]{Lemma}
   
   \newtheorem{claim}[thm]{Claim}
   \newtheorem{remark}[thm]{Remark}
   \newtheorem{definition}[thm]{Definition}

\newtheorem{state}{State Function}
\newcommand\skipi{{\vskip 10pt}}

\newcommand{\poly}{{\sf poly}}

\newcommand{\st}{{\sf State}}
\newcommand{\eps}{\varepsilon}

\newcommand{\vs}{{\sf v}}
\newcommand{\p}{{\sf p}}
\newcommand{\T}{\mathcal{T}}
\newcommand{\B}{\mathcal{B}}
\newcommand\inner[2]{\langle{#1},{#2}\rangle}
\newcommand{\norm}[1]{\left\lVert#1\right\rVert}

\newcommand{\agr}{{\sf agr}}

\newcommand{\mc}{\mathcal}
\DeclareMathOperator{\RS}{RS}
\DeclareMathOperator{\RM}{RM}

\newcommand{\E}{\mathop{\mathbb{E}}}
\newcommand{\Ff}{\mathbb{F}}

\newcommand{\List}{{\sf List}}

\newcommand{\rd}{{\sf rd}}

\newcommand{\err}{{\sf err}}
\newcommand{\Lc}{\mathcal{L}}
\newcommand{\tl}{\tau_{{\sf line}}}
\newcommand{\Lg}{\mathcal{L}_{{\sf good}}}

\newcommand{\Fqt}{\ensuremath{\mathbb{F}_{q}^{2}}}
\newcommand{\Fqm}{\ensuremath{\mathbb{F}_{q}^{m}}}
\newcommand{\Fqp}{\ensuremath{\mathbb{F}_{q'}}}

\newcommand{\Fill}{{\sf Fill}}
\newcommand{\wh}{\widehat}

\newcommand{\Side}{{\sf Side}}
\newcommand{\Q}{{\sf Q}}
\newcommand{\iQ}{{\sf inpQ}}
\newcommand{\pQ}{{\sf pfQ}}

\newcommand{\rics}{{\sf R1CS}}
\newcommand{\Len}{{\sf Len}}

\newcommand{\quo}{{\sf Quo}}
\newcommand{\Quo}{{\sf Quo}}

\newcommand{\rsi}{\texttt{RS-IOPP}}
\newcommand{\rsp}{\texttt{RS-Poly}}
\newcommand{\irmp}{\texttt{iRM-Poly}}
\newcommand{\irmb}{\texttt{iRMBatch}}
\newcommand{\rmp}{\texttt{RM-Poly}}
\newcommand{\ricsp}{\texttt{R1CS-Poly}}

\newcommand{\rmi}{\texttt{RM-IOPP}}
\newcommand{\irmi}{\texttt{iRM-IOPP}}
\newcommand{\gp}{\texttt{GenericPolyIOP}}
\newcommand{\batch}{\texttt{BatchedIOPP}}
\newcommand{\Eval}{\texttt{Eval}}
\newcommand{\rsb}{\texttt{RSBatch}}
\newcommand{\compile}{\texttt{Compile}}

\newcommand{\combine}{{\sf combine}}
\newcommand{\bat}{{\sf batch}}
\newcommand{\comp}{{\sf comp}}
\newcommand{\ec}{\eps_{\sf comp}}
\newcommand{\fij}{F^{(i)}_j}
\newcommand{\fhatij}{\wh{F}^{(i)}_j}
\newcommand{\fhatipj}{\wh{F}^{(i')}_j}
\newcommand{\cipj}{\mathcal{C}^{(i')}_j}

\newcommand{\fcalij}{\mathcal{F}^{(i)}_j}
\newcommand{\fcalpij}{\mathcal{F}^{(i)}_j}
\newcommand{\qenc}{q_{{\sf enc}}}
\newcommand{\Fqe}{\Ff_{q_{{\sf enc}}}}
\newcommand{\trmb}{\texttt{tRMBatch}}
\newcommand{\et}{\eps_{{\sf test}}}
\newcommand{\hij}{H^{(i)}_j}
\newcommand{\Disc}{{\sf Disc}}
\newcommand{\prox}{{\sf prox}}
\newcommand{\enc}{{\sf enc}}
\newcommand{\good}{{\sf good}}
\newcommand{\proj}{{\sf proj}}

\crefname{algorithm}{IOP}{IOPs}

\DeclareMathOperator{\iRM}{iRM}
\makeatletter
  \renewcommand{\ALG@name}{IOP}
\makeatother

\newcounter{protocol}

\title{Improved Round-by-round Soundness IOPs via 
Reed-Muller Codes}
\author{Dor Minzer\thanks{Department of Mathematics, Massachusetts Institute of Technology. Supported by NSF CCF award 2227876 and NSF CAREER award 2239160.}  \and Kai Zhe Zheng\thanks{Department of Mathematics, Massachusetts Institute of Technology, Cambridge, USA. Supported by the NSF GRFP DGE-2141064 and NSF CCF award 222787.}}
\date{\vspace{-5ex}}

\begin{document}
\maketitle
\begin{abstract}
    We give an IOPP (interactive oracle proof of proximity) for trivariate Reed-Muller codes that achieves 
    the best known query
    complexity in some range of security parameters. Specifically, for degree $d$ and security parameter $\lambda\leq \frac{\log^2 d}{\log\log d}$  , our IOPP has $2^{-\lambda}$ round-by-round soundness, $O(\lambda)$ queries, $O(\log\log d)$ rounds 
    and $O(d)$ length. This improves 
    upon the FRI [Ben-Sasson, Bentov, Horesh, Riabzev, ICALP 2018] and the STIR [Arnon, Chiesa, Fenzi, Yogev, Crypto 2024] 
    IOPPs for Reed-Solomon codes, that have larger query and round complexity standing at $O(\lambda \log d)$ and  $O(\log d+\lambda\log\log d)$
    respectively. We use our IOPP to give an IOP for the NP-complete language $\rics$ with the same parameters.

    Our construction is based on the line versus point test in the low-soundness regime. Compared to the axis parallel test (which is used in all prior works), the general affine lines test has improved soundness, which is the main source of our improved soundness.  
    
    Using this test involves several complications, most significantly that projection to affine lines does not preserve individual degrees, and we show how to overcome these difficulties.
    En route, we extend some existing machinery to more general settings. Specifically, we give proximity generators for Reed-Muller codes, show a more systematic way
    of handling ``side conditions'' in IOP constructions, and generalize the compiling procedure of [Arnon, Chiesa, Fenzi, Yogev, Crypto 2024] to general codes.
\end{abstract}
\section{Introduction}

A central question that appears in both complexity theory and cryptography is the following: 
\begin{center}
   \emph{How can a computationally weak party verify a large computation that is infeasible for them to perform?}
\end{center}
Towards answering this question a number of proof systems were developed. Some of the earliest models 
addressing it are the interactive protocols model and its multi-party counterpart~\cite{GMR85,BGKW88}, which were shown to be as powerful as the classes PSPACE~\cite{LFKN92,Shamir92} and NEXP respectively~\cite{BFL91}. These models later led to Probabilstically Checkable Proofs (PCPs)~\cite{FGLSS, AroraSafra, ALMSS}, and their interactive variant, Interactive Oracle Proofs (IOPs)~\cite{BCS, RRR}, 
which are the main topic of this paper.

In a PCP, a prover attempts to convince a verifier of the validity of a statement. For example, the prover may want to convince the verifier that the output of a known Boolean circuit $C$ on a given input $x\in\{0,1\}^n$ is $b$. To convince the verifier, the prover provides an additional string which is often referred to as proof or witness. The hope is that instead checking the statement by running the circuit itself, the verifier could use the witness to validate the computation more efficiently. Specifically, we wish the verifier to make a small number of queries to the proof and then decide whether to accept or reject the proof. The verifier should accept with probability $1$ if the statement is correct and the prover provided a valid proof, and should accept with probability at most $1/3$ if the statement is incorrect, regardless of the proof the prover provided.

In this language, the celebrated PCP theorem \cite{FGLSS, AroraSafra, ALMSS} states that every statement in NP has a polynomial sized proof that the verifier can check using $O(1)$ queries. The PCP theorem has seen bountiful applications in theoretical computer science, perhaps most notably to the field of \emph{hardness of approximation}, wherein it serves as the starting point of essentially all NP-hardness results for approximation problems. In contrast, PCPs have not received as much attention for their namesake task -- checking proofs! This is partly due to the fact PCP constructions often have a large (but still polynomial) size, as well as the fact that the witness is often rather complicated. 
Even though much progress has been made in the past three decades towards constructing more efficient PCPs, leading to constructions with quasilinear size \cite{BSS,Dinur07,MR10,BMV}, these constructions are still impractical. At present, even constructing linear sized PCPs is a major open problem, let alone practical ones.

To remedy this situation, an interactive variant of PCPs, called Interactive Oracle Proofs (IOPs in short), was introduced \cite{BCS, RRR}. In an IOP, a prover and a verifier exchange messages over multiple rounds. Each prover's message is itself a (long) proof, which the verifier has query access to. The goal is once again for the prover to convince the verifier of some statement, and for the verifier to validate the statement while doing much less work than was necessary to check the statement themselves. It can be seen that PCPs are special cases of IOPs that use only one round of interaction. 
Thus, it stands to reason that IOPs are more powerful, and one may hope to construct more practical IOPs
that outperform the best known PCPs. 

\subsection{Prior Works}
Much effort has gone into constructing IOPs that both theoretically beat PCPs and are also of practical use~\cite{BCGV16,BBCGGH17,BCGRS17,BBHR18,RZR24,BLNR22,RZR22,acy,acfy,WHIR}. The basic parameters of interest in IOPs include the alphabet size, the number of rounds, the number of queries, the completeness, the soundness, and the complexity of the verifier as well as of the prover. 
For practical purposes, one often cares about additional parameters:
\begin{enumerate}
    \item The proof length: this is the total number of symbols the prover sends to the verifier throughout the interaction. Ideally, it should be $O(n)$, and with a good constant factor.
    \item The round-by-round soundness of the protocol, a refinement of the standard soundness, motivated by non-interactive variants of IOPs called SNARKS~\cite{BCS}. Each partial transcript of 
    the protocol is assigned a state which is either ``doomed'' or ``not doomed''. The state at the start of the protocol is doomed if and only if the prover attempts to prove a wrong statement, and in the end the verifier rejects if and only if it is in a doomed state. Thus, the round-by-round soundness of a protocol is
    the maximum probability over all rounds and over all partial transcripts that the state turns from ``doomed'' to ``not doomed'' in a single round.
\end{enumerate}
Most relevant to the current paper are the works~\cite{acy,acfy}. 
Both of these works achieve reasonable proof length and an improved notion of soundness compared to prior works, but differ in the type of soundness they get as well as in the query and round complexity. More specifically,~\cite{acy} achieves polynomial length, inverse polynomial (standard) soundness with $O(\log\log n)$ query and round complexity, whereas~\cite{acfy} achieves linear length, $2^{-\lambda}$ round-by-round soundness, $O(\log n)$ rounds and $O(\log n + \lambda \log \log n)$ queries. In other words, the result of~\cite{acfy} has better length and the stronger round-by-round soundness, but it has
worse query and round complexity. Can one get the best
of both worlds and achieve round-by-round soundness 
with smaller number of rounds and queries?

\subsection{Main Result}
In this section we discuss the main result of this paper, which gives improved IOPs for all NP languages. To state our results, we define the Rank-One-Constraint-Satisfaction problem ($\rics$).

\paragraph{The $\rics$ problem:} As is standard in the area, we construct IOPs for a conveniently chosen, specific NP-Complete language $\rics$. 
Once an IOP for $\rics$ is established, IOPs for the class NP readily follows. Here and throughout, the problem $\rics$ is a problem wherein the input consists of a field $\mathbb{F}$ and $n$ by $n$ matrices $A,B,C$ over $\mathbb{F}$. The task is to determine whether there is a vector $z\in\mathbb{F}^n$ such that $z_{1} = 1$ and 
\[
\left(\sum\limits_{j}A_{i,j} z_j\right)
\left(\sum\limits_{j}B_{i,j} z_j\right)
=
\left(\sum\limits_{j}C_{i,j} z_j\right).
\]
In other words, an instance of $\rics$ consists of
a collection of rank $1$ quadratic functions as well
as a restriction constraint, and the task is to determine whether the system is satisfiable or not. 

Our main result is the following theorem:
\begin{thm} \label{thm: main rbr}
Fix a security parameter $\lambda$ and let $\Ff_q, \Fqe, \Ff_{q'}$ be characteristic $2$ fields such that each of $\Ff_q$ and $\Fqe$ is a subfield of $\Ff_{q'}$, and such that the sizes of the fields satisfy
\[
\qenc \geq \Omega(n^{1/3}) \quad, \quad q \geq \Omega(n^{1/3}) \quad \text{,and} \quad q' \geq \frac{2^{\lambda + 10} \cdot 
\max(\qenc, q)^4}{n^{2/3}}.
\]
Then there is an IOP for $\rics$ with the following guarantees. 
\begin{itemize}
    \item Input: An $\rics$ instance of size $n$.
    \item Completeness: If the instance is satisfiable, then the honest prover makes the verifier accept with probability $1$.
    \item Round-by-round soundness: $2^{-\lambda}$.
    \item Initial state: The initial state is doomed if and only if the instance is not satisfiable.
    \item Round complexity: $O(\log \log n)$.
    \item Proof query complexity: $O\left(\frac{\lambda}{\log(q/n^{1/3})}\right) + O\left(\frac{\lambda^2}{\log^2(\qenc/n^{1/3})} \cdot \log \log(n)\right)$.
    \item Alphabet size: $q'$.
    \item Length: $O(q^3) + O(\qenc^2)$.
\end{itemize}
\end{thm}

We next compare~\cref{thm: main rbr} to the results of~\cite{acy,acfy}.

\subsubsection{Good Round-by-round Soundness IOPs}
Choosing $q = \Theta(n^{1/3})$, $\qenc = \Theta(n^{2/5})$, and $q' = 2^{\lambda + 10} \cdot n^{14/15}$ in~\cref{thm: main rbr} we get the following result \footnote{Note that for $q$ and $\qenc$ to be subfields of $q'$ we need $q^{c_1} = \qenc^{c_2} = q'$ for some positive integers $c_1, c_2$. Concretely, we can achieve the desired setting by taking $q = 2^{3 \cdot s}$, $\qenc = 2^{5 \cdot s}$, and $q' = 2^{\lambda + c \cdot s}$ for some appropriate $s$ and large enough, but constant, $c$ which makes $\lambda + c \cdot s$ divisible by $15s$.} :
\begin{thm} \label{thm: main rbr set}
Fix a security parameter $\lambda$. There is an IOP for $\rics$ with the following guarantees. 
\begin{itemize}
    \item Input: An $\rics$ instance of size $n$.
    \item Completeness: If the instance is satisfiable, then the honest prover makes the verifier accept with probability $1$.
    \item Round-by-round soundness: $2^{-\lambda}$.
    \item Initial State: The initial state is doomed if and only if the instance is not satisfiable.
    \item Round complexity: $O(\log \log n)$.
    \item Proof query complexity: $O\left(\lambda \right) + O\left(\frac{\lambda^2}{\log^2(n)} \cdot \log \log(n)\right)$.
    \item Alphabet size: $2^{\lambda} \cdot \poly(n)$.
    \item Length: $O(n)$.
\end{itemize}
\end{thm}
As long as $\lambda\leq \log^2 n/\log\log n$ (which is a reasonable setting for the security parameter), we get that the query complexity in~\cref{thm: main rbr set} is $O(\lambda)$. Thus, in this regime of security parameters,~\cref{thm: main rbr set} improves upon the result of~\cite{acfy} both
in the round complexity and query complexity. We remark that the IOP behind~\cref{thm: main rbr set} could potentially also be practical: the protocol behind it is quite simple, the prover and verifier are efficient, and the implicit constant in the proof length in~\cref{thm: main rbr set} is moderate and can be greatly improved with more work. We discuss practicality further in~\cref{sec:concrete_eff}.
From a theoretical perspective, it seems that $O(\lambda)$ may be
as low as the query complexity can go for linear length IOPs, and that perhaps $O(\lambda/\log n)$ is possible for polynomial length IOPs.
This is due to the fact that all known IOPs (including ours) are based on error correcting
codes, and error correcting codes with linear length have distance bounded
away from $1$. We believe that one may hope to improve the round complexity all the way down to $O(1)$, perhaps at a mild increase in the query complexity.

\subsubsection{Polynomially Small Soundness IOPs}
Taking $q_{\enc} = q = n^{2}$ in~\cref{thm: main rbr} and $\lambda = \Omega(\log n)$, we get the following result:
\begin{thm} \label{thm: main standard}
There is an IOP for $\rics$ with the following guarantees. 
\begin{itemize}
    \item Input: A $\rics$ instance of size $n$.
    \item Completeness: If the instance is satisfiable, then the honest prover makes the verifier accept with probability $1$.
    \item Soundness: $\frac{1}{n}$.
    \item Round complexity: $O(\log \log n)$.
    \item Proof query complexity: $O(\log \log n)$.
    \item Alphabet size: $\poly(n)$.
    \item Length: $\poly(n)$.
\end{itemize}
\end{thm}
This matches the result of~\cite{acy} via
a different protocol. We remark that with more careful analysis, the length of the proof could be made $n^{1+\eps}$ (the implicit constants in the round and query complexity would depend on $\eps$).

\subsection{Proof Overview}
Here we give an overview of the proof of \cref{thm: main rbr}. The key component is an IOP of Proximity (IOPP in short) to the total degree Reed-Muller code 
\[
\RM_{q'}[d, \Ff^3_q] = \{f: \Ff^3_q \to \Ff_{q'} \; | \; \deg(f) \leq d\}.
\]
That is, an IOP which accepts functions $f \in \RM_{q'}[d, \Ff^3_q]$ and rejects functions $f$ that are far from $\RM_{q'}[d, \Ff^3_q]$. Henceforth, we use the term IOPP to refer to an IOP for a code in which the verifier should accept functions from the code, and reject functions far from the code. Prior works have mostly used IOPPs for Reed-Solomon codes, and much of the gain in our result comes from the fact we work with Reed-Muller codes. 

\begin{remark}
A few remarks are in order:
\begin{enumerate}
    \item While our final IOP for NP uses an IOPP to the $3$-variate total degree Reed-Muller code, our techniques can be made to work with Reed-Muller codes over any number of variables $m \geq 3$. It would be interesting to also be able to handle the $m = 2$ case and use this IOPP in our IOP for NP. Such a result could be useful towards concrete efficiency and we comment more on why our approach falls short in the $m= 2$ case later.
    \item Using our techniques one can establish improved IOPPs for Reed-Solomon codes of constant rate. 
    We defer further discussion to~\cref{sec:RS_IOPP}.
\end{enumerate}
\end{remark}
\subsubsection{Algebraization}
Fixing an instance of $\rics$ and identify its set of variables with a suitably chosen subset $H^3\subseteq \mathbb{F}_q^3$ of size $|H|^3 = n$. One can now think
of an assignment to the instance as a function $f\colon H^3\to \mathbb{F}_q$, and then consider its low-degree extension to $\mathbb{F}_q^3$, which would be in $\RM_{q}[d, \Ff^3_q]$. We could 
also think of the vectors $Af, Bf, Cf$ (where $A,B,C$ are the matrices defining the $\rics$ instance) as functions and consider their low-degree
extensions $f_A,f_B,f_C\in \RM_{q}[d, \Ff^3_q]$. With this set-up in mind,
verifying whether the given instance is satisfiable amounts to:
\begin{enumerate}
    \item Checking that $f_A f_B - f_C$ vanishes on $H^3$.
    \item Checking that $f$ evaluates to $1$ on the point in $H^3$ corresponding to the first variable.
\end{enumerate}
Of course, verifying these two items without any additional help from
the prover requires many queries. Thus, to aid the verifier, the prover
supplies functions tri-variate functions $g_1,g_2,g_3$ that are supposedly low-degree, as well as functions $f_A',f_B',f_C',f'$ that are supposedly $f_A,f_B,f_C$ and $f$. The intention is that 
\[
f_A f_B - f_C = V_H(x)g_1+V_H(y)g_2+V_H(z)g_3
\qquad\text{where }
V_H(x) = \prod\limits_{\alpha\in H}(x-\alpha),
\]
thereby $g_1,g_2,g_3$ would certify that $f_A f_B - f_C$ vanishes on $H^3$.
Assuming all of the above functions are indeed low-degree, this is easy to verify by random sampling, so our task effectively reduces to:
\begin{enumerate}
    \item Verifying that each one of the supplied functions 
    $f_A',f_B',f_C',f',g_1,g_2,g_3$ are low-degree. This is
    where we need to use IOPPs.
    \item Checking that the relation between each one of $f_A',f_B',f_C'$
    and $f$ holds. This is resolved using the sum-check protocol as we done in prior works, except that we have to use a multi-variate version.
\end{enumerate}
Henceforth, we focus our discussion on the construction of IOPPs, which
is where our gain in~\cref{thm: main rbr} comes from.

\subsubsection{IOPPs for Constant Rate Reed-Muller Codes} \label{sec: const rate rm intro}
Fix a function we wish to test from the above section, say $f$ for simplicity. For the protocol described above to have linear length we must take $q = O(n^{1/3})$, meaning that the code we wish to test proximity to is a Reed-Muller 
code with constant rate.  In order to test, we first observe that if $f$
is far from $\RM_{q'}[d, \Ff^3_q]$, say having agreement at most $\eps$ with each codeword there, then
\begin{equation}\label{intro:eq1}
\Pr_{\substack{P\subseteq \mathbb{F}_q^3\\\text{random plane}}}
\left[f|_{P}\text{ has more than $1.1\eps$ agreement with $\RM_{q'}[d, \Ff^2_q]$}\right]\leq \frac{1}{q^{\Omega(1)}}.
\end{equation}
Thus, if $f$ is far from the Reed-Muller code and we choose $T = O(\lambda/\log q)$ planes $P_1,\ldots,P_T$ uniformly, except with probability $2^{-\lambda}$, at least one of $f|_{P_i}$ will 
be far from the bivariate Reed-Muller code (and in fact, many of them 
will be far from it). This $\log q$ factor is essentially where the 
gain in the query complexity in~\cref{thm: main rbr set} comes from. The proof of~\eqref{intro:eq1} relies on the soundness of the
well-known line versus point test~\cite{rubinfeld1996robust,as,hkss} as well as spectral properties of the planes versus points inclusion graph,
and we refer the reader to~\cref{thm: plane vs point strong} for a 
more precise statement.

The main gain from moving to planes is that the length of the code
$\RM_{q'}[d, \Ff^2_q]$ is now $q^2 = O(n^{2/3})$. This reduction in length gives us room to move to a larger field $\mathbb{F}_{\qenc}$, with size $\qenc = q^{1 + c}$, for some $c > 0$, and thus improve the distance of the code. After doing so, we are able to further restrict down to lines randomly, in a similar fashion so that the new code we wish to test
is a Reed-Solomon code in the small rate regime. 

Overall, the discussion above reduces our IOPP in the constant rate regime is to the construction of IOPPs for small rate Reed-Solomon codes.

\subsubsection{IOPPs for Small Rate Reed-Solomon Codes}
Suppose now that we have a univariate function $g\colon \mathbb{F}_{\qenc}\to \mathbb{F}_{q'}$, and we wish to distinguish
between the case it is in the Reed-Solomon code 
\[
\RS_{\qenc}[d, \Ff_{\qenc}] = \{f: \Ff_{\qenc} \to \Ff_{q'} \; | \; \deg(f) \leq d\},
\]
and the case $g$ has at most $\eps$ agreement with each codeword there. Since there is no local tester for the Reed-Solomon code using less than $d$ queries, the verifier once again requires help from
the prover. To this end we use a (standard) embedding of
Reed-Solomon codes in bi-variate Reed-Muller codes. Specifically, if $f$ is a univariate function of degree at most $d$, then there 
is $Q\colon \mathbb{F}_{\qenc}^2\to \mathbb{F}_{q'}$ of individual
degrees at most $\sqrt{d}$ such that $f(x) = Q(x^{\sqrt{d}}, x)$. Thus,
the prover will supply the verifier with the truth table of $Q$, and now
the verifier's goal is to check that:
\begin{enumerate}
    \item $Q$ has individual degrees at most $\sqrt{d}$. Henceforth, 
    we refer to the code corresponding to functions of individual degree at most $\sqrt{d}$ as the $\sqrt{d}$ individual degree Reed-Muller code, and informally as the ``small individual degree Reed-Muller code'',
    \item Check that $Q$ and $f$ agree.
\end{enumerate}
Note that once the verifier manages to establish the first item, the second item can easily be achieves by randomly sampling $x$'s and checking if $Q(x^{\sqrt{d}},x)$ agrees with $f(x)$.

\subsubsection{IOPPs for Small Rate, Individual Degree Reed-Muller Codes}
We are thus left with the task of checking that $Q$ is a small individual degree Reed-Muller codeword. Previous works, such as~\cite{acy}, 
have used the axis-parallel low-degree test in the style of~\cite{polishchuk1994nearly} (though, solely this result is not quite sufficient for
them). We give a test with better soundness guarantees that is 
based on a combination of the line versus point test applied on top
of an axis parallel test.

To be more specific, we first choose random lines $\ell_1,\ldots,\ell_T$ 
for $T$ as above, and then test that $Q|_{\ell_1},\ldots,Q|_{\ell_T}$
are in $\RS_{\qenc}[2\sqrt{d}, \Ff_{\qenc}]$. We prove that if $Q$ is
far from being total degree $2\sqrt{d}$, then this test rejects with probability $1-2^{-\lambda}$. This test alone does not suffice however, as it 
clearly can still be the case that $Q$ has individual degrees between $\sqrt{d}+1$ and $2\sqrt{d}$ while having all line restrictions in $\RS_{\qenc}[2\sqrt{d}, \Ff_{\qenc}]$. Thus, we additionally apply the following axis parallel low-degree test: 
we choose points $\alpha, \beta$ uniformly 
and check that $Q'(y) = Q(\alpha ,y)$ and $Q''(x) = Q(x, \beta)$
both have degree at most $\sqrt{d}$. The key gain here is that once we know that $Q$ has total
degree $2\sqrt{d}$, then if it has individual degrees exceeding $\sqrt{d}$, the performance of the axis parallel test is significantly better soundness. In contrast, the axis parallel test does not have such strong soundness if we do not have the guarantee that $Q$ has low total degree.

\vspace{1ex}
In conclusion, our construction of IOPPs for small rate individual degree Reed-Muller codes and for small rate Reed-Solomon codes proceed by induction. We show how codes for degree parameter $\sqrt{d}$ for the former imply codes for degree $d$ in the latter. Each step in this process
involves $T^2$ queries, and as the number of rounds can be seen to be $O(\log\log d)$, we get the number of queries and rounds as in~\cref{thm: main rbr set}.

\subsubsection{Poly-IOPs, Anchoring, Side Conditions, and Proximity Generators}\label{sec:sc}
We finish this proof overview section by mentioning that, to carry out
the above protocol effectively, we must use several additional tools 
that appear in prior works such as~\cite{acy,acfy}. We briefly discuss these tools below.

\paragraph{Poly-IOPs:} To 
construct each IOPP as above, we first construct a relaxed version of it called a Poly-IOP. A Poly-IOP is an IOP where the prover's 
messages are guaranteed to come from some code. 
Thus, they are much easier to construct and analyze. As can be evident
from the prior description of our IOPs, 
their analysis becomes very simple once we make such assumptions on the functions sent by the prover.

As was shown in~\cite{acy}, a Poly-IOP can be turned into a proper
IOP provided that one has an IOPP for testing proximity to certain codes related to those that the prover's messages are assumed to come from in the Poly-IOP. We use this idea to turn our Poly-IOPs to IOPs as well. However, as the result of~\cite{acy}
is tailor-made for univariate functions, we have to establish a version of it suitable for our purposes, and we do so in~\cref{sec:poly-IOP}. 

\paragraph{Proximity Generators:} 
The diligent reader may notice that in the above description, we have often reduced a
task such as testing proximity of a function $f$ 
to $\RS_{\qenc}[d, \Ff_{\qenc}]$, to a conjunction of
proximity checks of multiple functions in $\RS_{\qenc}[2\sqrt{d}, \Ff_{\qenc}]$. Naively, one could run 
the suitable IOP on each one of these functions separately, however this strategy leads to significantly worse parameters. To improve upon this, we use the idea of proximity generators from~\cite{bciks}. 
Proximity generators allow us to take a linear combination of functions so
that if at least one of them is far from $\RS_{\qenc}[2\sqrt{d}, \Ff_{\qenc}]$, then the linear combination is also far from $\RS_{\qenc}[2\sqrt{d}, \Ff_{\qenc}]$ except with very small probability. This ``very small probability'' is governed by the size of the field from which we take the coefficients, and we take them from a field $\mathbb{F}_{q'}$ such that $\mathbb{F}_{q}\subseteq \mathbb{F}_{q'}$ is a subfield and $q'$ is sufficiently large compared to $q$. This is the parameter $q'$ in~\cref{thm: main rbr}. 

While \cite{bciks} gives proximity generators for Reed-Solomon codes, our IOP also requires proximity generators for Reed-Muller codes. Such results are given for general codes in \cite{BKS}, but for our purposes, we require an additional guarantee from our proximity generators called correlated agreement. At a high level, this statement says that if a random combination of functions is close to a codeword of a code $\mc{C}$, then there must be a sizable set of coordinates on which the functions simultaneously agree with some codewords. The Reed-Solomon proximity generator of \cite{bciks} satisfies this additional requirement, and indeed this feature is crucial to the analyses of FRI~\cite{BBHR18,acy}, and STIR~\cite{acfy}. Our arguments require proximity generators for Reed-Muller codes with correlated agreement, and to the best of our knowledge such proximity generators did not exist prior to this work. Specifically, we prove that the proximity generators for Reed-Solomon codes also work well for Reed-Muller codes (with correlated agreement), and our proof of this fact once again uses the soundness of the line versus point test. 

While sufficient for us, the proximity generators we give for Reed-Muller codes are qualitatively weaker than the proximity generators for Reed-Solomon codes; see~\cref{sec:proximity_gen} for more details. 
It would be interesting to give proximity generators for Reed-Muller codes with better performance.

Lastly, we remark that because we are using coefficients from $\mathbb{F}_{q'}$, the functions throughout the protocol are of the form $f\colon \mathbb{F}_q^m\to\mathbb{F}_{q'}$, which is an unusual setting for results such as the line versus point test. Fortunately, the analysis of the line versus point test works as is in this new setting, and we give a detailed analysis in the appendix.

\paragraph{Anchoring:} to transform a Poly-IOP to an
IOP we have to use the idea of anchoring from~\cite{acy,acfy}. This idea relies on a list decoding fact: once the prover sends a function $f$, say that it is supposedly from some Reed-Solomon code, then there are at most $O(1/\eps)$ Reed-Solomon codewords that have $\eps$-agreement with $f$. Thus, 
in a sense, the prover could try to pretend as if they meant each one of these, which often requires us to take a union bound and pay a factor of $O(1/\eps)$ 
in the error bound. This factor is often unaffordable though. 

The idea of anchoring essentially forces the prover to commit to one of these codewords. Namely, upon receiving $f$, the verifier samples a random point $z$, which is often not in the domain $\mathbb{F}_q$, but instead from the much larger $\mathbb{F}_{q'}$, sends this point $z$ to the prover, and receives a value $\beta$ in response. This value should be thought of as the value of the low-degree extension of $f$ at $z$, and intuitively it narrows down the possible list of Reed-Solomon codewords discussed above to only the functions that assign the value $\beta$ to the point $z$. With high probability over the random $z$, this list has size $1$, and the factor of $O(1/\eps)$ is saved.

\paragraph{Side-Conditions:}
we formalize conditions such as $f(z) = \beta$ arising in the anchoring idea as above, as well as in some other parts of our IOP, using the idea of side conditions. 
We require this idea for both Reed-Muller and Reed-Solomon codes. We focus on Reed-Solomon codes below for simplicity.
A side condition consists of a collection of points $A\subseteq\mathbb{F}_{q'}$, thought of as small, 
and a function $h\colon A\to\mathbb{F}_{q'}$ which gives the supposed values of these points. With this setup in mind, the prover wants to convince the verifier that $f\colon \mathbb{F}_q\to\mathbb{F}_{q'}$ is a degree $d$ function that is also consistent with the side conditions $(A,h)$. Towards this end, the prover could
send a function $g$ of degree at most $d-|A|$ such that
\[
f(x) = V_A(x)\cdot g(x) + h'(x),
\qquad\text{where }
V_A(x) = \prod\limits_{\alpha\in A}(x-\alpha)
\]
and $h' = \wh{h}|_{\mathbb{F}_{q}}$ where $\wh{h}$ is the low-degree extension of $h$ to $\mathbb{F}_{q'}$. This process reduces the task checking code membership check with side conditions, to the simpler tasks of checking code membership (of $g$) and identity between two functions.

\subsection{Discussion}\label{sec:variants}

\subsubsection{IOPPs for Constant Rate Reed-Solomon Codes}\label{sec:RS_IOPP}

With a little more work, we can also obtain IOPPs for constant rate Reed-Solomon codes using similar techniques. We omit the full construction and analysis however, because we do not need such an IOPP for our main result, \cref{thm: main rbr}. The guarantees of this IOPP are as follows.

\begin{thm} \label{thm: iopp const rs}
Let $q = C \cdot d$ for some constant $C > 100$ and $\eps \geq \frac{d}{q}$ be a proximity parameter. Then there exists an IOPP for $\RS_{q}[d, \Ff_q]$ with the following properties:
\begin{enumerate}
    \item Input: a function $f\colon \mathbb{F}_q \to\mathbb{F}_{q'}$.
    \item Completeness: if $f\in \RS_{q}[d, \Ff_q]$, then the verifier accepts with probability $1$.
     \item Initial State: the initial state is doomed if $f$ is 
     $\eps$-far from $\RS_{q}[d, \Ff_q]$.
    \item Round-by-round soundness: $2^{-\lambda}$.
    \item Round Complexity: $O(\log \log d)$.
    \item Query Complexity: $O(\lambda) + O\left( \frac{\lambda^2}{\log^2(d)}\cdot  \log \log(d)\right)$.
    \item Alphabet Size: $2^{\lambda}\cdot \poly(d)$.
    \item Length: $O(d)$.
\end{enumerate}
\end{thm}

The IOPP of \cref{thm: iopp const rs} is constructed in a similar manner to the IOPP for small rate Reed-Solomon Codes. Starting with the input function $f: \Ff_q \to \Ff_{q'}$, we again embed $f$ into a Reed-Muller codeword, but this time we use three variables. Here, we use the fact that if $f$ is a univariate function of degree at most $d$, then there is a function
$Q$ of individual degrees $d^{1/3}$ such that $f(x) = Q(x^{d^{2/3}}, x^{d^{1/3}}, x)$. Thus, the prover will supply the verifier with the truth table of $Q: \Fqe^3 \to \Ff_{q'}$, where $|\Fqe| = \Theta\left(d^{1/3}\right)$ and it remains for the verifier to check that 
\begin{enumerate}
    \item $Q$ has individual degrees at most $d^{1/3}$,
    \item $Q$ and $f$ agree.
\end{enumerate}
The first item can be accomplished using the constant rate Reed-Muller IOPP discussed in \cref{sec: const rate rm intro} (with a slight modification to handle individual degree versus total degree). The second item requires more work however, and in particular cannot be done in the same way as in the small rate case. The difference here is that we the agreement between $Q$ and $f$ can be constant rather than at most $1/q^{c}$, for some $c > 0$. Indeed, the verifier must be able to check that $Q$ and $f$ agree, even when the agreement is as large as $\Omega(d/q)$. In the small rate case, $\Omega(d/q) = q^{-\Omega(1)}$, so the verifier can differentiate $Q$ from $f$ with soundness error at most $2^{-\lambda}$ using $O(\lambda/\log(q))$ random queries to their respective valuations, and when handling the side conditions, the verifier needs to look at an entire subcube containing these points. In the dimension $m$ case, this results in $O(\lambda^m/\log^m(q))$ queries, which is okay. On the other hand, in the constant rate case, $\Omega(d/q) = O(1)$, so the verifier
needs $\Omega(\lambda)$ queries for the same task, leading to $\Omega(\lambda^m)$ queries when handling the side conditions. We wish to avoid dependence on $\Omega(\lambda^m)$ however, and this requires one additional idea.

While there is no way around the $\Omega(\lambda)$ query barrier needed to distinguish two $Q$ and $f$ in the $\Omega(1)$-agreement case, it is possible avoid the blowup in queries when handling the side conditions. Without going into too much detail, this is achieved by having the verifier make all of their queries along an affine line. The key fact we use is that the points of interest are all of the form $S = \{(\alpha^{d^{2/3}}, \alpha^{d^{1/3}}, \alpha) \; | \; \alpha \in \Ff_q \}$, and we may assume that $d^{2/3}$ and $d^{1/3}$ are powers of $2$, and furthermore, $\Ff_{q}$ has a subfield of size $d^{1/3}$, which we call $\Ff_{d^{1/3}}$. Then, choosing an affine line going through $s = (\alpha^{d^{2/3}}, \alpha^{d^{1/3}}, \alpha)$ in direction $t = (\beta^{d^{2/3}}, \beta^{d^{1/3}}, \beta)$, notice that this line $L(s, t) = \{s + \gamma \cdot t \; | \; \gamma \in \Ff_q \}$ contains many points of $S$. Indeed, as long as $\gamma \in \Ff_{d^{1/3}}$, we get a point of the form
\[
s + \gamma \cdot t = (\alpha^{d^{2/3}} + \gamma \beta^{d^{2/3}}, \alpha^{d^{1/3}} + \gamma \beta^{d^{1/3}}, \alpha + \gamma \beta) = ((\alpha + \gamma\beta)^{d^{2/3}}, (\alpha + \gamma\beta)^{d^{1/3}}, \alpha + \gamma \beta) \in S.
\]
Here we are using the following two facts: 1) $d^{1/3}$ and $d^{2/3}$ are powers of $2$ and in characteristic $2$-fields we have $(\alpha + \beta)^2 = \alpha^2 + \beta^2$; 2) since $\gamma \in \Ff_{d^{1/3}}$, we have $\gamma^{d^{2/3}} = \gamma^{d^{1/3}} = \gamma$. Given the above, it turns out that it is sufficient for the verifier to choose $\lambda/\log(q)$ lines $L(s, t)$ as above, and then compare $Q$ and $f$ at a $\log(q)$ of the points in $S \cap L(s,t)$ on each chosen line. For the full analysis, we analyze the spectral gap of the inclusion graph between the lines $\{L(s,t) \; | \; s,t \in S \}$ and $S$, and show a spectral sampling lemma in a similar spirit to \cref{lm: spectral sampling}.
 
\paragraph{Applications of \cref{thm: iopp const rs}:} We state \cref{thm: iopp const rs} because it may lead to immediate improvements in some of the numerous applications which make blackbox use of FRI. As one example, many polynomial commitment operate by using FRI to test a constant rate Reed-Solomon codeword with one side condition. For this task, one can use any IOPP for constant rate Reed-Solomon codes with the completeness and soundness guarantees of \cref{thm: iopp const rs}. As \cref{thm: iopp const rs} achieves these guarantees in a more efficient manner, so it is intriguing to see if it would lead to any improved polynomial commitment schemes.

\subsubsection{Concrete Efficiency}\label{sec:concrete_eff}

We believe that there is potential towards obtaining concretely efficient IOPs via our construction. While in the discussion above we focused our attention on their theoretical performance, our IOPs also have good prove and verifier complexities, comparable to those in FRI and STIR. In addition, we also obtain sizable improvements in the round and query complexities over FRI and STIR, which gives our IOPs the prospect of even surpassing FRI and STIR in certain aspects. We elaborate on these characteristics below. 

Currently, the main impractical aspect of our construction is the implicit constant in the length of the proof, which is moderately large. This constant may be improved in several ways though. For example, if one could work with bi-variate Reed-Muller codes instead of tri-variate ones, one would immediately get significant saving in the implicit constant. Alternatively, and as we explain below, one could also sidestep the proof length issue by combining our IOPs with constantly many rounds of either the FRI or STIR protocol.

\paragraph{Prover and Verifier Complexity:} Our Verifier complexity is $\poly \log n$, similar to FRI and STIR. Of more interest is the prover complexity which we now discuss.

The dominating term in the prover complexity is the time required to encode the length $n$ NP witness in a trivariate Reed-Muller code. As our encoding has constant rate, the complexity of this step is $O(n\log n)$ via a Fast-Fourier-Transform (similar to how the encoding is done in FRI \cite{BBCGGH17}). The concerning part, as of now, is the constant factor blow up that we suffer. In our current construction, starting with a length $n$ witness, the prover must provide the evaluation of a trivariate, degree $6n^{1/3}$ function over $\Ff_q^3$, where $q = C \cdot 6n^{1/3}$ for some small constant $C$. Taking $C = 2$, for example, this setting leads to a blow up of $1728\times$, and thus a constant factor of at least $1728$ in the time complexity of the Fast-Fourier-Transform. One could hope to improve this blowup by starting with a bivariate Reed-Muller encoding though. This encoding would allow for field size $q = C \cdot 4n^{1/2}$ and taking $C = 2$ now yields a blow up of $64x$. The reason that our current IOP requires $m \geq 3$ dimensions in the initial encoding is that we rely on the spectral gap of the affine planes versus point graph in $\Ff_q^m$, and we also require that going from $m$-dimensions to $2$-dimensions gives a reduction in size.

\paragraph{Combining our Folding Iteration with FRI or STIR:} A route towards concretely efficient IOPs which seems particularly promising is to first perform a few iterations of FRI or STIR, and then apply our IOP. Specifically, note that each iteration of FRI or STIR consists of a ``folding'' step with factor $k = O(1)$ and reduces the problem from size $n$ to size $n/k$. Thus, one could perform FRI or STIR for $t = O(1)$ rounds, with folding factor $k = O(1)$, to reduce the problem size from $n$ to $n / k^t$. Applying our IOP at this stage, the prover can now encode a length $n/k^t$ witness instead of a length $n$ witness. This leads to a blowup in prover complexity of $\frac{1728}{k^t} \times$ as opposed to $1728\times$. Note that since FRI or STIR is only used $O(1)$ times, we still gain improvements in round complexity and query complexity.

\paragraph{SNARKS:} Typically, concretely efficient IOPs are used in practice to construct non-interactive proofs such as succinct non-interactive arguments (SNARGs) or succinct non-interactive arguments of knowledge (SNARKs). This process generally requires compiling the IOP using some procedure such as the BCS transformation of \cite{BCS}, which compiles an IOP into a SNARK. If the IOP has round-by-round soundness $2^{-\lambda}$, then the BCS transformation guarantees that the SNARK has $\lambda$ bits of security. A key parameter in SNARKS is the argument size, and one wants this to be as small as possible. 

\paragraph{Argument Size:} The main parameters of an IOP that influence the size of its compiled SNARK are the query complexity {\sf q}, alphabet size {\sf A}, and length {\sf L}. Compiling an IOP with these parameters leads to a SNARK with size proportional to
\[
{\sf q} \cdot \left( \log({\sf A}) + \log({\sf L}) \right).
\]
See, for example, \cite[Chapter 25]{ChiesaYogev2024}. Compared to FRI and STIR, our protocol has the same asymptotic log-alphabet size and length, while improving on the query complexity asymptotically (for some reasonable settings of the security parameter $\lambda$). Thus, we obtain significant improvements on the size of the compiled SNARK for sufficiently large input size $n$. For concretely efficient implementations however, one typically works with, say, $n \in [2^{20}, 2^{30}]$, meaning the constant factors in our query complexity become important. We leave it to future work to tighten these factors. Specifically, tightening the constant factors in our query complexity requires improved results for the soundness of the line versus point test and the minimal agreement threshold for the Reed-Muller proximity generator theorems. Currently, both of these quantities are $(d/q)^{c}$ for some small but absolute $c > 0$, but it is plausible that one can obtain $c$ arbitrarily close to $1$. The assumption that $c$ here is close to $1$ is comparable to the Reed-Solomon list-decoding conjecture assumed by FRI and STIR to achieve small constant factor improvements, see \cite{BGKS, bciks}.

\section{Preliminaries}
In this section we set up a few basic notions and tools that are necessary in our arguments.

\subsection{The Interactive Oracle Proofs Model}
An interactive oracle proof (IOP) for a language $\mc{L}$ is an interactive protocol between a prover and a verifier. Initially, both parties hold a joint input $\mathrm{x}$, and the goal of the prover is to convince the verifier that $\mathrm{x}\in\mc{L}$. During the $i$th round of interaction, the prover first sends a message ${\sf p}_i$, after which the verifier sends a message ${\sf v}_i$. The prover's messages ${\sf p}_i$ are strings over some alphabet $\Sigma$, and one should think of these strings as (large) oracles which the verifier only has query access to. The verifier may query the input $\mathrm{x}$ as well as any ${\sf p}_j$ for $j \leq i$ during the $i$th round of interaction. Their message ${\sf v}_i$ can depend on the results of these queries as well as their internal randomness. One should think of the prover's next message ${\sf p}_{i+1}$ as a response to the verifier's message ${\sf v}_i$. At the end of the protocol, the verifier decides to accept or reject. 

\paragraph{Parameters:} the following parameters of IOPs will be of interest to us:
\begin{itemize}
    \item \textbf{Completeness:} we say that an IOP has completeness if for any input $\mathrm{x} \in \mc{L}$, the honest prover makes the verifier accept with probability $1$.\footnote{More generally, one can say an IOP has completeness $c$ if the honest prover can always convince the verifier to accept a valid input $\mathrm{x} \in \mc{L}$ with probability at least $c$. The $c=1$ case is then referred to as perfect completeness, and we only use it in this paper.}
    \item \textbf{Soundness:} we say that an IOP has soundness $s$ if for any input $\mathrm{x} \notin \mc{L}$ and any prover strategy, the verifier accepts with probability at most $s$.
    \item \textbf{Round complexity:} the maximum number of rounds of interaction in the IOP.
    \item \textbf{Input query complexity:} the total number of queries that the verifier makes to the input $\mathrm{x}$.
    \item \textbf{Proof query complexity:} the total number of queries that the verifier makes to the prover's messages.
    \item \textbf{Query complexity:} the sum of the input query complexity and proof query complexity.
    \item \textbf{Alphabet size:} $|\Sigma|$, or the size of the alphabet which the prover uses to write their messages.
    \item \textbf{Length:} $\sum_i |\p_i|$, or the total length of the prover's messages.
\end{itemize}

In addition to these parameters above, we will also be interested in a stronger type of soundness called \emph{round-by-round soundness}, which also requires the IOP to have a \emph{state function}. The state function, denoted by $\st$, assigns each partial transcript $\{\mathrm{x}, \p_1, \vs_1, \ldots, \p_i, \vs_i \}$ a state, 
\[
\st(\{\mathrm{x}, \p_1, \vs_1, \ldots, \p_i, \vs_i \}) \in \{0, 1\}.
\]
The partial transcript may consist of just the input $\{\mathrm{x}\}$, which corresponds to the first round of the protocol. If the state is $0$, then we say that the state is \emph{doomed}, and if the state is $1$ we say that the state is \emph{not doomed}. Given a state function $\st$, round-by-round soundness is defined as follows.

\begin{definition}
    We say that round $i$ has soundness $s$ if for any $\{\mathrm{x}, \p_1, \vs_1, \ldots, \p_{i-1}, \vs_{i-1}\}$ which has state $\st(\{(\mathrm{x}, \p_1, \vs_1, \ldots, \p_{i-1}, \vs_{i-1})\}) = 0$ and any prover message $\p_{i}$ during round $i$, 
    \[
    \Pr_{\vs_{i}}[\st(\{(\mathrm{x}, \p_1, \vs_1, \ldots, \p_{i}, \vs_{i})\}) = 1] \leq s.
    \]
\end{definition}
In words, the definition for soundness $s$ in round $i$ says that if the state is ``doomed'' prior to round $i$, then for any message $\p_i$ that the prover may send, the state after round $i$ will change to ``not doomed'' with probability at most $s$. The probability here is over the verifier's randomness when generating their message $\vs_i$.

\begin{definition}
An IOP for language $\mc{L}$ has round-by-round soundness $s$ if there is a state function such that the following holds:
\begin{itemize}
    \item The initial state is doomed if and only if $\mathrm{x} \notin \mc{L}$.
    \item The final transcript is doomed if and only if the verifier rejects.
    \item Every round has soundness $s$.
\end{itemize}
\end{definition}

Going forward, we refer to the first introduced notion of soundness as either standard soundness or simply as soundness when it is clear from context. 
Morally speaking, the following lemma shows round-by-round soundness is a stronger notion than 
standard soundness, and hence we will often strive to achieve it whenever possible.

\begin{lemma} \label{lm: standard vs rbr soundness}
    If a $k$-round IOP has round-by-round soundness $s$, then it has standard soundness $ks$.
\end{lemma}
\begin{proof}
  Fix the state function with which the IOP has round-by-round soundness $s$. If $x \notin \mc{L}$, then the initial state is doomed. Union bounding over the $k$-rounds, the final transcript is not doomed with probability at most $ks$, and hence the verifier accepts with at most this probability.
\end{proof}

\subsection{Spectral Sampling Lemma}

 In this section, we state and prove a version of the expander mixing lemma for real valued functions over the vertices of a bipartite graph. Versions of this lemma have appeared many times in different contexts 
We include a quick proof for the sake of completeness. 

Let $\mc{G} = \mc{G}(A, B, E)$ be a (weighted) graph. Let $n_1 = |A|$ and $n_2 = |B|$. Let $L_2(A)$ be the set of real-valued functions over $A$ endowed with the natural inner product with respect to the stationary distribution, and define $L_2(B)$ similarly. We let $\mathcal{T}: L_2(A) \to L_2(B)$ be the adjacency operator of $\mc{G}$, defined as follows. For a function $F \in L_2(A)$, $\mathcal{T}F \in L_2(B)$ is given by
\[
\mathcal{T}(b) = \E_{a \sim b}[F(v)],
\]
where $a \in A$ is a random neighbor of $b$ sampled according to the edge weights of $\mc{G}$. We denote its $2$-norm by $\norm{F} = \sqrt{\E_{a}[F(a)^2]}$. It is well known that if $\mc{T}$ has second singular value $\lambda$, then letting $F' = F - \E_{a \in A}[F(a)]$, we have
\begin{equation} \label{eq: singular value def}
    \norm{\mc{T} F'}_2 \leq \lambda \norm{F'}_2
\end{equation}

We will show the following lemma.

\begin{lemma} \label{lm: spectral sampling}
    Let $\mathcal{T}$ be a biregular adjacency operator with second singular value $\lambda$. Let $S \subseteq A$ have fractional size $|S|/|A| = \eps$. Then for any $G \in L_2(B)$ satisfying $0 \leq G(b) \leq 1$ for all $b \in B$, we have
    \[
    \left|\E_{a \in S, b \sim a}[G(b)] - \mu(G)\right| \leq \frac{\lambda}{\sqrt{\eps}}.
    \]
    where $\mu(G) = \E_{b \in B}[G(b)]$.
\end{lemma}
\begin{proof}[Proof of Lemma~\ref{lm: spectral sampling}]
    Let $F \in L_2(A)$ be the indicator function for $S$, and note that $\E_{a \in A}[F(a)] = \eps$. We can express 
    \[
    \E_{a \in S, b \sim a}[G(b)] = \frac{\inner{\mc{T}F}{G}}{\eps}.
    \]
    Now write $F' = F - \E_a[F(a)]$ and $G' = G - \E_b[G(b)]$. Then,
    \[
    \langle \mc{T}F, G \rangle =  \eps \mu(G) + \langle \mc{T} F'
    , G' \rangle,
    \]
    and we have
    \[
      \left|\E_{a \in S, b \sim a}[G(b)] - \mu(G)\right| = \left|\frac{\langle \mc{T}F', G' \rangle}{\eps} \right|.
    \]
    To conclude, we bound,
    \[
    |\langle \mc{T}F', G' \rangle| \leq \norm{\mc{T}F'}_2 \cdot \norm{G'}_2 \leq \lambda \norm{F'}_2 \norm{G'}_2 \leq \lambda \sqrt{\eps}.
    \]
    In the first transition we are using the Cauchy-Schwarz inequality, in the second transition~\eqref{eq: singular value def}, and in the third transition we are using the bounds $||F'||^2_2 \leq \eps$ and $||G'||^2_2 \leq 1$.
\end{proof}

We will particularly be interested in the case where $\mc{G}$ is a bipartite inclusion graph between subspaces of different dimensions. The following second singular values are well known 
(see for example~\cite[Section 9]{BCN} or~\cite[Section 9]{godsil2015erdos}).

\begin{lemma}\label{lem:spectral_val_bd}
    Fix an ambient space $\Ff_q^m$ for $m \geq 2$. Let $\mc{L}$ and $\mc{P}$ denote the set of affine lines and affine planes in $\Ff_q^m$ respectively. Let $\mc{G}(A, B)$ be the bipartite inclusion graph where $(a,b)$ is an edge if $a \supseteq b$. We have the following second singular values:
    \begin{itemize}
        \item $G = \mc{G}(\mc{L}, \Ff_q^m)$, $\lambda(G) \leq \frac{1}{\sqrt{q}}$,
        \item $G = \mc{G}(\mc{P}, \Ff_q^m)$, $\lambda(G) \leq \frac{1}{q}$,
        \item $G = \mc{G}(\mc{P}, \mc{L})$, $\lambda(G) \leq \frac{1}{\sqrt{q}}( 1 + o(1/\sqrt{q}))$.
    \end{itemize}
\end{lemma}

\section{Finite Fields, Polynomials and Agreement}
\subsection{Basic Notation}
Throughout the paper we use $\Ff_q$ to denote a finite field of size $q$. We let $\Ff_{q}[x]$ denote the set of univariate polynomials over $\Ff_{q}$ and more generally let $\Ff_{q}[x_1, \ldots, x_m]$ denote the set of $m$-variate polynomials over the variables $x_1, \ldots, x_m$. It is well known that (multivariate) polynomials, $\Ff_{q}[x_1,\ldots, x_m]$, are in one-to-one correspondence with functions $\Ff^m_{q} \to \Ff_{q}$. That is, for each function $\Ff_{q}^m \to \Ff_{q}$, there is exactly one polynomial in $\Ff_q[x_1,\ldots, x_m]$ which evaluates to this function over $\Ff_q^m$. To help distinguish between formal polynomials and functions, we will use hat notation to refer to polynomials, $\wh{f} \in \Ff_q[x_1,\ldots, x_m]$, and non-hat notation to refer to functions $f: L \to \Ff_q$. 

Oftentimes, we will deal with functions $f$ where $L \subsetneq \Ff^m_q$, and emphasizing this feature is the main reason for us making this notational distinction. Sometimes we will want to refer to the evaluation of a polynomial $\wh{f} \in \Ff_{q}[x_1,\ldots, x_m]$ over some subdomain $H \subseteq \Ff_{q}^m$. In this case, we use $\wh{f}|_H$ to refer to the function $f: H \to \Ff_{q}$ given by $f(x) = \wh{f}(x)$ for $x \in H$. We will call $\wh{f}|_H$ the evaluation of $\wh{f}$ over $H$ or the restriction of $\wh{f}$ to $H$.

\subsection{The Reed-Solomon and Reed-Muller Codes}
An important quantity related to a polynomial is its degree. Given a polynomial 
\[
f(x_1,\ldots, x_m) = \sum_{e = (e_1, \ldots, e_m) \in \{0, \ldots, q-1\}^m}C_e \cdot x_1^{e_1} \cdots x_m^{e_m} \in \Ff_{q}[x_1,\ldots, x_m],
\]
where each $C_e$ is a coefficient from $\Ff_q$, we define its degree as the maximum total degree of a monomial appearing in its expansion with nonzero coefficient, $\deg(f) = \max_{e: C_e \neq 0} e_1 + \cdots + e_m$.
Slightly abusing terminology, we will often say that $f$ has degree $d$ as long as its degree is at most $d$. We use the following notation to refer to the set of polynomials of degree (at most) $d$:
\[
\Ff_{q}^{\leq d}[x] = \{f \in \Ff_q[x] \; | \; \deg(f) \leq d \}, \quad \quad \Ff_{q}^{\leq d}[x_1,\ldots, x_m] = \{f \in \Ff_q[x_1,\ldots, x_m] \; | \; \deg(f) \leq d \}.
\]
\paragraph{Individual degrees:} in the case of multivariate polynomials, i.e.~the case that $m\geq 2$, we associate another notion of degree, which we refer to as individual degree. Given $d_1, \ldots, d_m \in \{0,\ldots, q-1\}$, we say that $f$ has individual degree $(d_1, \ldots, d_m)$ if $C_e = 0$ for all $e = (e_1, \ldots, e_m) \in \{0, \ldots, q-1\}^m$ such that $e_i > d_i$ for some $i \in [m]$. If $f$ has individual degree $(d_1, \ldots, d_m)$, we will write $\deg(f) \leq (d_1, \ldots, d_m)$. Note that even though we are using the same notation as (standard) degree, the presence of a tuple of degrees on the right hand side of the inequality will make it clear that we are referring to individual degree and not (standard) degree. Let $\Ff_q^{\leq (d_1, \ldots, d_m)}[x_1, \ldots, x_m]$ denote the set of polynomials of individual degree at most $(d_1, \ldots, d_m)$.

\subsubsection{Low-degree Extension}
Suppose that $f: L \to \Ff_{q'}$ is a univariate function where $L\subseteq\mathbb{F}_q$. An important polynomial in $\Ff_{q'}[x]$ associated with $f$ is its low degree extension. When a function $f$ is already defined, we will use $\wh{f} \in \Ff_{q'}[x]$ to refer to its low degree extension. Note that this is in line with our previous convention of using hat notation to refer to polynomials. The low degree extension is defined as follows:
\begin{definition}
    Given a function $f: L \to \Ff_{q'}$, its low degree extension $\wh{f}$ is the unique polynomial in $\Ff_q^{\leq |L|-1}[x]$ which satisfies $\wh{f}(x) = f(x)$ for all $x \in L$.
\end{definition}
We also consider low-degree extensions of multivariate polynomials, and in this case use individual degree. We similarly use the hat notation on a function to refer to the multivariate polynomial which is its low degree extension.
\begin{definition}
    Given a function $f: L_1\times \cdots \times L_m \to \Ff_{q'}$, its low degree extension $\wh{f}$ is the unique polynomial in $\Ff_q^{\leq (|L_1|-1, \ldots, |L_m|-1)}[x_1, \ldots, x_m]$ which satisfies $\wh{f}(x) = f(x)$ for all $x \in L_1 \times \cdots \times L_m$.
\end{definition}
With low degree extensions defined, we define the degree of a function $f: L_1 \times \cdots \times L_m \to \Ff_{q'}$ as the degrees of their low degree extensions, i.e., $\deg(f) = \deg(\wh{f})$. We say that a function $f$ has degree at most $d$ if $\deg(\wh{f})\leq d$. Likewise, for individual degree, we say that a function $f$ has individual degree at most $(d_1, \ldots, d_m)$ if $\wh{f}$ has individual degree at most $(d_1, \ldots, d_m)$. 

\subsubsection{The Definition of Reed-Solomon and Reed-Muller codes}
Three central objects in this paper are the Reed-Solomon code, the Reed-Muller code, and the individual degree Reed-Muller code, and we give the formal definition here. At a high level, these codes corresponds to univariate functions of low degree, multi-variate functions of low degree, and multi-variate functions of low individual degree.

\begin{definition}
    Given a degree parameter $d$, a finite field $\Ff_{q'}$, and an evaluation domain $U \subseteq \Ff_{q'}$, the degree $d$ Reed-Solomon code over $L$ is the following set of functions
    \[
    \RS_{q'}[d, U] = \{f: U \to \Ff_{q'} \; | \; \deg(f) \leq d\}.
    \]
\end{definition}

\begin{definition}
    Given a degree parameter $d$, a finite field $\Ff_{q'}$, and an evaluation domain $U \subseteq \Ff^m_{q'}$, the individual degree $(d_1, \ldots, d_m)$ Reed-Muller code over $U$ is the following set of functions
    \[
    \RM_{q'}[d, U] = \{f: U \to \Ff_{q'} \; | \; \deg(f) \leq d\}.
    \]
\end{definition}

\begin{definition}
    Given an individual degree parameter $(d_1,\ldots, d_m)$, a finite field $\Ff_{q'}$, and an evaluation domain $U \subseteq \Ff^m_{q'}$, the degree $d$ Reed-Muller code over $L$ is the following set of functions
    \[
    \iRM_{q'}[d, U] = \{f: U \to \Ff_{q'} \; | \; \text{$f$ has individual degree at most $(d_1, \ldots, d_m)$} \}.
    \]
\end{definition}


\subsection{Useful Features of the Reed-Solomon and Reed-Muller Codes}
In this section we describe some basic features of Reed-Solomon and 
Reed-Muller codes, as well as relations between them, that will be useful
for us. We begin with a 
quick motivating discussion to highlight the sort of properties our application requires from these codes.
\subsubsection{IOPs and Codes}
The primary goal of this paper is to construct efficient IOPs for the 
class NP, and more specifically for the NP-complete problem $\rics$. 
Recall that an instance of $\rics$ is a collection of rank $1$ quadratic equations, and our goal is to verify whether some assignment to $n$-variables, $x \in \Ff_q^{n}$, is a valid solution to it. 
A typical IOP for this task proceeds as follows:
\begin{enumerate}
    \item The prover provides a string $w$, which is supposedly the encoding, ${\sf Enc}(x)$, of $x$ using some error correcting code.
    \item The verifier checks, using $w$, whether the (supposedly) encoded assignment is an actual solution.
    \item Test: the prover and verifier run an auxiliary IOP to convince the verifier that $w$ is indeed the encoding of some assignment $x$.
\end{enumerate}
Since we wish to construct an IOP of linear length and small query complexity, the error correcting code we use in step 1 must have a few important properties: 
\begin{enumerate}
    \item {\bf Distance}: we would like the code to have good \emph{distance}. While we will define distance more formally later, this requirement informally says that any two codewords should differ on many coordinates. Good distance is needed to achieve item 2 above in a query efficient manner. At a high level, the reason is if the prover gives an encoding of a false witness, then the only hope for the verifier to detect this without reading all of ${\sf Enc}(x)$ is if ${\sf Enc}(x)$ looks very different than the encoding of every valid witness. Thus, good distance is a prerequisite towards the second item above.
    \item {\bf List decodability:} we would like that for any supposed encoding $w$ there are boundedly many legitimate codewords that have non-trivial correlation with $w$. This property is useful in establishing item 2 above: consider a scenario where the prover provides an encoding which is very far from any particular legitimate codeword, but still has non-trivial agreement with a large number of legitimate codewords. Then it may be the case that none of these codewords are encodings of valid witnesses, but collectively, they manage to fool the verifier with sufficiently large probability. Bounded size list decoding helps us avoid this issue.   
    \item {\bf Rate:} since we want the IOP to be of linear length, it means that the size of the supposed encoding sent in the first step has to be of linear length. Thus, we must use (at least in the beginning of the protocol) an error correcting codes with constant rate. 
    \item {\bf Local testability:} to facilitate item 3 above in a query efficient manner, we would like the code to have an efficient local test. That is, the verifier should be able to test whether ${\sf Enc}(x)$ is a true encoding of $x$ while only reading a small number of its entries.
\end{enumerate}
With these considerations, the Reed-Solomon code and Reed-Muller code are natural candidates. As we will see, the Reed-Solomon code has the desired distance, list-decoding and length properties, while the Reed-Muller code has the desired distance, list-decoding and local testability properties. Thus, each code is lacking one of the four desired properties, so in order to achieve all four simultaneously, we must figure out a way to combine the two. To this end, we use an embedding of a Reed-Solomon codeword in a Reed-Muller codeword, so that we can use the Reed-Muller local test to locally test the Reed-Solomon code. This approach also appears in the IOPPs FRI \cite{bciks} and STIR \cite{acfy}, as well as the IOPP of \cite{acy}, however we use different local testers for Reed-Muller codes.

In the remainder of this section, we discuss the relevant properties of the Reed-Solomon and Reed-Muller codes.

\subsubsection{Distance}
For two functions $f, g: U \to \Ff_{q'}$, denote the agreement 
and distance between $f$ and $g$ by
\[
\agr(f, g) = \Pr_{x \in U}[f(x) = g(x)],
\qquad\qquad
\delta(f, g) = \Pr_{x \in U}[f(x) \neq g(x)]
\]
respectively.
Note that $\agr$ is the fractional size of the set of inputs on which $f$ and $g$ agree while $\delta$ is the fractional size of the set of inputs on which $f$ and $g$ differ, so we always have $\agr(f,g) + \delta(f,g) = 1$. Given a family of functions $\mc{C} \subseteq \{U \to \Ff_{q'} \}$ we similarly write,
\[
\agr(f, \mc{C}) = \max_{g \in \mc{C}} \agr(f, g), \qquad \qquad \delta(f, \mc{C}) = \min_{g \in \mc{C}} \delta(f, g).
\]
It is again always the case that $\agr(f, \mc{C}) + \delta(f, \mc{C}) = 1$.

An important property of the Reed-Solomon and Reed-Muller codes is that any two functions from the code only agree on a small fraction of their inputs. In coding theoretic terms, this says that the Reed-Solomon and Reed-Muller codes have good distance. This is an easy consequence of the Schwartz-Zippel lemma.

\begin{lemma} \label{lm: schwartz-zippel for rs}
    For any two distinct functions $f, g \in \RS_{q'}[d, L]$, we have 
    $\Pr_{x \in L}[f(x) = g(x)] \leq \frac{d}{|L|}$.
\end{lemma}


\begin{lemma}\label{lm: schwartz-zippel for rm}
    For any two distinct functions $f, g \in \RM_{q'}[d, U]$ where $U = L \times \cdots \times L$, we have
    \[
    \Pr_{x \in U}[f(x) = g(x)] \leq \frac{d}{|L|}.
    \]
\end{lemma}


\subsubsection{List Decoding}
In this section we give a generic list decoding bound due to Blinovsky~\cite{blinovsky1986bounds}, upper-bounding the number of codewords with non-trivial agreement with a given word $f$. 

\begin{thm} \label{thm: list decoding}
    Let $\mathcal{F} \subseteq \{\mc{L} \to \Ff \}$ be a set of functions such that any pair of functions in $\mc{F}$ agree on at most $\delta$-fraction of their coordinates. Then for any $C > \sqrt{\delta}$ and any $F: \Lc \to \Ff$, there are at most $\frac{C}{C^2-\delta}$ functions in $\mc{F}$ that agree with $F$ on at least $C$-fraction of its coordinates. 
\end{thm}
\begin{proof}
Let $f_1,\ldots, f_m$ be all functions agreeing with $F$ on at least $C$ fraction of its entries. Let $N = |\Lc|$. Now consider the graph where the left side consists of $f_1, \ldots, f_m$ and the right side consists of points in $\mc{L}$. Add an edge between $f_i$ and $x \in \mc{L}$ if $f_i(x) = F(x)$. By assumption, the degree of every $f_i$ is at least $C\cdot N$. Delete edges so that the degree of every $f_i$ is exactly $C \cdot N$. 
    
We count the number of triples $f_i, f_j, x$ such that $(f_i, x)$ and $(f_j, x)$ are both edges. For every $x \in \Lc$, let $d(x)$ denote the degree of $x$. Then, the number of such triples is 
\[
\sum_{x \in \Lc} \binom{d(x)}{2} = \sum_{x \in \Lc} \frac{d(x)^2}{2} - \frac{d(x)}{2} \geq \frac{(C N m)^2}{2N} - \frac{C  m  N}{2} = \frac{C^2\ m^2N}{2} - \frac{CmN}{2},
\]
where we used Cauchy-Schwarz inequality to say that 
$\sum_{x}d(x)^2\geq (\sum\limits_{x} d(x))^2/N$. 
On the other hand, for every pair of distinct $f_i, f_j$, we have $f_i(x) = f_j(x)$ for at most $\delta N$ points $x \in \Lc$. Thus, the number of triples is at most $\binom{m}{2}\delta N \leq \frac{m^2\delta N}{2}$.
Combining the two bounds and re-arranging gives the result.
\end{proof}

\subsubsection{Local Testability}\label{sec:local_test}
As mentioned earlier, our arguments require locally testable codes, and
towards that end we use Reed-Muller codes (we recall that Reed-Solomon
codes of degree $d$ require $\Omega(d)$ queries) and the line versus point 
test. In particular, we use
\cref{thm: line vs point} which is a slight modification of the main theorem in to~\cite{hkss}. An important quantity is the exponent in the soundness of the line versus point test of \cref{thm: line vs point}. Throughout this paper we denote this value by $\tl$ and remark that
\begin{equation*}
    \tl < \frac{1}{20}.
\end{equation*}
Going forward, we will make use of this inequality without reference.

\begin{thm}\label{thm: line vs point}
There exists a constant $\tl \in (0,1/20)$ such that the following holds.
Fix a degree parameter $d$ and finite fields $\Ff_q$ and $\Ff_{q'}$ such that $q > d$ and $\Ff_q$ is a subfield of $\Ff_{q'}$. If $f: \Ff_q^m \to \Ff_{q'}$ satisfies
    \[
     \E_{L \subseteq \Ff_q^m}[\agr(f|_L, \RS_{q'}[d, \Ff_q])] = \eps \geq \left(\frac{d}{q} \right)^{\tl},
    \]
Then there is a degree $d$ function $F: \Ff_q^m \to \Ff_{q'}$ such that 
$\agr(f, F) \geq \eps - \left(\frac{d}{q} \right)^{\tl}$.
\end{thm}
\begin{proof}
   The proof is deferred to \cref{sec: line vs point}.
\end{proof}
The difference between \cref{thm: line vs point} stated above and the main theorem in \cite{hkss} is that here we consider functions from a subfield to a larger field, i.e.\ $\Ff_q' \supseteq \Ff_q$, while in \cite{hkss}, the main theorem is only shown in the $\Ff_{q} = \Ff_{q'}$ case. While we are not 
aware of a black-box reduction from the case in~\cref{thm: line vs point} to the case considered in~\cite{hkss}, 
their argument works essentially as is, and we include it for completeness.
Finally, we remark that prior work in \cite{AroraSafra} also shows a variant of the result in \cite{hkss}, but there the results only hold when $q = \poly(d)$, whereas \cite{hkss} holds for $q = O(d)$. As we care about constant rate, we need \cref{thm: line vs point} in the $q = O(d)$ regime, and therefore, we must adapt the result of \cite{hkss}.

\skipi

For our purposes we also need two implications \cref{thm: line vs point}.
The first of which gives an upper bound on the probability that $f$ which is far from degree $d$, is close to a degree $d$ function when restricted to a random line. 
\begin{thm} \label{thm: line vs point strong}
   Let $\eps \geq (d/q)^{\tl/2}$ be an agreement parameter. Suppose $f: \Ff_q^2 \to \Fqp$ satisfies
    \[
    \agr(f, \RS_{q'}[d, \Ff_q^2]) \leq \eps.
    \]
    Then,
    \[
    \Pr_{L \subseteq \Fqm}[\agr(f|_L, \RS_{q'}[d, \Ff_q]) \geq 1.01\eps] \leq \eps.
    \]
\end{thm}
\begin{proof}
    The proof is deferred to \cref{sec: proof of lvp strong}.
\end{proof}

The next result gives a bound which is stronger than the bound in~\cref{thm: line vs point strong},  but considers restrictions to randomly
chosen affine planes.
\begin{thm} \label{thm: plane vs point strong}
    Suppose $f: \Ff_q^3 \to \Fqp$ satisfies $\agr(f, \RM_{q'}[d, \Fqm]) \geq \eps \geq \left(\frac{d}{q} \right)^{\tl}$.
    Then
    \[
    \Pr_{P \subseteq \Fqm}\left[\agr\left(f|_P, \RM_{q'}[d, \Fqt]\right) \geq 1.2\eps\right] \leq \frac{100}{\eps^2 q}, 
    \]    
    where $P$ is a uniformly chosen affine plane in $\mathbb{F}_q^m$.
\end{thm}
\begin{proof}
    The proof is deferred to \cref{sec: proof of pvp strong}.
\end{proof}

\subsection{Embedding Reed-Solomon in Reed-Muller}

As we have seen, the Reed-Muller code has an efficient local test that can be carried out by reading only a sublinear part of its input (around $q^{1/m}$ queries). On the other hand, as discussed, one cannot hope for such an efficient local test for the Reed-Solomon code. To get around this issue, we show that one can embed a univariate function in a bivariate function by suitably ``splitting'' a single variable into two. The same splitting can be done into any $m \geq 2$ variables, but for simplicity we will only work with the bivariate case. 

\begin{lemma}\label{lem:split_rs_rm}
   Suppose $\wh{f} \in \Ff_{q'}^{\leq d}[x]$,
   let $k$ be an integer and set $s = \lfloor{\frac{d}{k}\rfloor}$. Then, there exists a bivariate polynomial $\wh{Q} \in \Ff^{(s, k-1)}_{q'}[x,y]$ such that $\wh{f}(x) = \wh{Q}(x^{k}, x)$.
\end{lemma}
\begin{proof}
    Write $\wh{f}(x) = \sum_{e = 0}^{d} C_e x^e$. For each $e$, we can write $e = k \cdot s_e + t_e$ for $s_e \leq s$ and $t_e \leq k-1$. Let $C_{e_1, e_2} = C_{k \cdot e_1 + e_2}.$ Then we may take 
    \[
    \wh{Q}(x, y) = \sum_{e_1 = 0}^{s} \sum_{e_2 = 0}^{k-1} C_{e_1, e_2} x^{e_1} y^{e_2}.
    \]
    It is clear that $\wh{Q} \in \Ff^{(s, k-1)}_{q'}[x,y]$ and $\wh{f}(x) = \wh{Q}(x^{k}, x)$.  
\end{proof}
\cref{lem:split_rs_rm} will help us in constructing IOPPs for Reed-Solomon codes by reducing it to local-testing type problems for Reed-Muller codes.

\subsection{Vanishing Polynomials}
Given a set $A \subseteq \Ff_{q'}$ we define the polynomial $\wh{V}_A \in \Ff_{q'}[x]$ as
\[
\wh{V}_A(x) = \prod_{\alpha \in A}(x-\alpha).
\]
Oftentimes we will be specifically interested in function which is the evaluation of $\wh{V}_A$ over a subfield, $\Ff_q$. We denote this function by $V_A: \Ff_q \to \Ff_{q'}$ so that $V_A(x) = \wh{V}_A(x)$ for all $x \in \Ff_q$. The subfield $\Ff_q$ will always be clear from context. Note that this definition makes sense even if $A$ is not contained in $\Ff_q$. 

When there is more than $1$ variable and $A \subseteq \Ff_{q'}$, we define $\wh{V}_{A,i} \in \Ff_{q'}[x_1,\ldots, x_m]$ by
\[
\wh{V}_{A,i}(x) = \prod_{\alpha \in A}(x_i - \alpha).
\]
Likewise, we use $V_{A,i}$ to denote the function which is the evaluation of $\wh{V}_{A,i}$ over a subfield vector space, $\Ff_q^m$, and the subfield $\Ff_q$ will always be clear from context. When the variable $x_i$ is clear from context, we will omit the index $i$ from the subscript and simply write $V_A(x)$.

A useful fact about vanishing polynomials, in the univariate case, is that they can be used to decompose polynomials which are known to vanish on some set.

\begin{lemma} \label{lm: univariate vanish division}
    Given a polynomial $\wh{f} \in \Ff_{q'}[x]$, we have that $\wh{f}(x) = 0$ for all $x \in A$ if and only if $\wh{f}(x) = \wh{V}_A(x) \cdot \wh{g}(x)$ for some polynomial $\wh{g} \in \Ff_{q'}[x]$. Furthermore, if $d \geq |A|$ is a degree parameter, we have that if $\deg(f) \leq d$, then $\deg(g) \leq d-|A|$.
\end{lemma}
\begin{proof}
    The reverse direction is clear. For the forward direction, suppose $\wh{f}(x) = 0$ for all $x \in A$. By the polynomial division theorem, there is a unique remainder polynomial $\wh{R} \in \Ff_{q'}^{|A|-1}[x]$ such that there exists $\wh{g} \in \Ff_{q'}[x]$ satisfying, 
    \[
    \wh{f} = \wh{V}_A \cdot \wh{g} + \wh{R}.
    \]
    Fixing $\wh{g}$ and $\wh{R}$ which satisfy this equation, we get that,
    \[
    \wh{V}_A(x) \cdot \wh{g}(x) + \wh{R}(x) = \wh{R}(x) = 0,
    \]
    for all $x \in A$. Hence $\wh{R}$ must be the zero polynomial as it has degree at most $|A|-1$ but $|A|$ roots. It follows that 
    \[
    \wh{f}(x) = \wh{V}_A(x) \cdot \wh{g}(x)
    \]
    for some polynomial $\wh{g} \in \Ff_{q'}[x]$. The furthermore is clear by setting $\deg(\wh{f}) = \deg(\wh{V}_A \cdot \wh{g}) = |A| + \deg(\wh{g})$.
\end{proof}

Taking \cref{lm: univariate vanish division} one step further, we can also obtain a decomposition for polynomial $\wh{f}$ which we know agrees with some function $h: A \to \Ff_{q'}$, i.e.\ $\wh{f}(x) = h(x)$ for all $x \in A$.

\begin{lemma} \label{lm: side condition vanish decomp}
    Let $\wh{f} \in \Ff_{q'}[x]$ be a polynomial, let $h: A \to \Ff_{q'}$ be a function over $A \subseteq \Ff_{q'}$, and let $\wh{h} \in \Ff^{\leq |A|-1}_{q'}[x]$ be the low degree extension of $h$ to $\Ff_{q'}$. Then, $\wh{f}(x) = h(x)$ for all $x \in A$ if and only if there exists a polynomial $\wh{g} \in \Ff_{q'}[x]$ such that 
    \[
    \wh{f} = \wh{V}_A \cdot \wh{g} + \wh{h}.
    \]
    Furthermore, if $d \geq |A|$ is a degree parameter, we have that if $\deg(f) \leq d$, then $\deg(g) \leq d-|A|$.
\end{lemma}
\begin{proof}
    The lemma follows immediately by applying \cref{lm: univariate vanish division} to $\wh{f} - \wh{h}$.
\end{proof}

We can also obtain multivariate versions of \cref{lm: univariate vanish division} and \cref{lm: side condition vanish decomp} via a multivariate version of the polynomial division algorithm.

\begin{lemma}  \label{lm: multivariate vanish division}
    Let $\wh{f} \in \Ff_{q'}[x_1, \ldots, x_m]$ be a multivariate polynomial, for each $i \in [m]$ let $\wh{V}_{A_i} \in \Ff_{q'}[x_i]$ be the vanishing polynomial in the variable $x_i$ over the set $A_i \subseteq \Ff_{q'}$, and let $d_i = |A_i|$. Then there exists $\wh{R} \in \Ff_{q'}^{\leq (d_1-1, \ldots, d_m-1)}[x_1, \ldots, x_m]$ and $\wh{Q}_1 \ldots, \wh{Q}_m \in \Ff_{q'}[x_1,\ldots, x_m]$ such that
    \[
    \wh{f} = \wh{R} + \sum_{i=1}^m \wh{V}_{A_i} \cdot \wh{Q}_i 
    \]
    Moreover, $\wh{R}$ is unique, in the sense that any for any other polynomials $\wh{R}' \in \Ff_{q'}^{(d_1-1, \ldots, d_m-1)}[x_1,\ldots,x_m]$ and $\wh{Q}'_1, \ldots, \wh{Q}'_m \in \Ff_{q'}[x_1,\ldots, x_m]$ satisfying the above equation, we must have $\wh{R}' = \wh{R}$. Consequently, if $\wh{f}$ vanishes over $A_1 \times \cdots \times A_m$, then $\wh{R} = 0$.
\end{lemma}
\begin{proof}
To produce $\wh{R}, \wh{Q}_1, \ldots, \wh{Q}_m$, we describe an algorithm 
 which is very similar to the division algorithm for univariate polynomials. Start off with $\wh{Q}_1 = 0$ and $\wh{f}' = \wh{f}$ as given. While $\wh{f}'$ has a monomial $c \cdot x_1^{e_1}\cdots x_m^{e_m}$ with $e_1 \geq d_1$ with $c \neq 0$, add the monomial $c \cdot x_1^{e_1-d_1}\cdots x_m^{e_m}$ to $\wh{Q}_1$ and subtract $c \cdot x_1^{e_1-d_1}\cdots x_m^{e_m} \cdot \wh{V}_{A_1}$ from $\wh{f}'$. In other words we perform the following two updates
 \[
 \wh{Q}_1 \to \wh{Q}_1 + c \cdot x_1^{e_1-d_1}\cdots x_m^{e_m}
 \]
 and
 \[
 \wh{f}' \to \wh{f}' - c \cdot x_1^{e_1-d_1}\cdots x_m^{e_m}  \cdot \wh{V}_{A_1}.
 \]
 Note that during this process, we do not increase the maximum $x_i$ degree of $\wh{f}$ for any $i \neq 1$, and we preserve the property
 \begin{equation} \label{eq: long div invar}    
 \wh{f} = \wh{f}' + \wh{V}_{A_1} \cdot \wh{Q}_1.
  \end{equation}

 Moreover, it is clear that continuing in this manner, we will obtain (in a finite number of steps) $\wh{Q}_1$ and $\wh{f}'$ such that \eqref{eq: long div invar} still holds, and $\wh{f}'$ has no monomials with $x_1$ degree at least $d_1$. We perform a similar algorithm to obtain $\wh{Q}_2,\ldots, \wh{Q}_m$, and note that the final $\wh{f}'$ has individual degrees at most $(d_1, \ldots, d_m)$. Setting $\wh{R} = \wh{f}'$, we get the desired polynomials. 

 For the moreover part, note that for any other  $\wh{R}' \in \Ff_{q'}^{(d_1-1, \ldots, d_m-1)}[x_1,\ldots,x_m]$ and $\wh{Q}'_1, \ldots, \wh{Q}'_m \in \Ff_{q'}[x_1,\ldots, x_m]$ satisfying \eqref{eq: long div invar}, we must have 
 \[
 \wh{R}|_{A_1 \times \cdots \times A_m} = \wh{R}'|_{A_1 \times \cdots \times A_m}.
 \]
 Since both $\wh{R}$ and $\wh{R}'$ have individual degrees $(d_1-1, \ldots, d_m-1)$, and $|A_i| = d_i$, it follows that $\wh{R} = \wh{R}'$. For the consequently part, note that if $\wh{f}$ vanishes on $A_1\times \cdots \times A_m$, then $\wh{R}$ also vanishes on $A_1\times \cdots \times A_m$. It since $\wh{R}$ has individual degrees $(d_1-1, \ldots, d_m-1)$ it must be the case that $\wh{R} = 0$.
\end{proof}

This immediately yields an analog of~\cref{lm: side condition vanish decomp} in the multi-variate setting.
\begin{lemma} \label{lm: multivariate side condition vanish decomp}
    Let $\wh{f} \in \Ff_{q'}[x_1,\ldots, x_m]$ be a polynomial, let $h: A_1 \times \cdots \times A_m \to \Ff_{q'}$ be a function over the product set $A_1 \times \cdots \times A_m \subseteq \Ff^m_{q'}$, and let $\wh{h} \in \Ff^{\leq (|A_1|-1, \ldots, |A_m|-1)}_{q'}[x_1,\ldots, x_m]$ be the low degree extension of $h$ to $\Ff^m_{q'}$. Then, $\wh{f}(x) = h(x)$ for all $x \in A_1 \times \cdots \times A_m$ if and only if there exists a polynomials $\wh{g}_1,\ldots, \wh{g}_m \in \Ff_{q'}[x_1,\ldots, x_m]$ such that 
     \[
    \wh{f}(x) = \sum_{i=1}^{m} \wh{V}_{A, i}(x_i) \cdot \wh{g}_i(x) + \wh{h}(x).
    \]
     Furthermore, for some individual degree parameter $(d_1, \ldots, d_m)$ such that each $d_i \geq |A_i|$, we have that each $\wh{g}_i$ has individual degree $(d_1, \ldots, d_i - |A_i|, \ldots, d_m)$ if $\wh{f}$ has individual degree $\wh{f}$.
\end{lemma}
\begin{proof}
    The lemma follows immediately by applying \cref{lm: multivariate vanish division} to $\wh{f} - \wh{h}$.
\end{proof}

\section{Proximity Generators and Combining Functions of Different Degrees}\label{sec:proximity_gen}

Proximity generators were first introduced in \cite{bciks} and provide a way of combining functions while preserving their proximity to some code. Suppose we have functions $f_1, \ldots, f_k: \Ff_q \to \Ff_{q'}$ such that at least one of them has agreement at most $\eps$ with $\RS_{q'}[d, \Ff_q]$. In order to test the proximity of each of $f_1, \ldots, f_k$ to $\RS_{q'}[d, \Ff_q]$ we could run $k$ separate IOPPs, however this may not be efficient for large $k$. Instead, we would like to combine them into a single function, and only test that. This task is achieved via proximity generators. 

Intuitively, if at least one of the $f_i$'s is far from degree $d$, then a random linear combination of them should still be far from degree $d$, as the errors should not cancel each other out. The proximity generator theorem of \cite{bciks} makes a formal assertion along these lines, and they prove that choosing $\xi_1, \ldots, \xi_k \in \Ff_{q'}$ according to some specific distribution, with high probability $\sum_{i=1}^k \xi_i \cdot f_i$ has agreement at most $\eps$ with degree $d$ functions. 
Thus, proximity generators allow us to test the proximity of $k$-functions simultaneously by testing a random linear combination of them.

The goal of this section is to give proximity generators for Reed-Muller codes. While it may be the case that the techniques of~\cite{bciks} can be adapted to the setting of the Reed-Muller codes, it is easier for us to derive the Reed-Muller proximity generators by appealing to the results of~\cite{bciks} in a black-box way. The techniques involved are standard in the PCP literature and rely on the line versus point test \cite{as, hkss}. Our proofs for the Reed-Muller proximity generators are given in \cref{sec: prox gen proof}.

\paragraph{Combining functions of different degrees:}  we also show how 
to use proximity generators to combine functions of various degrees. That is, suppose we want to design an IOPP in the scenario we have functions $f_1, \ldots, f_k: \Ff_q \to \Ff_{q'}$ such that some $f_i$ only has a small correlation with $\RS_{q'}[d_i, \Ff_q]$ (the $d_i$'s may be different). Again, we could run $k$ separate IOPPs for each function and degree, but to be more efficient, we can instead use proximity generators. A slight adaptation allows us to instead test proximity of a single function, $F$, which is derived from $f_1, \ldots, f_k$, to some arbitrary degree $d \geq \max_i d_i$ of our liking.

\subsection{Proximity Generators} 
\subsubsection{Reed-Solomon Proximity Generators}
Let us start by stating the proximity generator theorem for Reed-Solomon codes. Throughout this section, we fix finite field $\Ff_q, \Ff_{q'}$ such that $\Ff_q$ is a subfield of $\Ff_{q'}$, and a degree parameter $d$. One should think of $q' \gg q \cdot \poly(d)$. Set $\eps$ to be an agreement parameter. Throughout this section, we assume that 
\[
\eps \geq 3 \cdot \left(\frac{d}{q}\right)^{\tl} \geq 3 \cdot \left(\frac{d}{q}\right)^{1/20},
\]
unless stated otherwise. We also define
\[
\err(d,q,q') = \frac{q^4}{d^2 \cdot q'}
\]
to be the error function of the proximity generator from \cite{bciks}, and let
\[
\prox(k, d, q, q')
\]
be the proximity generator for the code $\RS_{q'}[d, \Ff_q]$ from \cite{bciks}. When the parameters $k,d,q,q'$ are clear from context, we simply use $\prox$ to refer to this distribution.

\begin{thm} [\cite{bciks}] \label{thm: prox gen RS with corr}
    Let $\eps \geq (d/q)^{1/2}$ be an agreement parameter and suppose the functions $f_1,\ldots, f_k: \Ff_q \to \Ff_{q'}$ satisfy
    \[
    \Pr_{(\xi_1,\ldots, \xi_k) \sim \prox(k,d,q,q')}\left[\agr_d\left(\sum_{i=1}^k \xi_i \cdot f_i\right) \geq \eps\right] \geq k \cdot \err(d, q, q').
    \]
    Then, there is a set $A \subseteq \Ff_q$ of size $\mu(A) \geq \eps$ such that for each $f_i$, there is a degree $d$ polynomial $h_i: \Ff_q \to \Ff_{q'}$ satisfying,
    \[
    f_i|_A = h_i|_A,
    \]
    for each $i \in [k]$.

    For the special case where $k = 2$, we remark that $\prox(2,d,q,q')$ is the uniform distribution over $\Ff_{q'}^2$. In this case we will write $\xi_1, \xi_2 \in \Ff_{q'}$.
\end{thm}

\subsection{Proximity Generators for Reed-Muller Codes}\label{sec: prox gen intro}
The theorems in the remainder of this section give versions of \cref{thm: prox gen RS with corr} that work for both total degree Reed-Muller codes and individual degree Reed-Muller codes.

\subsubsection{Proximity Generators with Correlated Agreement}
In our first adaptation, we are only able to handle a constant number of functions, i.e.\ the $k = O(1)$ case.
For the total degree case, our result reads as follows:
\begin{thm} [Correlated $\RM$ Proximity Generator] \label{thm: prox gen RM with corr}
    Suppose $k\geq 2$, $\eps\geq \frac{10 k}{q^{1/(2(7k+1))}}$
    and 
    $f_1, \ldots,f_k: \Ff_q^2 \to \Ff_{q'}$ satisfy
    \[
    \Pr_{\xi_1, \ldots,\xi_k \sim \prox(k,d,q,q')}\left[\agr\left(\sum_{i=1}^k \xi_i \cdot f_i, \RM_{q'}[d, \Ff_q^2] \right) \geq \eps\right] \geq \err(d,q,q'),
    \]
    then there is a subset of points $A \subseteq \Ff_q^2$ of size 
    \[
    \mu(A) \geq 0.999\eps
    \]
    and functions $h_1, h_2\in \RM_{q'}[d, \Ff_q^2]$ such that, $f_i|_A= h_i|_A$ for each $i \in [k]$.
\end{thm}
\begin{proof}
    The proof is deferred to \cref{sec: prox gen proof}.
\end{proof}

Similarly, for the individual degree case our result reads as follows:
\begin{thm} [Correlated Individual-$\RM$ Proximity Generator] \label{thm: prox gen iRM with corr}
Let $\eps \geq \frac{20}{q^{2/15}}$ and suppose  $f_1, f_2: \Ff_q^m \to \Ff_{q'}$ satisfy
    \[
    \Pr_{\xi_1, \xi_2 \in \Ff_{q'}}\left[\agr\left(\sum_{i=1}^2 \xi_i \cdot f_i, \RM_{q'}\left[(d,d), \Ff_q^2\right] \right) \geq \eps\right] \geq 2 \cdot {\sf err}(2d,q,q'),
    \]
    then there is a subset of points $A \subseteq \Ff_q^m$ of size 
    \[
    \mu(A) \geq 0.99\eps
    \]
    and oracle functions $h_1, h_2\in \RM_{q'}[(d,d), \Ff_q^2]$ such that, $f_i|_A= h_i|_A$ for each $i \in [2]$.
\end{thm}
\begin{proof}
    The proof is deferred to \cref{sec: proof of prox gen irm with corr}.
\end{proof}
\subsubsection{Proximity Generators without Correlated Agreement}
We also need versions of~\cref{thm: prox gen RS with corr} for larger $k$'s. In that case, our argument gets a conclusion weaker than the one in~\cref{thm: prox gen RS with corr}. Instead of being able to say that there is a sizable set $A$ wherein all $f_i$'s agree with some low-degree function, we are only able to guarantee that each $f_i$ has non-trivial agreement with a degree $d$ polynomial (i.e., the points of agreement of $f_i$ with its low-degree polynomial may be completely different than those of $f_{i'}$). This statement is sufficient for our purposes, but we suspect it should
be possible to get the stronger conclusion.
\begin{thm} [Uncorrelated $\RM$ Proximity Generator] \label{thm: prox gen RM without corr}
    Let $f_1, \ldots, f_k: \Ff_q^2 \to \Ff_{q'}$ be oracle functions. If
    \[
    \Pr_{(\xi_1, \ldots, \xi_k) \sim \prox(k,d,q,q')}\left[\agr\left(\sum_{i=1}^k \xi_i \cdot f_i, \RM_{q'}[d, \Ff^2_q] \right) \geq \eps \right] \geq 2k \cdot \err(d,q,q'),
    \]
     then  
    \[
    \agr\left(f_i, \RM_{q'}[d, \Ff_q^2]\right) \geq 0.99\eps \quad  \forall i \in [k].
    \]
\end{thm}
\begin{proof}
    The proof is deferred to \cref{sec: prox gen proof}.
\end{proof}


\begin{remark}
    We remark that our results only hold for agreement parameter $\eps \geq \Omega\left((d/q)^{\tl}\right)$, whereas the univariate Reed-Solomon proximity generator holds for $\eps$ as small as the Johnson bound, $\eps= \sqrt{d/q}$. The FRI conjecture posits that the univariate Reed-Solomon proximity generator holds for the smallest possible agreement, $\eps = d/q$. It is an intriguing future direction to improve the quantitative results of our proximity generators. One avenue towards this would be to improve $\tl$ in the soundness of the line versus point tests. We remark that higher dimensional versions of the line versus point test (which use planes or cubes) achieve improved and even optimal dependence on $q$ in the soundness \cite{MRpvp, BDN, MZ23}.
\end{remark}
    
\subsection{Combining Functions of Different Degrees} \label{sec: combine intro}

Next, we give versions of proximity generators in the case the degrees of the various functions may be different. 

\subsubsection{Univariate}
We first show how to go from a univariate function far from degree $d$ to a univariate function far from degree $d' \geq d$. \begin{definition} \label{def: uni combine}
Let $f_1, \ldots, f_\ell: \Ff_q \to \Ff_{q'}$ be functions, let $d_1,\ldots, d_{\ell}$ be degree parameters and let $d'$ be a target degree which is greater than all of the $d_i$. Then the combined function is generated as follows:
\begin{itemize}
    \item Choose $\xi_1, \ldots, \xi_{\ell} \in \Ff_{q'}$  according to the Reed-Solomon proximity generator, and choose $\xi_0 \in \Ff_{q'}$ uniformly at random.
    \item Set \[
\combine_{d'}(f_1,\ldots, f_{\ell}) = \sum_{i=1}^{\ell} \xi_i \cdot (f_i + \xi_0 \cdot x^{d' - d_i}\cdot f_i).
\]
\end{itemize}
\end{definition}
Clearly, if $f_i$ has degree at most $d_i$, 
then $f_i+\xi_0\cdot x^{d'-d_i}\cdot f_i$ has degree at
most $d'$. The following lemma shows that the converse direct also works in a robust sense:
\begin{lemma} \label{lm: combine one univariate}
    Fix degree parameters $d \leq d'$ and an agreement parameter $\eps \geq \sqrt{\frac{d'}{q}}$. Suppose $f: \Ff_q \to \Ff_{q'}$ satisfies $\agr(f, \RS_{q'}[d, \Ff_q]) \leq \eps$. Then
    \[
    \Pr_{\xi \in \Ff_{q'}}[\agr(f + \xi \cdot x^{d' - d}\cdot f, \RS_{q'}[d', \Ff_q]) \geq \eps] < \err = \frac{\poly(d')}{q'}.
    \]
\end{lemma}
\begin{proof}
    Suppose for the sake of contradiction that this is not the case. Then, by \cref{thm: prox gen RS with corr}, there exists $U \subseteq \Ff_q$ of fractional size at least $\eps$ and degree $d$ functions $g_1, g_2$ such that such that $f|_U = g_1|_U$ and $x^{d'-d} \cdot f|_U = g_2|_U$. It follows that 
    \[
    x^{d'-d}\cdot g_1|_U = g_2|_U.
    \]
    Since $U$ has fractional size greater than $d'/q$, it follows by the Schwartz-Zippel lemma that $x^{d'-d}\cdot g_1 = g_2$. Hence $g_1$ has degree at most $d$ and agrees with $f$ on a set of fractional size greater than $\eps$, which is a contradiction.
\end{proof}

We now give our proximity generators in the case the functions have different degrees.
\begin{lemma} \label{lm: combine multiple univariate}
Let $f_1, \ldots, f_{\ell}: \Ff_{q} \to \Ff_{q'}$ be functions and $d_1, \ldots, d_{\ell}$ be agreement parameters. Let $d' \geq \max_i(d_i)$ and suppose $i$, $\agr(f_i, \RS_{q'}[d_i, \Ff_q]) \leq \eps$ for some $\eps$ satisfying
\[
\eps \geq \sqrt{\frac{d'}{q}}.
\]
Then,
\[
\Pr[\agr(\combine_{d'}(f_1,\ldots, f_{\ell}),\RS_{q'}[d', \Ff_q]) \geq \eps ] \leq  \ell \cdot \err,
\]
where the probability is over the randomness in the combine function.
\end{lemma}
\begin{proof}
Let $i$ be the index such that $\agr(f_i, \RS_{q'}[d_i, \Ff_q]) \leq \eps$. Then,
\[
\Pr_{\xi_0 \in \Ff_{q'}}[\agr(f_i + \xi_0 \cdot x^{d' - d_i}\cdot f_i, \RS_{q'}[d', \Ff_q]) \geq \eps] \leq \err.
\]
by \cref{lm: combine one univariate}. Conditioned on this not being the case, we have that,
\[
\Pr_{\xi_1,\ldots, \xi_{\ell} \in \Ff_{q'}}[\agr(\combine_{d'}(f_1,\ldots, f_{\ell}),\RS_{q'}[d', \Ff_q] \geq \eps ] \leq (\ell-1)\cdot \err,
\]
by \cref{thm: prox gen RS with corr}. The lemma follows from a union bound.
\end{proof}

\subsubsection{Multivariate}
We now give a multi-variate version of our proximity generators for functions with different degrees, closely following the ideas in the univariate case.

\paragraph{Individual Degree:} First show how to combine degrees while preserving proximity to the individual degree Reed-Muller code. In this case, 
the ``combine'' operation works as follows:
\begin{definition} \label{def: multi combine}
Let $f_1, \ldots, f_\ell: \Ff_q^2 \to \Ff_{q'}$ be oracle functions, let $(d_{i,1}, d_{i,2})$ be degree parameters for $i \in [\ell]$, and let $(d', d')$ be a target degree which is greater than all of the $(d_{i,1}, d_{i,2})$. Then the combined function is generated as follows:
\begin{itemize}
    \item Choose $\xi_1, \ldots, \xi_{\ell} \in \Ff_{q'}$  according to the Reed-Solomon proximity generator, and choose $\xi_0 \in \Ff_{q'}$ uniformly at random.
    \item Set \[
\combine_{(d', d')}(f_1,\ldots, f_{\ell}) = \sum_{i=1}^{\ell} \xi_i \cdot (f_i + \xi_0 \cdot x_1^{d' - d_{i,1}}x_2^{d' - d_{i,2}}\cdot f_i).
\]
\end{itemize}
\end{definition}

\begin{lemma} \label{lm: combine one multivariate}
    Fix degree parameters $(d_1, d_2) \leq (d', d')$ and $\eps \geq 20 \cdot \left(\frac{d'}{q}\right)^{\tl}$.    
    Suppose that $f: \Ff^2_q \to \Ff_{q'}$ satisfies $\agr(f, \RM_{q'}[(d', d'), \Ff^2_q]) \leq \eps$. Then
    \[
    \Pr_{\xi \in \Ff_{q'}}[\agr(f + \xi \cdot x_1^{d' - d_1}x_2^{d'-d_2}\cdot f, \RM_{q'}[(d',d'), \Ff_q]) \geq 1.01\eps] < \err = \frac{\poly(d')}{q'}.
    \]
\end{lemma}
\begin{proof}
    Suppose for the sake of contradiction that this is not the case. Then, by \cref{thm: prox gen iRM with corr}, there exists $U \subseteq \Ff^2_q$ of fractional size at least $0.99\eps$ and degree $(d,d)$ functions $g_1, g_2$ such that such that 
    \[
    f|_U = g_1|_U \quad  \text{and} \quad x_1^{d'-d_1}x_2^{d'-d_2} \cdot f|_U = g_2|_U.
    \]
    It follows that 
    \[
    x_1^{d'-d_1}x_2^{d'-d_2} \cdot g_1|_U = g_2|_U.
    \]
    Since $U$ has fractional size greater than $d'/q$, it follows that $x_1^{d'-d_1}x_2^{d'-d_2} \cdot g_1 = g_2$ by the Schwartz-Zippel lemma. Hence $g_1$ must have individual degrees at most $(d_1, d_2)$ and agree with $f$ on a set of fractional size greater than $0.99\eps$. This is a contradiction.
\end{proof}

The proximity generator in this now easily follows:
\begin{lemma} \label{lm: combine multiple multivariate}
Let $f_1, \ldots, f_{\ell}: \Ff^2_{q} \to \Ff_{q'}$ be functions, let $(d_{i,1}, d_{i,2})$ be individual degree parameters for $i \in [\ell]$, let $(d', d')$ be a target degree which is greater than each $(d_{i,1}, d_{i,2})$, and let $\eps \geq 21\left(\frac{d'}{q}\right)^{\tl}$ be an agreement parameter. Suppose that $\agr(f_i, \RM_{q'}[(d_{i,1}, d_{i,2}), \Ff_q]) \leq \eps$ for some $i\in [\ell]$. Then
\[
\Pr[\agr(\combine_{(d',d')}(f_1,\ldots, f_{\ell}),\RM_{q'}[(d', d'), \Ff_q] \geq 1.03\eps ] \leq \ell\cdot \err,
\]
where the probability is over the randomness in the $\combine$ function.
\end{lemma}
\begin{proof}
Let $i$ be an index such that $\agr(f_i, \RM_{q'}[(d_{i,1}, d_{i,2}), \Ff^2_q]) \leq \eps$. Then,
\[
\Pr_{\xi_0 \in \Ff_{q'}}[\agr(f_i + \xi_0 \cdot x_1^{d' - d_{i,1}}x_2^{d' - d_{i,2}}\cdot f_i, \RM_{q'}[(d',d'), \Ff_q]) \geq 1.01 \eps] \leq \err
\]
by \cref{lm: combine one multivariate}. Conditioned on this not being the case, we have that,
\[
\Pr_{\xi_1,\ldots, \xi_{\ell} \in \Ff_{q'}}\left[\agr\left(\combine_{(d',d')}(f_1,\ldots, f_{\ell}),\RM_{q'}[(d',d'), \Ff_q]\right) \geq 1.03 \eps \right] \leq (\ell-1)\err,
\]
by the contrapositive of \cref{thm: prox gen iRM with corr}. The lemma follows from a union bound.
\end{proof}

\paragraph{Total Degree:} For the total degree Reed-Muller code, we can show the following analogues of \cref{lm: combine multiple multivariate}. The proofs are nearly identical to that of \cref{lm: combine multiple multivariate} so we omit them. The only difference is that we apply the Schwartz-Zippel lemma to the total degree Reed-Muller code in the analogue of \cref{lm: combine one multivariate} and we appeal to the total degree Reed-Muller proximity generator theorem (\cref{thm: prox gen RM without corr}) instead of the individual degree version (\cref{thm: prox gen iRM with corr}). We state our results on the total degree Reed-Muller for arbitrary dimension $m$, but for our purposes $m$ will always be $2$ or $3$.

\begin{definition} \label{def: multi combine tot}
Let $f_1, \ldots, f_\ell: \Ff^m_q \to \Ff_{q'}$ be oracle functions, let $d_1,\ldots, d_{\ell}$ be degree parameters and let $d'$ be a target degree which is greater than all of the $d_i$. Then the combined function is generated as follows:
\begin{itemize}
    \item Choose $\xi_1, \ldots, \xi_{\ell} \in \Ff_{q'}$  according to the total degree Reed-Muller proximity generator, and choose $\xi_0 \in \Ff_{q'}$ uniformly at random.
    \item Set \[
\combine_{d'}(f_1,\ldots, f_{\ell}) = \sum_{i=1}^{\ell} \xi_i \cdot (f_i + \xi_0 \cdot x_1^{d' - d_i}\cdot f_i).
\]
\end{itemize}
\end{definition}

Here we reuse the $\combine$ notation from the univariate case in \cref{def: uni combine}, but note that it will be clear from context which case we are referring to based on if the input functions are univariate or multivariate.

\begin{lemma} \label{lm: combine multiple multivariate total degree}
Let $f_1, \ldots, f_{\ell}: \Ff^m_{q} \to \Ff_{q'}$ be functions, let $d_1, \ldots, d_{\ell}$ be degree parameters for $i \in [\ell]$, let $d' \geq \max_i(d_i)$ be a target degree and let $\eps \geq 20\cdot 1.01^{m}\left(\frac{d'}{q}\right)^{\tl}$ be an agreement parameter. Suppose that $\agr(f_i, \RM_{q'}[d_i, \Ff_q]) \leq \eps$ for some $i\in[\ell]$. Then
\[
\Pr[\agr(\combine_{d'}(f_1,\ldots, f_{\ell}),\RM_{q'}[d, \Ff_q]) \geq 1.01^{m}\eps ] \leq \ell\cdot \err,
\]
where the probability is over the randomness in the $\combine$ function.
\end{lemma}
\begin{proof}
    Similar to the proof of~\cref{lm: combine multiple multivariate}.
\end{proof}

\section{Quotienting to Remove Side Conditions}\label{sec:quotient}
We will often be interested in testing the proximity of some $f: U \to \Ff_{q'}$ to a family of functions, say $\RS_{q'}[d, \Ff_q]$, subject to some constraints which we call side conditions. More specifically, we will want to test the proximity of $f$ to the subfamily of $\RS_{q'}[d, \Ff_q]$ consisting of $g$'s satisfying that $g|_A = h|_A$ for some prespecified $A$ and $h: A \to \Ff_{q'}$. We refer to $h$ as the side condition function, $A$ as the set of side condition points, and use the notation
\[
\RS_{q'}[d, \Ff_q \; | \; h] = \{g \in \RS_{q'}[d, \Ff_q] \; | \; g|_A = h|_A \}.
\]
We define $\RM_{q'}[d, \Ff_q \; | \; h]$ and $\RM_{q'}[(d,d), \Ff_q \; | \; h]$ analogously. 

The goal of this section is to we present a technique called quotienting, which will reduce the proximity testing with side conditions problem to standard proximity testing without side conditions. We will start by discussing quotienting in the univariate case, which was first given in \cite{acy, acfy}, and then move on to quotienting in the multivariate case. 

\subsection{Univariate Quotienting}
Fix an $f: \Ff_q \to \Ff_{q'}$ and a side condition function $h: A \to \Ff_{q'}$. We remark that we do not require $A \subseteq \Ff_q$; in that case, one should interpret the side conditions from $h$ as being conditions on the low degree extension of $f$ to $\Ff_{q'}$. That is, we are expecting that the low degree extension of $f$ over $\Ff_{q'}$, $\wh{f}$, satisfies $\wh{f}|_A = h|_A$. Suppose we want to test proximity of $f$ to the family $\RS_{q'}[d, \Ff_q \; | \; h]$, where the size of $A$ should be thought of as small compared to $d$. 
How should we go about doing this? 

Let $\wh{f}$ and $\wh{h}$ be the low-degree extensions of $f$ and $h$ respectively from $\Ff_q$ to $\Ff_{q'}$. Also let $\tilde{h}: \Ff_q \to \Ff_{q'}$ be denoted by $\wh{h}|_{\Ff_q}$. The idea behind quotienting is  the following fact: if $f \in \RS_{q'}[d, \Ff_q \; | \; h]$, then by~\cref{lm: side condition vanish decomp} we may write $f$ as: 
\begin{equation} \label{eq: division}  
\wh{f}(x) = \wh{V}_A \cdot \wh{g}(x) + \wh{h}(x),
\end{equation}
where, recall, $\wh{V}_A(x) = \prod_{\alpha \in A}(x-a)$ is the vanishing polynomial over $A$. We note that the polynomial $\wh{g} \in \Ff_{q'}^{\leq d- |A|}$ is uniquely determined, and given access to $f$, the verifier can simulate access to $\wh{g}$ over $\Ff_q \setminus A$. Thus, to test proximity of $f$ to $\RS_{q'}[d, \Ff_q \; | \; h]$, the verifier can use this simulated access to $\wh{g}$ and instead attempt to test proximity of $\wh{g}|_{\Ff_q}$ to $\RS_{q'}[d - |A|, \Ff_q ]$. There is a minor issue with this idea, since the verifier cannot access the values of $\wh{g}$ over $A$. To remedy this, the prover also fills in the values of $\wh{g}$ on $A$. Since $|A|$ is small, a cheating prover will not be able to use this to their advantage and affect the acceptance probability of the verifier by much. In terms of completeness, an honest prover can simply provide the true values and this way not decrease the acceptance probability of the verifier. Our discussion motivates the following definition of a quotient function:

\begin{definition}
   Given a function $f: \Ff_q \to \Ff_{q'}$, a set $ A \subseteq \Ff_{q'}$, and a fill function $\Fill: A \to \Ff_{q'}$, define the following function from $\Ff_q \to \Ff_{q'}$,
 \begin{equation} \label{eq: uni quotient def}
    \quo_1(f, A, \Fill)(x), = \begin{cases}
        \frac{f(x)}{V_A(x)} \quad \text{if } x \notin A \\
        \Fill(x) \quad \text{if } x \in A.
\end{cases}
\end{equation}
\end{definition}
The next lemma shows that the honest prover can indeed provide values for $\Fill$ which makes the quotient above a low degree function. 
\begin{lemma} \label{lm: uni quo completeness}
    If $f \in \RS_{q'}[d, \Ff_q \; | \; h]$, then there is some $\Fill: A \to \Ff_{q'}$ for which $\quo_1(f-\tilde{h}, A, \Fill) \in \RS_{q'}[d-|A|, \Ff_q]$.
\end{lemma}
\begin{proof}
    Consider the function
    \[
   g(x) =  \frac{\wh{f}(x) - \wh{h}(x)}{\wh{V}_A(x)}.
    \]
    Since the numerator vanishes on $A$, it follows that $\wh{V}_A$ divides it, and so $g$ is a polynomial. Considering degrees, it follows that $g \in \Ff_{q'}^{\leq d - |A|}[x]$. Taking $\Fill$ to be the evaluation over $A$ of $g$, we get that $\quo_1(f-\tilde{h}, A, \Fill)$ is the evaluation of $g$ over $\Ff_q$.
\end{proof}
On the other hand, the following lemma shows that if $f$ is far from  $\RS_{q'}[d, \Ff_q \; | \; h]$, then for any $\Fill: A \to \Ff_{q'}$, the quotiented function $\quo_1(f-\tilde{h}, A, \Fill)$ is far from  $\RS_{q'}[d- |A|, \Ff_q]$.
\begin{lemma} \label{lm: uni quo soundness}
    If  $\agr(f, \RS_{q'}[d, \Ff_q \; | \; h]) \leq \eps$, then for any $\Fill: A \to \Ff_{q'}$, we have
    \[
    \agr(\quo_1(f-\tilde{h}, A, \Fill), \RS_{q'}[d-|A|, \Ff_q]) \leq  \eps + \frac{|A|}{q}.
    \]
\end{lemma}
\begin{proof}
 Fix an arbitrary $\Fill: A \to \Ff_{q'}$ and let $g = \quo_1(f-\tilde{h}, A, \Fill)$. Suppose $g|_U = F|_U$ for some $F \in \RS_{q'}[d-|A|, \Ff_q]$. Then $F \cdot \wh{V}_A|_{\Ff_q} + \tilde{h} \in \RS_{q'}[d, \Ff_q \; | \; h]$, and for $x \in U \setminus A$, we have
 \[
 F(x) \cdot \wh{V}_A(x)+ \wh{h}(x) = g(x) \cdot \wh{V}_A(x) + \wh{h}(x) = f(x).
 \]
It follows that $F \cdot \wh{V}_A|_{\Ff_q} + \tilde{h}$ agrees with $f$ on $U \setminus A$. Thus, $\mu(U \setminus A) \leq \eps$ and $\mu(U) \leq \eps + |A|/q$.
\end{proof}

\subsection{Bivariate Quotienting}
Bivariate Quotienting is a natural extension of univariate quotienting. Fix a function $f: \Ff_q^2 \to \Ff_{q'}$ and consider a side condition function $h: A \times B \to \Ff_{q'}$. Note that here we require the side condition function to be defined over a subcube. Similarly to the univariate case, we do not require $A \times B \subseteq \Ff_{q}^2$, and we think of $A, B$ as small compared to $\Ff_q$. Throughout this section, let $\wh{V}_A, \wh{V}_B \in \Ff_{q'}[x]$ be the vanishing polynomials of $A$ and $B$ and $V_A$ and $V_B$ be their evaluations over $\Ff_q$. Also let $\wh{h} \in \Ff_{q'}^{(|A|-1,|B|-1)}[x,y]$ be the low degree extension of $h$, and let $\tilde{h}$ be $\wh{h}|_{\Ff_q^2}$.
 
 Suppose we want to test proximity of $f$ to $\RM_{q'}[(d,d), \Ff_q^2 \; | \; h ]$. By~\cref{lm: multivariate side condition vanish decomp}, 
 we may write the low degree extension of $f \in \RM_{q'}[(d,d), \Ff_q^2 \; | \; h ]$ as:
 \begin{equation} \label{eq: side condition decomp}  
 \wh{f}(x,y) = \wh{V}_A(x)\cdot \wh{g}_1(x,y) + \wh{V}_B(y) \cdot \wh{g}_2(x,y) + \wh{h}(x,y),
  \end{equation}
where $\wh{g}_1 \in \Ff^{(d - |A|, d)}[x,y], \wh{g}_2 \in \Ff_{q'}^{(d, d-|B|)}[x,y]$. We note that~\eqref{eq: side condition decomp} suggests that in order to test proximity of $f$ to $\RM_{q'}[(d,d), \Ff_q^2 \; | \; h ]$, the verifier can instead attempt to simultaneously test the proximities of $\wh{g}_1|_{\Ff_q^2}$ to $\RM_{q'}[(d-|A|, d), \Ff_q^2]$ and $\wh{g}_2|_{\Ff_q^2}$ to $\RM_{q'}[(d, d-|B|), \Ff_q^2]$.\footnote{Ultimately, we will use proximity generators in a suitable fashion to reduce this to testing proximity to a single family, without side conditions.}

One difference with the univariate case is that the verifier cannot access $\wh{g}_1$ or $\wh{g}_2$ given only access to $\wh{f}$ and $\wh{h}$. To resolve this, the prover will provide $g_1: \Ff_q^2 \to \Ff_{q'}$, in which scenario the verifier will be able to simulate access to $\wh{g}_2$ over $\Ff_q^2 \setminus (\Ff_q \times B)$. Since $B$ is small, this is nearly all of $\Ff_q^2$, and as in the univariate case the prover fills in the rest of the values of $g_2$. The bivariate quotient function is then defined as follows:
\begin{definition}
    Given a function $f: \Ff_q^2 \to \Ff_{q'}$, a set $B \subseteq \Ff_{q'}$, and a function $\Fill: \Ff_q \times B \to \Ff_{q'}$, define the following function from $\Ff_q^2 \to \Ff_{q'}$
    \begin{equation} \label{eq: multi quotient def}
\quo_2(f, B, \Fill)(x,y) = \begin{cases}
    \frac{f(x,y)}{V_B(y)} \quad \text{if } y\notin B \\
    \Fill(x,y) \quad \text{if } y \in B.
\end{cases}
\end{equation}
\end{definition}

In analogy to the univariate case, we have two lemmas regarding the proximities of the quotiented functions to respective Reed-Muller codes without side conditions. To state them, we introduce the notion of correlated agreement. 
\begin{definition}
    We say that two functions $g_1, g_2: \Ff_q^2 \to \Ff_{q'}$ have $\eta$-correlated agreement with degrees $(d_1, d_2)$ and $(d'_1, d'_2)$ respectively if there exist functions $G_1, G_2: \Ff^2_{q} \to \Ff_{q'}$ of degrees $(d_1, d_2)$ and $(d'_1, d'_2)$ respectively and a set $U \subseteq \Ff_q^2$ of fractional size at least $\eta$ such that
\[
g_1|_U = G_1|_U \quad \text{and} \quad g_2|_U = G_2|_U.
\]
\end{definition}
\begin{remark}
The reason that $\eta$-correlated agreement with individual degrees is useful is that if $g_1$ and $g_2$ \emph{do not} have such correlated agreement, then we can combine them using \cref{lm: combine multiple multivariate} and obtain, with high probability, a function that is far from having low individual degree. 
\end{remark}

With this in hand, the following lemma handles the completeness case and shows that the honest prover can provide functions $g_1$ and $\Fill$ which will result in a quotient with low individual degrees in the case that $f \in \RM_{q'}[(d, d), \Ff_q^2 \; | \; h]$:
\begin{lemma} \label{lm: ind deg bivariate side condition decomp completeness case}
Let $f \in \RM_{q'}[(d, d), \Ff_q^2 \; | \; h]$, and recall that $\tilde{h} = \wh{h}|_{\Ff^2_{q}}$. Then there exist a $g_1: \Ff_q^2 \to \Ff_{q'}$ of degree $(d - |A|, d)$ and $\Fill: \Ff_q \times B \to \Ff_{q'}$ such that, 
    \[
    \Quo_2(f(x,y) - V_A(x)\cdot g_1(x,y) - \tilde{h}(x,y), B, \Fill) \in \RM_{q'}[(d, d-|B|), \Ff_q^2].
    \]
\end{lemma}
\begin{proof}
Let $\wh{g}_1, \wh{g}_2 \in \Ff_{q'}[x,y]$ be the functions from the decomposition in~\eqref{eq: side condition decomp}. 
Take $g_1 = \wh{g}_1|_{\Ff_q^2}$, and note that (after division)
$\frac{f(x,y) - V_A(x)\cdot g_1(x,y) - \tilde{h}(x,y)}{V_B(y)}$
is a polynomial, so we may take $\Fill$ agreeing with it over $\Ff_q \times B$. It is clear that $g_1$ is of the desired degree. As for, $ \Quo_2(f(x,y) - V_A(x)\cdot g_1(x,y) - \tilde{h}(x,y), B, \Fill)$, note that it is the evaluation of $\wh{g}_2$ from~\eqref{eq: side condition decomp}, so it is also of the desired degree.
\end{proof}

The next lemma handles the soundness case and is a bit more technical. Morally speaking, it says that $f \in \RM_{q'}[d, \Ff_q^2 \; | \; h]$, then there is no way for the prover to provide a $g_1$ and $\Fill$ such that $g_1$ and the quotient calculated by the verifier are consistent. 

\begin{lemma} \label{lm: bivariate side condition decomp}
Let $h: A \times B \to \Ff_{q'}$ be a side condition function, with $A, B \subseteq \Ff_{q'}^2$ and recall that $\tilde{h} = \wh{h}|_{\Ff^2_{q}}$. Suppose that $Q: \Ff_q^2 \to \Ff_{q'}$ satisfies
\[
\agr(Q, \RM_{q'}[(d,d), \Ff_q^2 \; | \; h]) \leq \eps.
\]
Then for any $g_1: \Ff_q^2 \to \Ff_{q'}$, $\Fill: \Ff_q \times B \to \Ff_{q'}$, and $g_2 = \Quo_2(Q - V_A \cdot g_1 - \tilde{h}, B, \Fill)$, the functions $g_1, g_2$ cannot have $\eps + \frac{|B|}{q}$-correlated agreement with individual degrees $(d - |A|, d)$ and $(d, d - |B|)$ respectively.
\end{lemma}
\begin{proof}
    Suppose for the sake of contradiction that the claim is false, and let $G_1, G_2$ be the functions with individual degrees $(d - |A|, d)$ and $(d, d - |B|)$ with which $g_1, g_2$ have correlated agreement, and denote by $U\subseteq \mathbb{F}_q^2$ the set of points they agree on. Let $\wh{G}_1, \wh{G}_2$ be their low degree extensions and consider the degree $(d, d)$ polynomial
\[
\wh{G} = \wh{V}_A \cdot \wh{G}_1 + \wh{V}_B \cdot \wh{G}_2 + \wh{h},
\]
as well as the following function over $\Ff_q^2$:
\[
g = V_A \cdot g_1 + V_B \cdot g_2 + \tilde{h}.
\]
We have that $\wh{G}$ agrees with $g$ on $U$ and $g$ agrees with $Q$ on $U \setminus \left(\Ff_q \times B\right)$.
Thus, $\wh{G}$ has individual degrees $(d, d)$, agrees with $h$ on $A \times B$, and agrees with $Q$ on $U \setminus \left(\Ff_q \times B\right)$. Noting that
\[
\mu(U) - \mu\left(\Ff_q \times B\right) > \eps + \frac{|B|}{q} - \frac{|B|}{q} \geq \eps,
\]
we get a contradiction to the assumption that $\agr(Q, \RM_{q'}[Q, \Ff_q^2 \; | \; h]) \leq \eps$. 
\end{proof}

As previously mentioned, the purpose of \cref{lm: bivariate side condition decomp} is to show that if $Q$ is far from an individual Reed-Muller code with side conditions, then for any $g_1$ and $\Fill$ provided, the functions $g_1$ and $g_2$ cannot have high correlated agreement with the some low individual degrees. In order for this to be useful however, we would like to combine $g_1$ and $g_2$ into a single function, and show that their lack of correlated agreement results in a combined function that is far from low individual degree. This is achieved via \cref{lm: combine multiple multivariate}: 

\begin{lemma} \label{lm: bivariate side condition decomp+combine}
    Suppose that $\agr(f, \RM_{q'}[(d,d), \Ff_q^2 \; | \; h]) \leq \eps$ for $\eps \geq 21 \left(\frac{d}{q}\right)^{\tl}$, and let $|B| \leq q^{1/4}$. Fix any $g_1: \Ff_q^2 \to \Ff_{q'}$ and $\Fill: \Ff_{q} \times B \to \Ff_{q'}$, and take 
    $g_2 = \quo_2(f - V_A \cdot g_1 - \tilde{h}, B, \Fill)$
    where $\tilde{h}=\hat{h}|_{\mathbb{F}_q^2}$ and $\wh{h}$ is the low-degree extension of $h$ over $\mathbb{F}_{q'}^2$. Then we have
    \[
    \Pr\left[\agr(\combine_{(d,d)}(g_1, g_2), \RM_{q'}[(d,d), \Ff_q^2]) \geq 1.04\eps \right] \leq \frac{\poly(d)}{q'},
    \]
    where the randomness is over the random choice of proximity generator coefficients in $\combine$.
\end{lemma}
\begin{proof}
    Fix any $g_1: \Ff_q^2 \to \Ff_{q'}$ and $\Fill: \Ff_q \times B \to \Ff_{q'}$. By \cref{lm: bivariate side condition decomp}, $g_1$ and $g_2$ can have correlated agreement at most $\eps + \frac{|B|}{q}$ with degrees $(d-|A|, d)$ and $(d, d-|B|)$ respectively. The result then follows from \cref{lm: combine multiple multivariate}.
\end{proof}

A similar result holds for the total degree Reed-Muller code.

\begin{lemma} \label{lm: total deg bivariate side condition decomp+combine}
    Suppose that $\agr(f, \RM_{q'}[d, \Ff_q^2 \; | \; h]) \leq \eps$ for $\eps \geq 22 \left(\frac{d}{q}\right)^{\tl}$, and let $|B| \leq q^{1/4}$. Fix any $g_1: \Ff_q^2 \to \Ff_{q'}$ and $\Fill: \Ff_{q} \times B \to \Ff_{q'}$, and take $g_2 = \quo_2(f - V_A \cdot g_1 - \tilde{h}, B, \Fill)$ where $\tilde{h}=\hat{h}|_{\mathbb{F}_q^2}$ and $\wh{h}$ is the low-degree extension of $h$ over $\mathbb{F}_{q'}^2$ 
    Then we have that 
    \[
    \Pr\left[\agr(\combine_{d}(g_1, g_2), \RM_{q'}[d, \Ff_q^2]) \geq 1.04\eps \right] \leq \frac{\poly(d)}{q'},
    \]
    where the randomness is over the random choice of proximity generator coefficients in $\combine$.
\end{lemma}
\begin{proof}
    The proof is nearly identical to the proof of \cref{lm: bivariate side condition decomp+combine}, and is hence omitted.
\end{proof}

\subsection{Trivariate Quotienting}
Likewise, quotienting can be extended to three variables in the natural way.
\begin{definition}
    Given an oracle function $f: \Ff_q^3 \to \Ff_{q'}$, a set $C \subseteq \Ff_{q'}$, and a function $\Fill: \Ff_q \times \Ff_q \times C \to \Ff_{q'}$, define the following function from $\Ff_q^3 \to \Ff_{q'}$
    \begin{equation} \label{eq: 3var multi quotient def}
\quo_3(f, B, \Fill)(x,y,z) = \begin{cases}
    f(x,y,z) \quad \text{if } z\notin C \\
    \Fill(x,y,z) \quad \text{if } z \in C.
\end{cases}
\end{equation}
\end{definition}

If a function $\wh{f} \in \Ff_{q'}[x,y,z]$ agrees with some side condition $h: A \times B \times C \to \Ff_{q'}$, then by~\cref{lm: multivariate side condition vanish decomp} there exist functions $\wh{g}_1, \wh{g}_2, \wh{g}_3 \in \Ff_{q'}[x,y,z]$ such that
\[
\wh{f} = \wh{V}_{A, 1} \cdot \wh{g}_1 + \wh{V}_{B, 2}\cdot \wh{g}_2 + V_{C,3}\cdot \wh{g}_3 + \wh{h}. 
\]
Using this decomposition, one can similarly define a combine function as follows and obtain an analogue of \cref{lm: total deg bivariate side condition decomp+combine} for the total degree trivariate Reed-Muller code, which we state below (there are analogs for for the individual degree case, but we do not require them).

\begin{lemma} \label{lm: total deg trivariate side condition decomp+combine}
    Suppose that $\agr(f, \RM_{q'}[d, \Ff_q^3 \; | \; h]) \leq \eps$ for $\eps \geq 1.02 \left(\frac{d}{q}\right)^{\tl}$, and let $|C| \leq q^{1/4}$ and set $\tilde{h} = \wh{h}|_{\mathbb{F}_q^3}$. Fix any $g_1, g_2: \Ff_q^2 \to \Ff_{q'}$ and $\Fill: \Ff_{q} \times \Ff_q \times C \to \Ff_{q'}$, and take $g_2 = \quo_3(f - V_{A,1} \cdot g_1 - V_{B,2} \cdot g_2 - \tilde{h}, C, \Fill)$. Then we have
    \[
    \Pr\left[\agr(\combine_{d}(g_1, g_2, g_3), \RM_{q'}[d, \Ff_q^2]) \geq 1.04\eps \right] \leq \frac{\poly(d)}{q'},
    \]
    where the randomness is over the random choice of proximity generator coefficients in $\combine$.
\end{lemma}
\begin{proof}
The proof is very similar to the proof of~\cref{lm: bivariate side condition decomp+combine}, and we omit the details.
\end{proof}
\section{Poly-IOPs}\label{sec:poly-IOP}
In this section we present the Poly-IOP model, which is a version of an IOP with a promise for each prover message. 
Following ideas from~\cite{acy} we also give a transformation that turns a Poly-IOP into a legitimate IOP by using low degree IOPPs.

\subsection{The Poly-IOP model}
A Poly-IOP is an idealized model of IOPs wherein the prover's messages each round are promised to be from some prespecified error correcting code. 

\begin{definition}
    Fix a field $\Ff_{q'}$. A Poly-IOP is an IOP in which each round 
    the verifier and prover interact, and at the end the verifier decides to accept or reject.

For each round of interaction, $i$, there is a list of families $\mc{C}^{(i)}_1, \ldots, \mc{C}^{(i)}_{k_i}$, where for each $i$ and $j$ we have that $\mc{C}^{(i)}_j \subseteq \Ff_{q'}[x_1, \ldots, x_{m_{i,j}}]$. During round $i$, the prover sends polynomials, $\wh{F}^{(i)}_j$ for $j=1,\ldots, k_i$ such that $\wh{F}^{(i)}_j \in \mc{C}^{(i)}_{j}, \forall j \in [k_i]$. We consider the following complexity parameters of Poly-IOPs:
\begin{itemize}
    \item \textbf{Round complexity:} the number of rounds of interaction during the interaction phase.
    \item \textbf{Input query complexity:} the number of queries to the input oracle.
    \item \textbf{Proof query complexity:} for each polynomial sent by the prover, $\wh{F}^{(i)}_j$, we record the number of queries, $\Q^{(i)}_j$, made to $\wh{F}^{(i)}_j$, as well as the queried points, $z^{(i)}_{j,1}, \ldots, z^{(i)}_{j, \Q^{(i)}_j}$.
    \item \textbf{Completeness:} the maximum probability, over all prover strategies, that the verifier accepts in the case that the input is in the language.
    \item \textbf{Round-by-round soundness:} defined the same way as in an IOP.
\end{itemize}
\end{definition}
For our purposes, $\mc{C}^{(i)}_j$ will always be either a degree $d$ Reed-Solomon Code, a degree $d$ Reed-Muller code, or an individual degree $(d,d)$ Reed-Muller code, for some degree parameter $d$. Thanks to the good distance of these codes, the soundness of Poly-IOPs is often much easier to analyze. At a high level, this is because if the prover sends a low degree polynomial that is different than the one intended, then the verifier can easily detect this by making a few random queries to the polynomial's evaluation. 

Having said that, while Poly-IOPs themselves are not legitimate IOPs, they can be compiled into legitimate IOPs. To facilitate that we need
to take a few points into consideration:
\begin{itemize}
    \item Length: the Poly-IOP polynomials are typically over $\Ff_{q'}$ with $q' \gg n$, so to construct linear size IOPs, we cannot afford to send the evaluation of even a single polynomial.
    \item Removing the promise: In an IOP, there is no way to force the prover to send functions from some fixed family, and we must find another way to (effectively) enforce this. To resolve this, we use an IOPP testing proximity to the families $\mc{C}^{(i)}_1, \ldots, \mc{C}^{(i)}_{k_i}$ appearing in the Poly-IOP. 
\end{itemize}
In the next section we describe how to resolve these issues via a process called compilation.

\subsection{Compiling Poly-IOPs to IOPs using a Batched Code Membership IOPP}
In this section we describe our compilation result,~\cref{thm: poly iop transformation}. Stating it requires some set up, which we now give.

\paragraph{Poly-IOP for a language: $\gp(\mc{L})$.} Suppose that $\gp(\mc{L})$ is Poly-IOP for some language $\mathcal{L}$ over the field $\Ff_{q'}$ with the following guarantees:
\begin{itemize}
    \item Input: an instance $f$, which one should think of as an oracle function.
    \item Completeness: if $f \in \mc{L}$, then the honest prover makes the verifier accept with probability $1.$
    \item Round-by-round soundness: $2^{-\lambda}$.
    \item Initial State: the initial state is doomed if and only if $f \notin \mc{L}$.
     \item Round complexity: $\rd_{\poly}$.
    \item Input query complexity: $\iQ_{\poly}$
    \item Proof query complexity: for each $i \in [\rd_{\poly}]$ suppose there are $k_i$ polynomials sent in round $i$, $\{ \wh{F}^{(i)}_j\}_{j \in [k_i]}$, which are supposedly from the codes $\{\mc{C}^{(i)}_j\}_{j \in [k_i]}$. Suppose the verifier
    queries the polynomial $\wh{F}^{(i)}_j$ at the $\Q_{i,j}$ points $z^{(i)}_{j,1}, \ldots, z^{(i)}_{j, \Q_{i,j}} \in \Ff_{q'}^{m_i}$.
\end{itemize}

\paragraph{Batched IOPP: $\batch(\{\eps_{i,j}\}_{i \in [\rd], j \in [k_i]})$.} Suppose that $\batch$ is a batched code-membership IOPP that is tailored for $\gp$. By that, we mean that the codes that $\batch$ accepts and rejects depend on the $\mc{C}^{(i)}_j$'s appearing in $\gp$ as elaborated below. The codes that we will be testing are evaluations of some family of polynomials, $\mc{C} \subseteq \Ff_{q'}[x_1,\ldots, x_m]$, over a subset of the domain, $U \subseteq \Ff_{q'}^m$. To this end, we define the code
\[
\Eval\left[\mc{C}, U\right] = \{\wh{f}|_U \; | \; \wh{f} \in \mc{C} \}.
\]
For $A \subseteq \Ff_{q'}^m$ and a side condition function $h: A \to \Ff_{q'}$, we also define a version with side conditions
\[
\Eval\left[\mc{C}, U \; | \; h\right] = \{\wh{f}|_U \; | \; \wh{f} \in \mc{C}, \wh{f}|_A = h \}.
\]
For our purposes, $\mc{C}$ is always a Reed-Solomon, total degree Reed-Muller, or individual degree Reed-Muller code.  The IOPP $\batch$ has the following guarantees:
\begin{itemize}
    \item Input: For each $i \in [\rd], j \in [k_i]$:
    \begin{itemize}
         \item a proximity parameter $\eps_{i,j} \geq 2\sqrt{\delta_{i,j}}$, where 
         \[
        \delta_{i,j} := \max_{\wh{F}, \wh{G} \in \mc{C}^{(i)}_j} \Pr_{x \in H^{(i)}_j}[\wh{F}(x) = \wh{G}(x)],
         \]
        \item a function $F^{(i)}_j: H^{(i)}_j \to \Ff_{q'}$, where $H^{(i)}_j \subseteq \Ff_{q'}^{m_i},$
        \item a side condition function, $h_{i,j}: \Side_{i,j} \to \Ff_{q'}$, which maps $\Side_{i,j} = \{z^{(i)}_{j,1}, \ldots, z^{(i)}_{j, \Q_{i,j}} \}$ to values of the queries to $\wh{F}^{(i)}_j$ in the $\gp$.
         \item The family $\mc{C}^{(i)}_j$, from $\gp$.
    \end{itemize} 
    \item Completeness: if for all $i \in [\rd_{\poly}], j \in [k_i]$, we have $F^{(i)}_j \in \Eval\left[\mc{C}^{(i)}_j, H^{(i)}_j \; | \; h\right]$, then the honest prover makes the verifier accept with probability $1$.
    \item Round-by-round soundness: $2^{-\lambda}$.
    \item Initial State: the initial state is doomed if and only if for at least one $i \in [\rd]$ we have
    \[
    \agr\left(F^{(i)}_j, \Eval\left[\mc{C}^{(i)}_j, H^{(i)}_j \; | \; h\right]\right) \leq \eps_{i,j}.
    \]
    \item Round complexity: $\rd_{\bat}$
    \item Input query complexity: $\iQ_{\bat}$.
    \item Proof query complexity: $\pQ_{\bat}$.
    \item Length: $\Len_{\bat}$.
\end{itemize}
With this set up in hand, the compilation theorem now reads as follows:
\begin{thm} \label{thm: poly iop transformation}
Suppose we have $\gp(\mc{L})$ and $\batch(\{\eps_{i,j}\}_{i \in [\rd], j \in [k_i]})$ as described above. Let 
\[
\hat{\delta}_{i,j} := \max_{\wh{F}, \wh{G} \in \mc{C}^{(i)}_j} \Pr_{x \in \Ff_{q'}^{m_{i,j}}}[\wh{F}(x) = \wh{G}(x)] \quad \text{and} \quad \delta_{i,j} := \max_{\wh{F}, \wh{G} \in \mc{C}^{(i)}_j} \Pr_{x \in H^{(i)}_j}[\wh{F}(x) = \wh{G}(x)],
\]
and suppose for all $i \in [\rd], j \in [k_i]$ we have
\begin{equation} \label{eq: compile requirements}
\eps_{i,j} \geq 2 \sqrt{\delta_{i,j}} \quad \text{and} \quad \frac{\wh{\delta}_{i,j}}{\eps_{i,j}} \leq 2^{-\lambda + 5\log(k_i)}. 
\end{equation}
Then there is an IOP for the language $\mc{L}$ with the following parameters
\begin{itemize}
    \item Completeness: if $f \in \mc{L}$ then the honest prover makes the verifier accept with probability $1$.
    \item Round-by-round soundness: $2^{-\lambda}$.
    \item Initial State: the initial state is doomed if and only if $f \notin \mc{L}$.
    \item Round complexity: $\rd_{\poly} + \rd_{\bat} + O(1)$.
    \item Input query complexity: $\iQ_{\poly}$.
    \item Proof query complexity: $\iQ_{\bat} + \pQ_{\bat} + 
    \sum_{i = 1}^{\rd_{\poly}} k_i+\sum_{i=1}^{\rd_{\poly}} \sum_{j=1}^{k_i} \Q_{i,j}$.
    \item Length: $\Len_{\bat} + \sum_{i=1}^{\rd_{\poly}}\sum_{j = 1}^{k_i} |H^{(i)}_j|$.
\end{itemize}
We refer to this IOP as $\compile(\gp(\mc{L}) ,\batch(\{\eps_{i,j}\}_{i \in [\rd], j \in [k_i]}))$.
\end{thm}

\begin{remark}
    For all of our purposes in this paper, one can replace the proof query complexity in \cref{thm: poly iop transformation} with $O(\iQ_{\batch} + \pQ_{\batch})$. The reason is that the queries corresponding to the $\sum_{i = 1}^{\rd_{\poly}} k_i+\sum_{i=1}^{\rd_{\poly}} \sum_{j=1}^{k_i} \Q_{i,j}$ term above end up each leading to at least one input or proof query in the batched IOPP.
\end{remark}

\paragraph{High level overview of the compiled IOP:} we now describe the IOP in~\cref{thm: poly iop transformation}, which we call $\compile$ in short. At a high level $\compile$ has two main components: a Poly-IOP simulation phase, and a proximity test phase.

During the Poly-IOP simulation phase the prover and the verifier simulate $\gp$, except that instead of sending the entire polynomial $\wh{F}^{(i)}_j$ (which would be too large), the prover only sends its evaluation over some smaller domain $H^{(i)}_j$. The prover still answers queries over $\Ff_{q'}^{m_{i,j}}$, however. One key addition the so called ``anchoring'' trick: as soon as the prover sends such an evaluation, the verifier performs an out-of-domain sample, asking for the supposed value of $\wh{F}^{(i)}_j(z^{(i)})$, for a randomly chosen $z^{(i)} \in \Ff_{q'}^{m_i}$.  
The reason for that is that, with high probability, this out-of-domain sample anchors the prover, in the sense there will only be one member of $\mc{C}^{(i)}_j$ which has agreement at least $\eps_{i,j}$ with $F^{(i)}_j$ that additionally agrees with the prover's response to the out-of-domain sample. We denote this function by $\wh{G}^{(i)}_j$. 

From this point on, the prover has two options: they can either answer all of the queries in $\gp$ to $\wh{F}^{(i)}_j$ according to $\wh{G}^{(i)}_j$, or they can deviate from $\wh{G}^{(i)}_j$ (i.e.\ provide at least one answer that is inconsistent with $\wh{G}^{(i)}_j$). In the former case, note that the verifier will reject by the soundness of $\gp$ and the round-by-round soundness of the compiled IOP follows from that of $\gp$.

Now suppose it is the latter case that occurs. This is where the second phase comes into play. In that case, note that the oracle $F^{(i)}_j$ does not have $\eps_{i,j}$ agreement with \emph{any} member of $\mc{C}^{(i)}_j$ which also agrees with the side condition $h_{i,j}$. This is because the side conditions include both the out-of-domain sample and the point where the prover deviated from $\wh{G}^{(i)}_j$. In this case $\batch(\{\eps_{i,j}\}_{i \in [\rd], j \in [k_i]})$ will reject and the round-by-round soundness of the compiled IOP then follows from that of $\batch$.

\paragraph{Formal description of the compiled IOP:} formally, the IOP $\compile$ proceeds as follows:
\paragraph{Poly IOP Compilation: $\compile(\gp, \batch(\{\eps_{i,j}\}_{i \in [\rd], j \in [k_i]}))$}
\begin{enumerate}
    \item \textbf{Poly-IOP Simulation Phase:}
    \begin{enumerate} 
        \item For $i \in [\rd]$:
        \begin{enumerate}
            \item \textbf{P: } The prover sends the oracle functions $F^{(i)}_j: H^{(i)}_j \to \Ff_{q'}$, for $j \in [k_i]$. In the honest case, each $F^{(i)}_j$ is the evaluation of the polynomial $\wh{F}^{(i)}_j \in \mc{C}^{(i)}_j$ which the honest prover in $\gp$ would have sent.
            \item \textbf{V: } The verifier chooses points $z^{(i)}_{j} \in \Ff^{m_{i,j}}_{q'}$  uniformly at random for each $j \in [k_i]$ and sends all of these points to the prover.
            \item \textbf{P: } The prover sends $\zeta^{(i)}_1, \ldots, \zeta^{(i)}_{k_i} \in \Ff_{q'}$. In the honest case, these values agree with the polynomials that honest prover in $\gp$ would have sent:
            \[
            \zeta^{(i)}_j = \wh{F}^{(i)}_j\left(z^{(i)}_j\right), \; \forall j \in [k_i].
            \]
            \item \textbf{V: } The verifier generates the set of queries that the verifier in $\gp$ would have made during round $i$ of $\gp$. They do so using the same randomness as the verifier from $\gp$. For each $i' \leq i$ and each $j \in [k_{i'}]$, the verifier sends the prover each point that the verifier from $\gp$ would have queried $\wh{F}^{(i')}_{ j}$ during round $i$ of $\gp$. The verifier also sends the message that the verifier from $\gp$ would have sent during round $i$ of $\gp$.
            \item \textbf{P: } The prover sends a field element of $\Ff_{q'}$ in response to each query.
            \item Both parties proceed to the next round of the Poly-IOP.
        \end{enumerate}
        
    \end{enumerate}

    \item \textbf{Decision and Test Phase:}
     \begin{enumerate}
     \item \textbf{V:} The verifier simulates the verifier from $\gp$ and rejects if the verifier from $\gp$ would have. If this is the case, then the protocol terminates.  Otherwise, both parties proceed to the next step.

    During the simulation, for each query to $\wh{F}^{(i)}_j$ made by the verifier from $\gp$, the current verifier acts as if the answer is what the prover provided during step iii above. Using these answers, the current verifier can then simulate the decision of the verifier from $\gp$. 
        \item For each $i \in [\rd], j \in [k_i]$ let $\Side_{i,j} \subseteq \Ff_{q'}^m$ consist of the points where the verifier queries the $\wh{F}_{i}$ during step iv above, as well as the point $z^{(i)}_j$. Let 
        \[
        h_{i,j}: \Side_{i,j} \to \Ff_{q'}
        \]
        be the map consisting of the prover's responses given in step v of the Poly-IOP interaction phase.
        \item \textbf{P + V:} Both parties run $\batch(\{\eps_{i,j}\}_{i \in [\rd], j \in [k_i]})$ with the following inputs for each $i \in [\rd], j \in [k_i]$,
        \begin{itemize}
            \item Functions $F^{(i)}_j$,
            \item Side condition functions $h_{i,j}$,
            \item Function Family $\mc{C}^{(i)}_j$.
        \end{itemize}
    \end{enumerate}
\end{enumerate}

\subsection{Analysis of $\compile$: the proof of \cref{thm: poly iop transformation}}
We analyze the parameters of $\compile$ below. The analysis of the round-by-round soundness is the most involved part, and is deferred to~\cref{sec:rbr_compile}. The round complexity, input query complexity, proof query complexity, and length are easy to verify, and we next argue about the completeness.

In the completeness case, there is a strategy for the honest prover in $\gp$ which will make the verifier in $\gp$ accept with probability $1$. Let $\wh{F}^{(i)}_{j} \in \mc{C}^{(i)}_j$ be the polynomials that the honest prover in $\gp$ sends. The honest prover in $\compile$ 
will first simulate the honest prover in $\gp$ during the Poly-IOP Simulation Phase. Specifically, they send $F^{(i)}_j = \wh{F}^{(i)}_j|_{H^{(i)}_j}$ agreeing with $\wh{F}^{(i)}_j$ for each $i,j$ and they answer all queries in steps ii and iv according to the polynomials $\wh{F}^{(i)}_j$. By the completeness of $\gp$, the verifier in $\compile$ will not reject. As for the ``Test Phase'', we have $F^{(i)}_j \in \Eval[\mc{C}^{(i)}_j, H^{(i)}_j \; | \; h_{i,j}]$ for all $i,j$, and by the completeness case of $\batch(\{\eps_{i,j}\}_{i \in [\rd], j \in [k_i]})$ the verifier always accepts.

\subsection{Round-By-Round Soundness for \cref{thm: poly iop transformation}}\label{sec:rbr_compile}
To show the round-by-round soundness for $\compile$ we will go over each round of interaction and define the doomed states.

\subsubsection{The Poly-IOP Interaction Phase}
We fix a round $i$. For each $i' \leq i, j' \in [k_{i'}]$, let $h^{(i)}_{i', j'}$ be side condition function for the polynomial $\wh{F}^{(i')}_{j'}$ after round $i$. That is, the domain of $h_{i', j'}$ consists of all of the points where the verifier queries $\wh{F}^{(i')}_{j'}$ and $h^{(i)}_{i', j'}$  maps these points to the answers provided by the prover (during step v of the Poly-IOP Simulation Phase). Let $\mc{F}^{(i')}_{j'}$ consist of all $\wh{F} \in \mc{C}^{(i')}_{j'}$ that satisfy 
\begin{itemize}
    \item $\agr\left(F^{(i')}_{j'}, \wh{F}|_{H^{(i')}_{j'}}\right) \geq \eps_{i',j'}$.
    \item $\wh{F}$ agrees with $h^{(i)}_{i', j'}$.
\end{itemize}

During round $i$ of the Poly-IOP interaction phase, there are three rounds of interaction. We will go through each of them, define the state function afterwards, and show soundness for that round.

\paragraph{Interaction 1:}
\begin{itemize}
\item \textbf{P: } The prover send the function $F^{(i)}_j: H^{(i)}_j \to \Ff_{q'}$.
\item \textbf{V: } The verifier sends $z^{(i)}_j \in \Ff^{m_{i,j}}_{q'}$.
\end{itemize}

\begin{state} \label{state: poly iop simulation 1}
    The state is doomed if and only if the previous state was doomed and the following holds for every $j \in [k_i]$. There do not exist distinct $\wh{F}, \wh{F}' \in \mc{C}^{(i)}_j$ satisfying:
    \begin{itemize}
        \item $\wh{F}(z^{(i)}_j) =  \wh{F}'(z^{(i)}_j)$,
        \item the evaluations of $\wh{F}, \wh{F}'$ over $H^{(i)}_j$ both have agreement at least $\eps_{i, j}$ with the function $F^{(i)}_j$.
    \end{itemize} 
In other words, the anchoring step was successful and there is at most one member of $\mc{C}^{(i)}_j$ agreeing with the prover's answer to $z^{(i)}_j$ and agreeing non-trivially with the evaluation over $H^{(i)}_j$ that the prover provided. Note that if the state is doomed then we have $|\fcalij| \leq 1$ for all $j \in [k_i]$.
\end{state}
\begin{lemma}
Suppose the previous state is doomed. Then for any $F^{(i)}_j$ sent by the prover, \cref{state: poly iop simulation 1} will be doomed with probability at least $1-2^{-\lambda}$ over the verifier's random choice of $z^{(i)}_1,\ldots, z^{(i)}_{k_i}$.
\end{lemma}
\begin{proof}
Assume that the previous state was doomed. For each $F^{(i)}_j$ sent by the prover, let 
\[
\List_{\eps_{i,j}}(F^{(i)}_j) = \left\{\wh{F} \in \mc{C}^{(i)}_j \; | \; \agr(F^{(i)}_j, \wh{F}|_{H^{(i)}_j}) \geq \eps_{i,j} \right\}.
\]
The state will be doomed after the verifier's randomness if and only if for every $j$, no two members of $\List_{\eps_{i,j}}(F^{(i)}_j)$ agree on $z^{(i)}_j$.
By \cref{thm: list decoding}, we have that $\left|\List_{\eps_{i,j}}(F^{(i)}_j)\right| \leq \frac{\eps_{i,j}}{\eps_{i,j}^2 - \delta_{i,j}}\leq \frac{\sqrt{2}}{\eps_{i,j}}$. The probability that any two polynomials in $\List_{\eps_{i,j}}(F^{(i)}_j)$ agree on a randomly chosen $z^{(i)}_j \in \Ff_{q'}^{m_{i,j}}$ is at most $\wh{\delta}_{i,j}$, so by a union bound with probability at least
\[
1 - \wh{\delta}_{i,j} \cdot\left|\List_{\eps_{i,j}}(F^{(i)}_j)\right|^2 \geq 1- \wh{\delta}_{i,j} \cdot \frac{2}{\eps^2_{i,j}} \geq 1 - \frac{1}{k_i^5 \cdot 2^{-\lambda}}
\]
we have that no two members of the list agree on $z^{(i)}_j$. In order for the state to be doomed, we need this to hold for all $j \in [k_i]$. Thus, performing another union bound over $j \in [k_i]$, we get that the state is doomed with probability at least $1-  \frac{k_i}{k_i^5 \cdot 2^{-\lambda}} \geq 1 - 2^{-\lambda}$ as desired. 
\end{proof}

The second interaction and state function are as follows.

\paragraph{Interaction 2}

\begin{itemize}
\item \textbf{P: } For every $j \in [k_i]$, the prover sends a field element $\zeta^{(i)}_j$. In the honest case $\zeta^{(i)}_j = \wh{F}^{(i)}_j(z^{(i)}_j)$.
\item \textbf{V: } The verifier sends the prover a list of all points they query $\wh{F}_j$ at for all $j \leq i$. The verifier also sends their Poly-IOP message for the next round, if there is such a round.
\end{itemize}

While defining the next state function, let us assume that \cref{state: poly iop simulation 1} is doomed. Indeed, this is the only case that matters as far as round-by-round soundness is concerned. If \cref{state: poly iop simulation 1} is doomed, then for each $j \in [k_i]$, after the prover sends $\zeta^{(i)}_j$ in the above interaction, there is at most one member of $\mc{C}^{(i)}_j$ which evaluates to $\zeta^{(i)}_j$ on $z^{(i)}_j$ and has at least $\eps_{i,j}$ agreement with $\fij$ over $\hij$. In this case, $|\fcalij| \leq 1$, and if $|\fcalij| = 1$, let $\fhatij$ be the unique member in $\fcalij$.

\begin{state} \label{state: poly iop simulation 2}
    After the verifier's message, we call the state doomed if at least one of the following holds:
\begin{itemize}
    \item For each $i' \leq i$ and each $j \in [k_{i'}]$, $|\mc{F}_j^{(i')}| = 1$ and the corresponding state in $\gp$ is doomed. More specifically, in this case, the prover has consistently responded according to a single polynomial $\fhatipj \in \cipj$ for every round $i' \leq i$ thus far and every polynomial $j \in [k_{i'}]$ of that round. We define the state as doomed if and only if $\gp$ would be doomed assuming that the prover in $\gp$ answered with the polynomials $\fhatipj$ for each round $i' \leq i$ thus far and each $j \in [k_{i'}]$ of that round. 
    \item For some $i' \leq i, j \in [k_{i'}]$, $|\mc{F}_j^{(i')}| = 0$. In this case, there is at least one $i', j$ for which the prover has not answered consistently according to a single member of $\cipj$.
\end{itemize}
\end{state}
The soundness of this round is discussed in the next lemma. It essentially follows from the soundness of the corresponding round in $\gp$.
\begin{lemma}
    If the previous state was doomed, then the next state is doomed with probability at least $1 - 2^{-\lambda}$. 
\end{lemma}
\begin{proof}
As we assume the previous state was doomed, we have that for all $i' < i, j \in [k_{i'}]$, $|\mc{F}_j^{(i')}| \leq 1$. If for any of these $i', j$ we have $|\mc{F}_j^{(i')}| = 0$, then the next state is automatically doomed and we are done, so suppose that $|\mc{F}_j^{(i')}| = 1$ for all of these $i', j$.

Since \cref{state: poly iop simulation 1} is doomed, it follows that no matter what $\zeta^{(i)}_j$'s the prover sends, we will have $|\fcalij| \leq 1$ for all $j \in [k_i]$. Once again, if for any $j \in [k_i]$ we have $|\fcalij| = 0$, then the next state is automatically doomed and we are done. Thus, assume  that $|\fcalij| = 1$ for all $j \in [k_i]$ henceforth.

In this case we are in the setting of item 1 of \cref{state: poly iop simulation 2} and \cref{state: poly iop simulation 2} is doomed with probability at least $1-2^{-\lambda}$ by the round-by-round soundness of $\gp$.
\end{proof}

The final interaction and state function of the Poly-IOP simulation phase are the following.

\paragraph{Interaction 3:}
\begin{itemize}
     \item \textbf{P: } The prover responds with a field element in $\Ff_{q'}$ for each such query.
     \item \textbf{V: } The verifier sends an empty message.
\end{itemize}
\begin{state} \label{state: poly iop simulation 3}
    The state is doomed if and only if at least one of the following holds
    \begin{itemize}
        \item For all $i' \leq i, j \in [k_{i'}]$, we have $|\fcalpij| = 1$, and the corresponding state in $\gp$ is doomed. In this case there is exactly one $\wh{F}^{(i')}_j \in \cipj$ consistent the the prover's answers so far, and the prover has continued to answer according to this $\wh{F}^{(i')}_j$ in the present interaction.
        \item For some $i' \leq i, j \in [k_{i'}]$ we have $|\fcalpij| = 0$. In this case there is some $i', j$ where the prover has not consistently answered according to a single  $\wh{F}^{(i')}_j \in \cipj$.
    \end{itemize}
\end{state}

In fact, regardless of how the prover responds here, the next state will be automatically doomed if the previous state was doomed. Informally, this is because the round-by-round soundness was handled by the verifier's random choice of queries in the previous round. The simulation phase is designed so that once the verifier decides on the queries to the polynomials $\fhatij$'s, we have soundness regardless of how the prover decides to respond to these queries.

\begin{lemma}
    If the previous state was doomed, then the next state is doomed with probability $1$.
\end{lemma}
\begin{proof}
We assume that the previous state is doomed. If for some $i' \leq i, j \in [k_{i'}]$, we have $|\fcalpij| = 0$, then we are done. Otherwise, the prover has responded with values that agree with the unique $\fhatipj\in \fcalpij$ for each $i', j$. In this case the second item of \cref{state: poly iop simulation 3} is satisfied and the state is doomed.
\end{proof}

\subsubsection{The Decision and Test Phase}
The remainder of the round-by-round soundness follows from that of $\batch(\{\eps_{i,j}\}_{i \in [\rd], j \in [k_i]})$. The last round of interaction before the parties run $\batch$ is the following.

\paragraph{Interaction 1:}
\begin{itemize}
     \item \textbf{V: } The verifier simulates the verifier from $\gp$. For each query to $\wh{F}^{(i)}_j$ made by the verifier from $\gp$, the current verifier acts as if the answer is that provided by the current prover during step iii of the Poly-IOP Simulation phase.
\end{itemize}
After this round, the state function is as follows. 

\begin{state} \label{state: decision state}
The state is doomed if and only if one of the following holds
\begin{itemize}
    \item The verifier rejects and terminates the protocol.
    \item For some round $i$ and $j \in [k_i]$, we have $|\fcalij| = 0$.
\end{itemize}
\end{state}

\begin{lemma}
    If the previous round is doomed, then \cref{state: decision state} is doomed with probability at least $1-2^{-\lambda}$.
\end{lemma}
\begin{proof}
    If for some round $i$ and $j \in [k_i]$, we have $|\fcalij| = 0$, then the state is automatically doomed and we are done. Otherwise, we have $|\fcalij| = 1$ for all rounds $i$ and $j \in [k_i]$. In this case, the verifier will simulate the decision of the verifier from $\gp$ assuming that the prover in $\gp$ sent $\fhatij$ for all rounds $i$ and $j \in [k_i]$. The soundness of this round then follows from the round-by-round soundness of $\gp$.
\end{proof}

In the remainder of $\compile$, both parties run the IOPP $\batch(\{\eps_{i,j}\}_{i \in [\rd], j \in [k_i]})$, and the round-by-round soundness therein follows from the assumed round-by-round soundness of $\batch$. Indeed, one can check that if the state is doomed going into step c of the ``Decision and Test phase'' of $\compile(\gp, \batch(\{\eps_{i,j}\}_{i \in [\rd], j \in [k_i]}))$, then one of the following holds: 
    \begin{itemize}
        \item The first item of \cref{state: decision state} holds and the protocol is terminated. 
        \item The second item of \cref{state: decision state} holds and the input to $\batch(\{\eps_{i,j}\}_{i \in [\rd], j \in [k_i]})$ is such that the initial state of $\batch$ doomed.
    \end{itemize} 
In the first case, the protocol has terminated and there are no more rounds to analyze. In the second case, the remaining rounds have soundness error at most $2^{-\lambda}$ due to the round-by-round soundness of $\batch$.

\section{IOPPs in the Low Rate Regime} \label{sec: low rate}


In this section, we construct IOPs of Proximity for the degree $d$ Reed-Solomon and individual degree $(d,d)$ Reed-Muller codes in the low rate regime. By low rate, we mean that the degree parameter $d$ is polynomially smaller than the field size $q$.
These low rate IOPPs have size that is roughly quadratic in $d$, making them unsuitable for our final application. Instead, they are used after we have already reduced the problem size by a square-root factor.

\subsection{Low Rate Reed-Solomon IOPP}
Fix a security parameter $\lambda$ and a sequence of degrees of the form $2^{2^{k}}$ for $k = 1, \ldots, M$. Fix two finite fields $\Ff_q, \Ff_{q'}$ such that $\Ff_q$ is a subfield of $\Ff_{q'}$ and
\[
q > 2^{2^M+1}, \quad \quad q' = 2^{\lambda} \cdot \poly(q).
\]
Furthermore, for each $k \in [M]$, set 
\[
\eps_k = \left(\frac{2^{2^k+1}}{q}\right)^{\tl / 2}, \quad\quad T_k = \Theta\left(\frac{\lambda}{\log(1/\eps_k)} \right).
\]
Notice that for every $k \in [M]$, we have $T_k = O(\lambda/\log(q))$. With these parameters in mind, our IOPP for the Reed-Solomon code in the low-rate regime reads as follows:

\begin{thm}[Low Rate Reed-Solomon]\label{thm: RS low rate}
Fix a degree parameter $d = 2^{2^k}$ and let $\lambda, q, q'$ be as above and set an agreement parameter $\eps \geq \eps_k$ be as above. Then there is an IOPP to the Reed-Solomon code, which we refer to as $\rsi(f, d, q, \lambda, \eps)$, that has the following guarantees:
\begin{itemize}
    \item {\bf Input:} a function $f: \Ff_q \to \Ff_{q'}$.
    \item {\bf Completeness:} if $f \in \RS_{q'}[d, \Ff_{q}]$, then the honest prover makes the verifier accept with probability $1$.
    \item {\bf Round-by-round soundness:} $2^{-\lambda}$.
    \item {\bf Initial state:} the initial state is doomed if and only if $\agr(f, \RS_{q'}[d, \Ff_{q}]) \leq \eps$.
    \item {\bf Round complexity:} $O(k) = O(\log \log d)$.
    \item {\bf Input query complexity:}  
    \[
    \iQ_{\RS, k} = O\left(\frac{\lambda}{\log(1/\eps_k)}\right).
    \]
    \item {\bf Proof query complexity:} 
    \[
    \pQ_{\RS, k} = O\left(\sum_{i=1}^{k-1}  \frac{\lambda^2}{\log^2(1/\eps_i)} \right).
    \]
    \item {\bf Length:} $O(q^2\cdot k) = O(q^2 \log \log d)$.
\end{itemize}
\end{thm}
\begin{proof}
    Deferred to~\cref{sec: batch low rate}, after we establish preliminary results in~\cref{sec:pf_rs_lowrate_iopp,sec:low_rate_rm_ind}.
\end{proof}



Our low-rate IOPP for individual degree Reed-Muller codes can be stated as follows:
\begin{thm} (Low Rate Individual Degree Reed-Muller) \label{thm: ind RM low rate}
Fix a degree parameter $d = 2^{2^k}$ and let $\lambda, q, q'$ be as above and set an agreement parameter $\eps \geq \eps_k$ be as above. Then there is an IOPP to the Reed-Muller code with individual degrees $d,d$, which we refer to as $\irmi(f,(d,d),q, \lambda, \eps_k)$, that has the following guarantees:
\begin{itemize}
    \item {\bf Input:} a function $f: \Ff^2_q \to \Ff_{q'}$.
    \item {\bf Completeness:} if $f \in \RM_{q'}[(d,d), \Ff^2_{q}]$, then the honest prover makes the verifier accept with probability $1$.
    \item {\bf Round-by-round soundness:} $2^{-\lambda}$.
    \item {\bf Initial state:} the initial state is doomed if and only if $\agr(f, \RM_{q'}[(d,d), \Ff^2_{q}]) \leq \eps$.
    \item {\bf Round complexity:} $O(k) = O(\log \log d)$
    \item {\bf Input query complexity:} 
    \[
    \iQ_{\iRM, k} = O\left(\frac{\lambda}{\log(1/\eps_k)}\right).
    \]
    \item {\bf Proof query complexity:} 
    \[
    \pQ_{\iRM, k} = O\left(\sum_{i=1}^{k-1}  \frac{\lambda^2}{\log^2(1/\eps_i)} \right).
    \]
    \item {\bf Length:} $O(q^2\cdot k) = O(q^2 \log \log d)$.
\end{itemize}
\end{thm}
\begin{proof}
    Deferred to~\cref{sec: batch low rate}, after we establish preliminary results in~\cref{sec:pf_rs_lowrate_iopp,sec:low_rate_rm_ind}.
\end{proof}

\paragraph{The proofs of~\cref{thm: RS low rate} and ~\cref{thm: ind RM low rate}:} we prove both results together by induction on $k$. Namely, to construct the IOPP $\rsi(f, d_{k}, q, \lambda, \eps_{k})$ and $\irmi(f, d_{k}, q, \lambda, \eps_{k})$ we will assume that we have access to the IOPPs $\rsi(f, d_{k-1}, q, \lambda, \eps_{k-1})$ and $\irmi(f, d_{k-1}, q, \lambda, \eps_{k-1})$ as per the statements of~\cref{thm: RS low rate},~\cref{thm: ind RM low rate} respectively. Towards this end, in~\cref{sec:pf_rs_lowrate_iopp,sec:low_rate_rm_ind}
we first construct appropriate Poly-IOPs. Then, in~\cref{sec: batch low rate} we compile them into proper IOPPs using~\cref{thm: poly iop transformation}. The inductive hypothesis is used in the compilation step.

\subsection{Low Rate Reed-Solomon Poly-IOPP}\label{sec:pf_rs_lowrate_iopp}
Towards \cref{thm: RS low rate}, we start by constructing a Poly-IOPP for testing proximity to $\RS_{q'}[d, \Ff_q]$. 

\begin{lemma} \label{lm: poly iop rs}
With the notation of \cref{thm: RS low rate}, there is a Poly-IOPP with the following guarantee:
\begin{itemize}
        \item {\bf Input:} a function $f: \Ff_q \to \Ff_{q'}$.
        \item {\bf Completeness:} if $f \in \RS_{q'}[d, \Ff_q]$, then the honest prover makes the verifier accept with probability $1$.
        \item {\bf Round-by-round soundness:} $2^{-\lambda}$.
        \item {\bf Initial state:} the initial state is doomed if and only if $\agr(f, \RS_{q'}[d, \Ff_q]) \leq \eps_k$.
        \item {\bf Number of polynomials:} $1$.
        \item {\bf Input query complexity:} $T_k$.
        \item {\bf Proof query complexity:} $T_k$.
    \end{itemize}
\end{lemma}
\begin{proof}
We begin by describing the Poly-IOPP.
\begin{algorithm}[H]  \caption{Reed-Solomon Poly IOP: $\rsp(f, d, q, \lambda, \eps_k)$} \label{poly iop RS}
    \begin{algorithmic}[1]
        \STATE \textbf{P:} The prover sends a polynomial $\widehat{Q} \in \Ff_{q'}^{\leq(\sqrt{d}, \sqrt{d})}[x,y]$. In the honest case,
    \[
    \widehat{Q}(x^{\sqrt{d}}, x) = \wh{f}(x),
    \]
    where $\wh{f} \in \Ff_{q'}[x]$ is the low degree extension of $f$. Let $S = \{(\alpha^k, \alpha) \; | \; \alpha \in \Ff_q \}$.
    \STATE \textbf{V:} The verifier chooses $T_k$-points, $(\alpha_{i}^{\sqrt{d}}, \alpha_i) \in S, \forall i \in [T_k]$, uniformly at random and checks for each $i \in [T_k]$ if
        \[
        \wh{Q}(\alpha_i^{\sqrt{d}}, \alpha_i) = f(\alpha_i) 
        \]
        If any check fails then the verifier rejects. Otherwise the verifier accepts.
    \end{algorithmic}
\end{algorithm}

    We show that $\rsp(f, d, \lambda, \eps_k)$ satisfies the guarantees of \cref{lm: poly iop rs}. The round complexity, input query complexity, and proof complexity are straightforward to check. 

    To see the completeness, let $\wh{f} \in \Ff_{q'}[x]$ denote the low degree extension of $f$. If $f \in \RS_{q'}[d, \Ff_q]$, then $\deg(\wh{f}) \leq d$. The honest prover can send $\wh{Q} \in \Ff_{q'}^{\leq (\sqrt{d}, \sqrt{d})}[x,y]$ according to \cref{lem:split_rs_rm} which satisfies $\wh{Q}(x^{\sqrt{d}}, x) = \wh{f}(x)$. In this case it is clear that the verifier will accept with probability $1$.

    For the round-by-round soundness, there is only one round of interaction. We define the state as doomed after this round of interaction if and only if one of the verifier's checks fails and the verifier rejects. If $\agr(f, \RS_{q'}[d, \Ff_q]) \leq \eps_{k}$, then for any $\wh{Q}\in \Ff_{q'}^{\leq(\sqrt{d}, \sqrt{d})}$ that the prover sends, note that
    \[
 \Pr_{(\alpha^{\sqrt{d}}, \alpha) \in S}[f(\alpha) = \wh{Q}(\alpha^{\sqrt{d}}, \alpha)] \leq \eps_k.
\]
Indeed, otherwise $f$ would have agreement greater than $\eps_k$ with the degree $d$ function $\wh{Q}(x^{\sqrt{d}}, x)$ over $\Ff_q$. It follows that the verifier rejects with probability at least $1 - \eps_k^{T_k} \geq 1 - 2^{-\lambda}$. The round-by-round sondness foloows.
\end{proof}

\subsection{Low Rate Individual Degree Reed-Muller IOPP}\label{sec:low_rate_rm_ind}
Towards \cref{thm: ind RM low rate}, we construct a Poly-IOP for testing proximity to $\RM_{q'}[(d,d), \Ff_q]$. In the Poly-IOP the prover will send a two univariate polynomials with degree at most $d$ and one bivariate polynomial with total degree $2d$. We will ultimately compile this Poly-IOP into a suitable batched code membership IOP which relies on the degree $d$-Reed Solomon IOPP.

\begin{lemma} \label{lm: poly iop ind rm}
    There is a Poly-IOPP for proximity to the $(d,d)$-degree Reed-Muller code with the following guarantees:
    \begin{itemize}
        \item \textbf{Input:} a function $f: \Ff_{q}^2 \to \Ff_{q'}$.
        \item \textbf{Completeness}: if $f \in \RM_{q'}[(d,d), \Ff_q]$, then the honest prover makes the verifier
        \item \textbf{Round-by-round soundness}: $2^{-\lambda}$
        \item \textbf{Initial state:} the initial state is doomed if and only if $\agr(f, \RM_{q'}[(d,d), \Ff_q^2]) \leq \eps_k$.
        \item \textbf{Round complexity:} $5$.
        \item \textbf{Input query complexity:} $T_k$.
        \item \textbf{Proof query complexity:} $T_k+4$.
    \end{itemize}
\end{lemma}

\begin{proof}
We begin with a formal description of the Poly-IOPP:
\begin{algorithm}[H]\label{poly iop: irm}
  \caption{Individual Degree Reed-Muller IOP: $\irmp(f, (d,d), q, \lambda, \eps_k)$}
  \begin{algorithmic}[1]
    \STATE \textbf{P:} The prover sends $\widehat{Q} \in \Ff^{\leq 2d}_{q'}[x,y]$. In the honest case, $\wh{Q}$ is the low degree extension of $f$.
    \STATE \textbf{V:} The verifier chooses $T_k$ points $a_1, \ldots, a_{T_k} \in \Ff_q^2$ uniformly at random and checks
    \[
    \widehat{Q}(a_i) = f(a_{i}), \quad  \forall i \in [T_k].
    \]
    If any check fails then the verifier rejects.
    \STATE \textbf{V:} The verifier chooses $\alpha, \beta \in \Ff_{q'}$ uniformly at random and sends these to the prover.
    \STATE \textbf{P:} The prover sends $\widehat{F}_{1}, \widehat{F}_{2} \in \Ff_{q'}^{\leq d}[x]$. In the honest case,
    \[
    \widehat{F}_{1}(x) = \widehat{Q}(x, \beta) \quad \text{and} \quad  \widehat{F}_{2}(y) = \widehat{Q}(\alpha, y).
    \]
    \STATE \textbf{V:} The verifier chooses $\zeta_1, \zeta_2 \in \Ff_{q'}$ uniformly at random and checks if 
    \[
    \widehat{F}_{1}(\zeta_1) = \widehat{Q}(\zeta_1, \beta) \quad \text{and} \quad \widehat{F}_2(\zeta_2) = \widehat{Q}(\alpha, \zeta_2).
    \]
    If either check fails, then the verifier rejects. Otherwise, the verifier accepts.
  \end{algorithmic}
\end{algorithm}
At a high level 
the Poly-IOP is doing two checks to handle two separate cases in the soundness case. First, the prover may send $\wh{Q}$ which as total degree $2d$ and individual degrees at most $(d,d)$. In this case, $\wh{Q}$ should disagree with $f$ on at most $\eps_k$-fraction of $\Ff_q^2$, and the verifier will catch this with high probability in step 2. In the second case, the prover may send $\wh{Q}$ has one variable with degree between $d+1$ and $2d$. This case is handled by steps 3 through 5. The verifier catches the prover here by randomly fixing each of the variables and comparing the resulting univariate polynomials (in $x$ and $y$ respectively) $\wh{Q}(x, \beta)$ and $\wh{Q}(\alpha, y)$ to degree $d$ functions. Since one of these functions should have degree between $d+1$ and $2d$, it will differ from every degree $d$ function almost everywhere, so the verifier will reject with high probability in step 5.

    We now give the formal analysis of the Poly-IOP. 
    In the completeness case, $f \in \RM_{q'}[(d,d), \Ff_q]$. Let $\wh{f}$ be its low degree extension over $\Ff_{q'}$. The honest prover sends $\wh{Q} = \wh{f}$ in the first step. In this case, it is clear that the verifier will not reject during the second step. During the fourth step, the honest prover sends $\wh{F}_1(x) = \wh{f}(x, \beta)$ and $\wh{F}_2(y) = \wh{f}(\alpha, y)$. Since the honest prover also sent $\wh{Q} = \wh{f}$, it is clear that the checks in step 5 will pass with probability $1$. Overall, this establishes the completeness case. The round complexity and query complexities are easy to see, so we now argue the round-by-round soundness. 
\paragraph{First Round of Interaction:}
    After the first round of interaction, i.e.\ after step 3, we define the state as doomed if and only if one of the following holds:
    \begin{itemize}
        \item One of the checks in step 2 fails, in which case the protocol terminates and every state afterwards is doomed.
        \item One of the univariate polynomials $\wh{Q}(x, \beta)$, $\wh{Q}(\alpha, y)$ has degree greater than $d$.
    \end{itemize}
    If the initial state is doomed, then, $\agr(f, \RM_{q'}[(d,d), \Ff_q^2]) \leq \eps_k$. We consider two cases.

    First suppose $\wh{Q}$ has individual degrees at most $d$. Then it follows that, 
    \[
    \Pr_{a \in \Ff_q^2}[\wh{Q}(a) = f(a)] \leq \eps_k.
    \]
    In this case, with probability at least $1 - \eps_k^{T_k} \geq 1 - 2^{-\lambda}$, one of the checks in step 2 will fail.

    Now suppose $\wh{Q}$ does not have individual degrees at most $d$. Without loss of generality suppose it is the $x$-degree that is at least $d+1$. Then note that
    \[
    \Pr_{\beta \in \Ff_{q'}}\left[\deg(\wh{Q}(x, \beta)) \geq d+1\right] \geq 1 - \frac{d-1}{q'} \geq 1 - 2^{-\lambda}.
    \]
    Indeed, consider the coefficient of the highest $x$ monomial in $\wh{Q}(x, y)$. This coefficient is a polynomial in $y$, and since the total degree of $\wh{Q}(x,y)$ is $2d$, it is a degree $d-1$ polynomial in $y$, so the inequality above follows from the Schwartz-Zippel lemma.
    
    This shows that the first round has soundness error at most $2^{-\lambda}$.

    \paragraph{Second Round of Interaction:} After the second round of interaction, i.e.\ after step 5, the state is doomed if and only if one of the verifier's checks in step 5 fails.

    Suppose the previous state is doomed. We will show that for any message the prover sends in step 4, the next state will be doomed with probabilty at least $1- 2^{-\lambda}$. If the verifier has already rejected, then the next state is automatically doomed and we are done. Thus, we suppose that one of $\wh{Q}(x, \beta), \wh{Q}(\alpha, y)$ has degree at least $d+1$. Without loss of generality, suppose it is the former case. Then $\wh{F}_1(x) \neq \wh{Q}(x,\beta)$, and both have degree at most $2d$, so
    \[
    \Pr_{\xi_1 \in \Ff_{q'}}[\wh{F}_1(\xi_1) \neq \wh{Q}(\xi_1, \beta)] \geq 1 - \frac{2d}{q'}.
    \]
    Overall, this shows that the state will be doomed with probability at least, $1 - \frac{2d}{q'} \geq 1 - 2^{-\lambda}$ and completes our analysis of round-by-round soundness.
\end{proof}

\section{Batched Code Membership IOPPs for Compiling to IOPs} \label{sec: batch low rate}
In this section we give the inductive proofs of 
\cref{thm: RS low rate} and~\cref{thm: ind RM low rate}. For the base cases, where $d = O(1)$, the IOPs are trivial, and we focus henceforth on the inductive step.

\subsection{Set-up and an Auxiliary IOPP}
For the inductive step let us suppose that we have the $k-1$ cases of \cref{thm: RS low rate} and \cref{thm: ind RM low rate}. Namely, we assume that we have access to the IOPPs $\rsi(f, \sqrt{d}, q, \lambda, \eps_{k-1})$ and $ \irmi(f, (\sqrt{d}, \sqrt{d}), q, \lambda, \eps_{k-1})$ that satisfy the guarantees in~\cref{thm: RS low rate} and~\cref{thm: ind RM low rate} respectively. In addition to these IOPPs, we will also assume access to an analogous IOPP for the degree $2\sqrt{d}$ Reed-Solomon code that satisfies the following properties.
\begin{lemma} \label{lm: 2sqrtd rsi}
  Given $\rsi(f, \sqrt{d}, q, \lambda, \eps_{k-1})$ as in~\cref{thm: RS low rate}, one can construct an IOPP for degree $2\sqrt{d}$ Reed-Solomon code, which we denote by $\texttt{RS-IOPP}(f, 2\sqrt{d}, q, \lambda, \eps_{k-1})$, that has the following properties:
\begin{itemize}
    \item {\bf Input:} a function $f: \Ff_{q} \to \Ff_{q'}$, an agreement parameter $\eps_{k-1}>0$, and a degree parameter $2\sqrt{d}$ which is a of the form $2^{2^{k-1}+1}$.
    \item {\bf Completeness:} If $\deg(f) \leq 2\sqrt{d}$, then the honest prover makes the verifier accept with probability $1$.
    \item {\bf Round-by-round soundness:} $\lambda$.
    \item {\bf Initial state:} the initial state is doomed if and only if 
    $\agr(f, \RS_{q'}[2\sqrt{d}, \Ff_q]) \leq \eps_{k-1}$.
    \item {\bf Round complexity:} $O(k-1)$.
    \item {\bf Input query complexity:} 
    \[
    \iQ_{\RS, 2\sqrt{d}} = O\left(\frac{\lambda}{\log(1/\eps_{k-1})} \right).
    \]
    \item {\bf Proof query complexity:} 
    \[
    \pQ_{\RS,  2\sqrt{d}} = O\left(\sum_{i=1}^{k-1}  \frac{\lambda^2}{\log^2(1/\eps_i)} \right).
    \]
    \item {\bf Length:} $O\left((k-1) \cdot q^2\right)$.
\end{itemize}
\end{lemma}
\begin{proof}
The idea is that the prover sends $f_1, f_2 \in \Ff_{q'}^{\leq \sqrt{d}}[x]$ which supposedly satisfies $f_1 + x^{\sqrt{d}} \cdot f_2 = f$, and then both parties run $\rsi( f_1 + \xi\cdot f_2, \sqrt{d}, q, \lambda, \eps_{k-1})$ for a $\xi \in \Ff_{q'}$ that the verifier chooses randomly.
We defer the formal proof to~\cref{sec: 2sqrtd}.
\end{proof}

To complete the proofs of the inductive steps in \cref{thm: RS low rate} and \cref{thm: ind RM low rate}, we will apply \cref{thm: poly iop transformation} and compile the Poly-IOPs from \cref{lm: poly iop rs} and \cref{lm: poly iop ind rm} into legitimate IOPPs. The only task left is to construct the batched code membership IOPPs required by \cref{thm: poly iop transformation}. These IOPPs are nearly the ones assumed to exist by the inductive assumption. The only subtlety that arises is that we will need to test proximity to Reed-Solomon or Reed-Muller codes with side conditions. Hence the main work in this section is reducing from IOPPs for codes with side conditions to IOPPs for codes without side conditions,
and we use the tools from~\cref{sec:quotient}.

\subsection{Compiling $\texttt{RSPoly}$ for \cref{thm: RS low rate}}

From \cref{lm: poly iop rs}, we already have a Poly-IOP that performs the task of \cref{thm: RS low rate}. To complete the proof, we will construct a batched code membership IOPP, $\texttt{RSBatch}$, and use it to compile the Poly-IOP $\texttt{RSPoly}$ into a legitimate IOP, satisfying \cref{thm: RS low rate}. Our final IOPP will be $\compile(\rsp, \rsb(\eps'_{k-1}), \eps'_{k-1})$, for 
\begin{equation} \label{eq: eps' for compile}  
\eps'_{k-1} = 0.9\eps_{k-1} - \frac{T_k + 1}{q}.
\end{equation}

In the Poly-IOP, $\texttt{RSPoly}$, the prover sends one polynomial, $\wh{Q} \in \Ff_{q'}^{\leq (\sqrt{d}, \sqrt{d})}[x,y]$. When compiling $\texttt{RSPoly}$ to an IOP, the IOP prover will instead send a function $Q: \Ff_{q}^2 \to \Ff_{q'}$, which is supposedly the evaluation of $\wh{Q}$ over $\Ff_q^2$. The verifier will query $\wh{Q}$ at $T_k + 2$ points in total: the out-of-domain sample $a_0 \in \Ff_{q'}^2$, and the Poly-IOP queries $a_1, \ldots, a_{T_k} \in \Ff_q^2$. The prover will provide a value for $\wh{Q}(a_0)$, while the verifier can query the oracle $Q$ for the values  $\wh{Q}(a_1), \ldots,\wh{Q}(a_{T_k})$. Let $A, B \subseteq \Ff_{q'}$ be sets of size $T_k+1$ such that $a_0,\ldots, a_{T_k+1} \in A \times B$. In more detail, $A$ consists of all of the $x$-coordinates of the $a_i$'s, while $B$ consists of the $y$-coordinates. Let $h: \{a_0, \ldots, a_{T_k+1}\} \to \Ff_{q'}$ be the side condition function which agrees with the values $\wh{Q}(a_0), \ldots, \wh{Q}(a_{T_k})$ provided by the prover. Let $\wh{h}$ be the low degree extension of $h$ to $\Ff_{q'}^2$ that agrees with $\wh{Q}$ on $A \times B$. That is, we may first extend $h$ to the entire subcube $A \times B$ in a manner that agrees with $\wh{Q}$, and then take the low degree extension to $\Ff_{q'}^2$. Note that $\wh{h} \in \Ff_{q'}^{(T_k, T_k)}[x,y]$, so the prover can send $\wh{h}$ by giving its evaluation over $A \times B$ and then having the verifier calculate the low degree extension. 

To complete the proof the inductive step and conclude \cref{thm: RS low rate}, it remains to show how to construct the code membership IOPP with side conditions for the oracle function $Q$ with side conditions according to the queries above. Specifically, we must test proximity of $Q$ to $\RM_{q'}[(\sqrt{d}, \sqrt{d}), \Ff_q^2 \; | \; h]$, and we will do so with the help of \cref{lm: bivariate side condition decomp+combine} which allows us to reduce to the case without side conditions case:
\begin{algorithm}[H] \caption{Batched IOP for Reed-Solomon: $\texttt{RSBatch}$} \label{iop: batch rs}
    \begin{algorithmic}[1]
         \STATE \textbf{P:} The prover sends $h: A \times B \to \Ff_{q'}$, that agrees with the side conditions, where $A \times B$ is the smallest subcube containing all of the side condition points. That is, they only provide the values of $h$ at $A \times B \setminus \{a_0,\ldots, a_{T_k+1}\}$, while the verifier assumes that the remaining values of $h$ agree with the side conditions. The prover also sends a function $g_{1}: \Ff_q^2 \to \Ff_{q'}$ and a function $\Fill:\Ff_q \times B \to \Ff_{q'}$.
        \STATE Let 
    \[
    g_{2} = \quo(f - V_A \cdot g_1 - \tilde{h}, B, \Fill),
    \]
    where $\tilde{h}$ is the degree $(T_k, T_k)$-extension of $h$ over $\Ff_q^2$. In the honest case, $g_1$ and $g_2$ are as \cref{lm: ind deg bivariate side condition decomp completeness case} and $h = \wh{Q}|_{A \times B}$. 
    \STATE \textbf{V:} The verifier chooses $F = \combine_{(\sqrt{d}, \sqrt{d})}(g_1, g_2)$ as in \cref{lm: combine multiple multivariate}.
    \STATE \textbf{P + V:} Both parties run $\irmi(F, (\sqrt{d}, \sqrt{d}), \lambda, \eps_{k-1})$.
    \end{algorithmic}
\end{algorithm}
We will now show that \cref{iop: batch rs} satisfies the requirements outlined in \cref{thm: poly iop transformation}. This will show that $\compile(\texttt{RSPoly}(f, d, \lambda, \eps_k), \texttt{RSBatch})$ satisfies the guarantees of \cref{thm: RS low rate}, thereby completing the proof of \cref{thm: RS low rate}.

\begin{lemma} \label{lm: batch rs}
The IOP $\texttt{RSBatch}$ has the following guarantees.
\begin{itemize}
\item {\bf Input:} a function $Q: \Ff_{q}^2 \to \Ff_{q'}$ and a side condition function $h$ as described above.
\item {\bf Completeness:} if $Q \in \RM_{q'}[(\sqrt{d}, \sqrt{d}), \Ff_q^2 \; | \; h]$, then the honest prover makes the verifier accept with probability $1$.
\item {\bf Round-by-round soundness:} $2^{-\lambda}$. 
\item {\bf Initial state:} the initial state is doomed if and only if $\agr(Q, \RM_{q'}[(\sqrt{d}, \sqrt{d}), \Ff_q^2 \; | \; h]) \leq \eps'_{k-1}$.
\item {\bf Round complexity:} $O(k)$.
\item {\bf Input query complexity:} $O(T_k)$.
\item {\bf Proof query complexity:} $\pQ_{\iRM, k-1} + O(\iQ_{\iRM, k-1}) +(T_k+2)^2$.
\item {\bf Length:} $O\left(k \cdot q^2 \right)$
\end{itemize}
\end{lemma}
\begin{proof}
We discuss all of the guarantees except for round-by-round soundness, which we save for the next section. The round complexity, input query complexity, proof query complexity, and length are clear. In particular, they from the inductive assumption that $\irmi$ satisfies \cref{thm: ind RM low rate}. For the proof query complexity, note that there are $\pQ_{\iRM, k-1}$ proof queries when running $\irmi(F, (\sqrt{d}, \sqrt{d}), \lambda, \eps_{k-1})$. Furthermore, the $\iQ_{\iRM, k-1}$ input queries to $F$, become proof queries. Note that each query to $F$ can be simulated by $O(1)$ oracle queries. We now argue the completeness.

In step 1, the honest prover completes the side condition function according to $\wh{Q}$ by sending $h = \wh{Q}|_{A \times B}$. They also send $g_1$ and $\Fill$ according to \cref{lm: ind deg bivariate side condition decomp completeness case}. By \cref{lm: ind deg bivariate side condition decomp completeness case}, $F = \combine_{(\sqrt{d}, \sqrt{d})}(g_1, g_2) \in \RM_{q'}[(\sqrt{d}, \sqrt{d})]$ and the remainder of the completeness follows from that of $\irmi(F, (\sqrt{d},\sqrt{d}), \lambda, \eps_{k-1})$.
\end{proof}

\subsubsection{Round-by-Round Soundness for \cref{lm: batch rs}}
There is only one round of interaction, after which we define the state as doomed if and only if 
\begin{equation} \label{eq: batch rs state}    
\agr(F, \RM_{q'}[(\sqrt{d}, \sqrt{d}), \Ff_q^2]) \leq \eps_{k-1}.
\end{equation}
The soundness of this round is established by the next claim.
\begin{claim}
Suppose the initial state is doomed, meaning 
\[
\agr(Q, \RM_{q'}[(\sqrt{d}, \sqrt{d}), \Ff_q^2 \; | \; h]) \leq \eps'_{k-1}.
\]
Then for any $g_1, \Fill$ that the prover sends, we have that 
\[
\agr(F, \RM_{q'}[(\sqrt{d}, \sqrt{d}), \Ff_q^2] \leq \eps_{k-1},
\]
with probability at least $1-2^{-\lambda}$.
\end{claim}
\begin{proof}
Fix any $g_1, \Fill$. 
By \cref{lm: bivariate side condition decomp} we have that $g_1$ and $g_2$ have correlated agreement at most $\eps'_{k-1} + \frac{T_k+1}{q}$ with degrees $(\sqrt{d}-T_k-1, \sqrt{d})$ and $(\sqrt{d}-T_k-1, \sqrt{d})$ respectively. 
Applying \cref{lm: bivariate side condition decomp+combine} we get that with probability at least $1-2^{-\lambda}$, the function $F = \combine_{(\sqrt{d}, \sqrt{d})}(g_1, g_2)$ has agreement at most $1.04(\eps'_{k-1} + \frac{T_k+1}{q}) \leq \eps_{k-1}$ with degree $(\sqrt{d}, \sqrt{d})$ over $\Ff_q^2$.
\end{proof}

The remainder of the round-by-round soundness follows from that of $\irmi$. Indeed, if~\eqref{eq: batch rs state} holds and the state is doomed prior to step $4$, then we are in the doomed initial state of the IOPP $\irmi(F, (\sqrt{d},\sqrt{d}), q, \lambda, \eps_{k-1})$ as described in \cref{thm: ind RM low rate}. 

\subsubsection{Concluding the Proof of \cref{thm: RS low rate}}
With $\rsp$ and $\texttt{RSBatch}$, we are ready to complete the proof of the inductive step for \cref{thm: RS low rate}, by plugging both IOPs into $\compile$ from~\cref{thm: poly iop transformation} . This also concludes the the proof of \cref{thm: RS low rate}. We summarize this process in the following claim.

\begin{claim}
Let $d = 2^{2^{k}}$. Assuming access to the IOPP 
$\irmi(f, (\sqrt{d}, \sqrt{d}), q, \lambda, \eps_{k-1})$ satisfying 
the $k-1$ case of \cref{thm: ind RM low rate}, the IOP $\compile(\rsp, \rsb(\eps'_{k-1}), \eps'_{k-1})$ satisfies the guarantees of the $k$ case of \cref{thm: RS low rate}.
\end{claim}
\begin{proof}
    The parameters are straightforward to verify by applying \cref{thm: poly iop transformation} with the guarantees for $\rsp$ and $\texttt{RSBatch}$ given by \cref{lm: poly iop rs} and \cref{lm: batch rs} respectively. 
\end{proof}

\subsection{Compiling $\texttt{iRMPoly}$ for \cref{thm: ind RM low rate}}
We keep $\eps'_{k-1}$ as in \cref{eq: eps' for compile}. In particular, note that 
\[
\eps'_{k-1} \leq \min\left(0.9\eps_{k-1}, \eps_{k-1}-\frac{3}{q}\right).
\]

We next design an appropriate batched code membership IOPP appropriate 
for compiling $\texttt{iRMPoly}$ from~\cref{thm: ind RM low rate}.
In the Poly-IOP $\texttt{iRMPoly}$, the prover sends three polynomials: $\wh{Q} \in \Ff_{q'}^{\leq 2\sqrt{d}}[x,y]$, $\wh{F}_1 \in \Ff_{q'}^{\leq \sqrt{d}}[x]$, and $\wh{F}_2 \in \Ff_{q'}^{\leq \sqrt{d}}[y]$. When compiling $\texttt{RSPoly}$ to an IOP, the prover will instead send the functions $Q, F_1, F_2$ which are given over $\Ff_q^2, \Ff_q,$ and $\Ff_q$ respectively. These oracle functions are supposedly the evaluations of $\wh{Q},\wh{F}_1,$ and $\wh{F}_2$ respectively. While simulating $\irmp$ the verifier will query $\wh{Q}$ at $T_k+3$ and each of $\wh{F}_1$ and $\wh{F}_2$ at $2$ points. Let $A_Q \times B_Q$, be a subcube containing the side condition points for $\wh{Q}$. Note that we may have $|A_Q| = |B_Q| = T_k+3$. Let $h_Q, h_1, h_2$ be the side condition functions for $\wh{Q}, \wh{F}_1, \wh{F}_2$, i.e. each side condition function maps the verifier's queries to the prover's answers for that respective polynomial. Let $\wh{h}_Q \in \Ff_{q'}^{\leq (T_k+2, T_k+2)}[x,y], \wh{h}_1 \in \Ff_{q'}^{\leq 1}[x], \wh{h}_2 \in \Ff_{q'}^{\leq 1}[y]$ be the low degree extensions after completing $h_Q, h_1,$ and $h_2$ respectively after completing them to $A_Q \times B_Q$, $A_1$, and $A_2$ in agreement with their respective polynomials.

To complete the proof the inductive step and conclude \cref{thm: ind RM low rate}, it remains to show how to construct the code membership IOPP with side conditions for the oracle functions above with side conditions. Namely, we must construct IOPPs for the following:
\begin{enumerate}
\item $Q$ to $\RM_{q'}[2\sqrt{d}, \Ff_q^2 \; | \; h_Q]$.
\item $F_1$ to $\RS_{q'}[\sqrt{d}, \Ff_q \; | \; h_1]$.
\item $F_2$ to $\RS_{q'}[\sqrt{d}, \Ff_q \; | \; h_2]$.
\end{enumerate}
The batched IOPP is motivated by the following decompositions from~\cref{lm: uni quo completeness,lm: ind deg bivariate side condition decomp completeness case}, saying
that in the completeness case we have
\begin{align} \label{eq: irm decomps}
    Q(x,y) &= V_{A_Q}(x) \cdot g_{Q, 1}(x,y) + V_{B_Q}(y) \cdot g_{Q,2}(x,y) + \wh{h}_Q|_{\Ff_q^2} \notag\\
    F_1(x) &= V_{A_1}(x) \cdot f_1(x) + \wh{h}_1|_{\Ff_q}  \notag\\
    F_2(x) &= V_{A_2}(x) \cdot f_2(x) + \wh{h}_2|_{\Ff_q}, 
\end{align}
for $g_{Q,1}, g_{Q,2} \in \RM_{q'}[2\sqrt{d}- T_k-3, \Ff_q^2]$ and $f_1, f_2 \in \RS_{q'}[\sqrt{d}-2, \Ff_q]$. The batched IOPP is presented below.
\begin{algorithm}
\begin{algorithmic}[1]
\caption{Batched IOPP for compiling $\irmp$: $\irmb$\label{iop: batch irm}}
    \STATE \textbf{P:} The prover sends the following functions
    \begin{itemize}
        \item $\wh{h}_Q \in \Ff_{q'}^{(T_k+2, T_k+2)}[x,y]$ which agrees with the side conditions for $Q$. In the honest case, this is the low degree extension of $\wh{Q}|_{A_Q \times B_Q}$. The prover sends the polynomial by completing the evaluation over $A_Q\times B_Q$, allowing the verifier to calculate $\wh{h}_Q$ via interpolation. 
        \item $\wh{h}_1, \wh{h}_2 \in \Ff_{q'}^{\leq 1}[x]$. In the honest case, these are the low degree extensions of the side condition functions on $\wh{F}_1|_{A_1}, \wh{F}_2|_{A_2}$ respectively.
        \item $g_{Q,1}: \Ff_q^2 \to \Ff_{q'}$.
        \item $\Fill_Q: X \times Y \to \Ff_{q'}$.
        \item $\Fill_1: A_1 \to \Ff_{q'}$ and $\Fill_2: A_2 \to \Ff_{q'}$
    \end{itemize} 
    \STATE Let $h_1, h_2$ be the functions that are the evaluations of $\wh{h}_1, \wh{h}_2$ over $\Ff_q$ respectively. Let $h_Q$ be the function which is the evaluation of $\wh{h}_Q$ over $\Ff_q^2$. Define the following functions over $\Ff_q^2$, $\Ff_q$, and $\Ff_q$ respectively
    \begin{align*}
        g_{Q,2} &= \Quo_2(Q - V_X \cdot g_{Q,1} - h_Q, Y, \Fill_Q), \\
        f_1 &= \Quo(F_1 - h_1, A_1), \\
        f_2 &= \Quo(F_2 - h_2, A_2).
    \end{align*}
    In the honest case, all of the functions appearing are as in~\eqref{eq: irm decomps}.

    \STATE \textbf{V:} The verifier does the following:
    \begin{itemize}
        \item Choose coefficients according to the proximity generator to obtain
        \[
        G = \combine_{2\sqrt{d}}(g_{Q,1}, g_{Q,2}).
        \]
        They send the random coefficients in $\combine$ to the prover.
        \item Choosing $T_k$ lines, $L_1, \ldots, L_{T_k} \subseteq \Ff_q^2$ and also generates
        \[
         F = {\sf combine}_{2\sqrt{d}}(G|_{L_1},\ldots, G|_{L_{T_k}}, f_1, f_2)
        \]
        according to \cref{def: uni combine}. They send the lines and the random coefficients used in $\combine$ to the prover.
    \end{itemize}

    \STATE \textbf{P + V:} Both Parties run $\rsi(F, 2\sqrt{d}, \lambda, \eps_{k-1})$.
\end{algorithmic}
\end{algorithm}
We will now show that \cref{iop: batch irm} satisfies the requirements outlined in \cref{thm: poly iop transformation}. 
\vspace{0.5cm}
\begin{lemma} \label{lm: batch irm}
The IOPP, $\texttt{iRMBatch}$ has the following guarantees.
\begin{itemize}
\item {\bf Input:} functions $Q: \Ff_q^2 \to \Ff_{q'}$, $F_1: \Ff_q \to \Ff_{q'}$, and $F_2: \Ff_q \to \Ff_{q'}$, as well as their respective side condition functions $h_Q, h_1,$ and $h_2$ as described above. 

\item {\bf Completeness:} if $Q \in \RM_{q'}[2\sqrt{d}, \Ff_q^2 \; | \; h_Q]$, $F_1 \in \RS_{q'}[\sqrt{d}, \Ff_q \; | \; h_1]$ and $F_2 \in \RS_{q'}[\sqrt{d}, \Ff_q \; | \; h_2]$, then the honest prover makes the verifier accept with probability $1$.
\item {\bf Round-by-round soundness:} $2^{-\lambda}$. 
\item {\bf Initial state:} the initial state is doomed if and only if one of the following holds 
\begin{itemize}
    \item 
    $\agr(Q, \RM_{q'}[2d, \Ff_{q}^2 \; | \; h_Q]) \leq \eps'_{k-1}$,
    \item $\agr(F_1, \RS_{q'}[d, \Ff_{q} \; | \; h_1]) \leq \eps'_{k-1}$,
     \item $\agr(F_2, \RS_{q'}[d, \Ff_{q} \; | \; h_2]) \leq \eps'_{k-1}.$
\end{itemize}
\item {\bf Round complexity:} $O(k)$.
\item {\bf Input query complexity:} $T_k \cdot O(\iQ_{\RS, k-1})$.
\item {\bf Proof query complexity:} $\pQ_{\RS, k-1} + T_k \cdot O(\iQ_{\RS,k-1}) + O(T_k^2)$.
\item {\bf Length:} $O(k \cdot q^2)$.
\end{itemize}
\end{lemma}
\begin{proof}
We discuss all of the guarantees except for round-by-round soundness, which we defer for the next section. The round complexity, input query complexity, and proof query complexity are clear. For the proof complexity, the verifier makes $O(T_k^2)$ proof queries to read all of the side condition functions and $\pQ_{\RS, k-1}$ proof queries when running $\rsi(F, 2\sqrt{d}, \lambda, \eps_{k-1})$ in step 4. Finally, each input query to $F$ when running  $\rsi(F, 2\sqrt{d}, \lambda, \eps_{k-1})$ is simulated via $O(1)$ proof queries. We analyze the completeness below.

In step 1, the honest prover completes the side condition function $\wh{h}_Q$ according to $\wh{Q}$ and similarly completes $\wh{h}_1, \wh{h}_2$ according to $\wh{F}_1$, $\wh{F}_2$. By \cref{lm: uni quo completeness,lm: ind deg bivariate side condition decomp completeness case}, the honest prover can send $g_{Q,1}, \Fill_Q ,\Fill_1$, and $\Fill_2$ which result in $g_{Q,2} \in \RM_{q'}[2\sqrt{d} - T_k-1), \Ff_q^2]$ and $f_1, f_2 \in \RS[\sqrt{d}-2, \Ff_q]$. It follows that $G = \combine_{2\sqrt{d}}(g_1, g_2) \in \RM_{q'}[2\sqrt{d}]$, and thus for any line $L \subseteq \Ff_q^2$, we have $G|_{L} \in \RS_{q'}[2\sqrt{d}, \Ff_q]$. As a result, $F = \combine_{2\sqrt{d}}(G|_{L_1}, \ldots, G|_{L_{T_k}}, f_1, f_2) \in \RS_{q'}[2\sqrt{d}, \Ff_q]$ with probability $1$ and the remainder of the completeness follows from that of $\rsi(F, 2\sqrt{d}, \lambda, \eps_{k-1})$.
\end{proof}
\subsubsection{Round-by-round Soundness for \cref{thm: ind RM low rate}}
There is only one round of interaction, after which we define the state as doomed if and only if 
\[
\agr(F, \RS_{q'}[2\sqrt{d}, \Ff_q]) \leq \eps_{k-1}.
\]

Now suppose that the state is doomed before this round of interaction. Then, one of the following holds:
\begin{align*}
    \agr(Q, \RM_{q'}[2d, \Ff_{q}^2 \; | \; h_Q]) \leq \eps'_{k-1}, \\
    \agr(F_1, \RS_{q'}[d, \Ff_{q} \; | \; h_1]) \leq \eps'_{k-1}, \\ 
     \agr(F_2, \RS_{q'}[d, \Ff_{q} \; | \; h_2]) \leq \eps'_{k-1},
\end{align*}
The soundness of the first round then follows from the next lemma combined with \cref{lm: combine multiple univariate}.

\begin{lemma}
If any of the following hold
\begin{align*}
    \agr(Q, \RM_{q'}[2d, \Ff_{q}^2 \; | \; h_Q]) \leq \eps'_{k-1}, \\
    \agr(F_1, \RS_{q'}[d, \Ff_{q} \; | \; h_1]) \leq \eps'_{k-1}, \\ 
     \agr(F_2, \RS_{q'}[d, \Ff_{q} \; | \; h_2]) \leq \eps'_{k-1},
\end{align*}
Then, for any $g_1, \Fill_Q, \Fill_1, \Fill_2$ that the prover sends, we have,
\[
\E_{L_1, \ldots, L_K}\left[ \Pr_{\xi_i \in \Ff_{q'}}[\agr(F, \RS_{q'}[2\sqrt{d}, \Ff_q]) > \eps_{k-1}]\right] \leq 2^{-\lambda - 10}.
\]
\end{lemma}
\begin{proof}
First suppose that either the second or third item holds, and by symmetry assume it is the second, namely
$\agr(F_1, \RS_{q'}[d, \Ff_{q} \; | \; h_1]) \leq \eps'_{k-1}$.
Using~\cref{lm: uni quo soundness}, for any $\Fill_1$ that the prover sends we have
\[
\agr(f_1, \RS_{q'}[d, \Ff_q]) \leq \eps'_{k-1} + \frac{3}{q},
\]
where $f_1$ is as defined in \cref{iop: batch irm}. In this case, the desired result follows from \cref{lm: combine multiple univariate}. 

Now suppose that it is the first item that holds, that is $\agr(Q, \RM_{q'}[2d, \Ff_{q}^2 \; | \; h_Q]) \leq \eps'_{k-1}$. It follows from \cref{lm: total deg bivariate side condition decomp+combine} that for any $g_1, \Fill_Q$ that the prover sends we have
\[
\agr(G, \RM_{q'}[2\sqrt{d}-|Y|, \Ff_q^2]) \leq 1.04\eps'_{k-1},
\]
with probability at least $1 - 2^{-\lambda-10}$. Since $(1.04\eps'_{k-1})^{2} \geq 20(\sqrt{d}/q)^{\tl}$, \cref{thm: line vs point strong} implies
    \[
    \Pr_{L \subseteq \Ff_q^2}[\agr(Q|_L, \RS_{q'}[2\sqrt{d}, \Ff_q^2] \geq 1.06\eps'_{k-1}] \leq \eps'_{k-1}.
    \]
Thus, with probability at least $1 - \eps_{k-1}^{'T_k}= 1-2^{-\lambda-10}$, we have that for at least one line $L_i$, 
\[
\agr(Q|_{L_i},  \RS_{q'}[2d, \Ff_q^2] \leq 1.06\eps'_{k-1}.
\]
In this case, \cref{lm: combine multiple univariate} says that $F$ has agreement at most $1.06\eps'_{k-1} \leq \eps_{k-1}$ with probability at least $1 - 2^{-\lambda-10}$. The result follows from a union bound. 
\end{proof}
\subsubsection{Concluding the Proof of \cref{thm: ind RM low rate}}
With $\irmp$ and $\texttt{iRMBatch}$, we are ready to complete the proof of the inductive step for \cref{thm: ind RM low rate}, by plugging both IOPs into $\compile$. This in turn concludes the the proof of \cref{thm: ind RM low rate}. We summarize this process in the following claim.

\begin{claim}
Let $d = 2^{2^{k}}$. Assuming access to the IOPP, $\rsi(f, 2\sqrt{d}, q, \lambda, \eps_{k-1})$ which satisfies the $k-1$ case of \cref{thm: ind RM low rate}, the IOPP $\compile(\irmp, \irmb(\eps'_{k-1}))$ satisfies the guarantees of the $k$ case of \cref{thm: RS low rate}.
\end{claim}
\begin{proof}
    The parameters are straightforward to verify by applying \cref{thm: poly iop transformation} with the guarantees for $\irmp$ and $\irmb$ given by \cref{lm: poly iop ind rm} and \cref{lm: batch irm} respectively. 
\end{proof}

\section{IOPPs for Constant Rate Reed-Muller Codes}

In this section we present our constant rate trivariate Reed-Muller IOPP. Specifically, we want to test the proximity of $f: \Ff_q^3 \to \Ff_{q'}$ to the code $\RM_{q'}[6d, \Ff_q^3]$, where $q = \Theta(d)$. The IOPP construction operates restricting to a smaller dimensional subspace in a similar manner as our low rate IOPPs; however, the constant rate regime presents additional challenges which we briefly discuss here.

The first difficulty that arises is that we cannot rely on \cref{thm: line vs point strong} and look at the restriction of $f$ to random lines. This is because we are in the constant rate regime and have to reject $f$ that may have agreement as high as $\eps = \Omega(1)$ with degree $6d$. Using the bound for the line verus point test in \cref{thm: line vs point strong}, we would then need to choose $\lambda / \log(1/\eps) = \Omega(\lambda)$ random lines to find one on which $f$ is far from degree $6d$, and this is too many. To resolve this, we rely instead on the stronger bound for randomly chosen affine planes in \cref{thm: plane vs point strong}. This theorem allows us to choose only $O(\lambda/\log(q)) = O(\lambda/\log(d))$ affine planes to find at least one on which $f$ is far from degree $6d$. 

From here it may seem that we are in the clear and can simply use the proximity generator theorem and test the proximity of some auxiliary bivariate function $g: \Ff_q^2 \to \Ff_q$ to $\RM_{q'}[6d, \Ff_q^2]$, but notice that we are still stuck in the constant rate setting. Indeed, the code $\RM_{q'}[6d, \Ff_q^2]$ has constant rate again. To get out of this setting though, we can take advantage of the fact that $\RM_{q'}[6d, \Ff_q^2]$ is bivariate instead of trivariate. Specifically, we can now consider the encoding over some larger field $\Fqe$ and instead test proximity to $\RM_{q'}[6d, \Fqe^2]$. The upshot is we only suffer length $\qenc^2$, so we can afford to make $\qenc$ slightly larger than $q$. Specifically, by choosing $\qenc$ which is polynomially larger than $q$, but still much smaller than $q^{3/2}$, we can simultaneously achieve low rate and sublinear length (sublinear in $q^3$ that is)! Once we have reduced to testing proximity to a low rate Reed-Muller code, we can (essentially) apply our IOPPs from \cref{sec: low rate}.

We now formally state our main theorem, which will require stating some of the parameters mentioned above. Fix a security parameter $\lambda$, degree parameter of the form $d = 2^{2^k}/6$, and finite fields $\Ff_q, \Ff_{q'}$ such that $\Ff_q$ is a subfield of $\Ff_{q'}$ and 
\[
 q' = 2^{\lambda} \poly(q) \quad,  \quad q > 6d \quad
\]
Our input will be a function $f: \Ff_q^3 \to \Ff_{q'}$ and we will be testing proximity from $f$ to $\RM_{q'}[6d, \Ff_q^3]$. 

Towards the construction of this IOPP, we will also require a second, larger field to be used in constructing the batched IOPP during the compiling step. We let this field be $\Fqe$, and we require it to be a subfield of $\Ff_{q'}$. We also choose its size to be polynomially larger than $q$ but much smaller than $q^{3/2}$. To be concrete, we set $q_{{\sf enc}} = q^{1.2}$, and also define 
\begin{equation} \label{eq: eps test definition} 
\et = 23\left(\frac{6d}{\qenc} \right)^{\tl/2}.
\end{equation}
As the IOPP of this section is intended for use when the rate is constant, one should think of $q = \Theta(d)$.

The reason we must design a different IOPP here compared to \cref{sec: low rate} is that in the current setting, $\lambda/\log(1/\eps) = \Theta(\lambda)$, so we cannot afford terms such as  $\left(\lambda/\log(1/\eps)\right)^2$ in our query complexity. This was affordable in the low rate setting, as there, the agreement parameter was inverse polynomial in $q$, and thus $\left(\lambda/\log(1/\eps)\right)^2 = O\left(\lambda/\log(q)\right)^2$, which is often small compared to $\lambda$. 

During the compilation step, we will use the IOPP $\rsi$ from \cref{thm: RS low rate} when the input is over the field $\Fqe$, and below we recall its parameters adjusted to our setting:
\begin{lemma}
Given a security parameter $\lambda$, a degree parameter $6d = 2^{2^k}$, and field sizes $q_{\enc}$ and $q'$ as above, there is an IOPP with the following guarantees. 
\begin{itemize}
    \item {\bf Input:} a function $G: \Fqe \to \Ff_{q'}$, parameters $6d, q', \lambda$, as above, and an agreement parameter $\et$ as in \cref{thm: RS low rate}.
    \item {\bf Completeness:} If $G \in \RS_{q'}[6d, \Fqe]$ then the honest prover makes the verifier accept with probability $1$.
    \item {\bf Round-by-round soundness:} $2^{-\lambda}$.
    \item {\bf Initial state:} the initial state is doomed if and only if $\agr(G, \RS_{q'}[6d, \Fqe]) \leq \et$.
    \item {\bf Round complexity:} $O(\log \log d)$
    \item {\bf Input query complexity:} $O\left(\frac{\lambda}{\log(q_{{\sf enc}})}\right)$.
    \item {\bf Proof query complexity:} $O\left(\frac{\lambda^2}{\log^2(q_{{\sf enc}})} \log \log d\right).$
    \item {\bf Length:} $O(q_{{\sf enc}}^2 \cdot \log \log (d))$.
\end{itemize}
We refer to this IOPP as $\rsi(f, d, q_{\enc}, \lambda, \et)$
\end{lemma}
\begin{proof}
    We apply \cref{thm: RS low rate} with degree $6d$ and parameters $\lambda, q_{\enc}, q', \et$ as above. 
\end{proof}

Our goal in this section is to construct an IOPP as per the following theorem.
\begin{thm}\label{thm: tot rm constant rate}
Fix degree parameter $6d$ and security parameter $\lambda$. Let $\Ff_q, \Fqe, \Ff_{q'}$ be fields such that 
$\Ff_q$ and $\Fqe$ are subfields of $\Ff_{q'}$ with sizes that satisfy:
\[
q > 6d  \quad,  \quad q' = 2^{\lambda} \poly(q) \quad \text{and} \quad \qenc > 6d.
\]
Then for any proximity parameter 
\[
\eps \geq 23\left(\frac{6d}{q}\right)^{\tl/2},
\]
there is an IOPP for the degree $6d$ Reed-Muller code with the following guarantees:
  \begin{itemize}
      \item {\bf Input:} a function $f: \Ff_q^3 \to \Ff_{q'}$.
      \item {\bf Completeness:} if $f \in \RM_{q'}[6d, \Ff_q^3]$, then the honest prover makes the verifier accept with probability $1$.
      \item {\bf Round-by-round soundness:} $2^{-\lambda}$.
      \item {\bf Initial state:} the initial state is doomed if and only if $\agr(f, \RM_{q'}[6d, \Ff_q^3]) \leq \eps$.
      \item {\bf Round complexity:} $O(\log \log d)$.
      \item {\bf Input query complexity:} 
      $O\left(\frac{\lambda}{\log(1/\eps)}\right)$.
      \item {\bf Proof query complexity:} $
      O\left(\frac{\lambda}{\log(1/\eps)}\right) + O\left(\frac{\lambda^2}{\log^2(d)} \log \log d\right)$.
      \item {\bf Length:} $O(q^3) + O(\qenc^2)$.
  \end{itemize}
We call this IOP $\rmi(f, 6d, q, \qenc, \lambda, \eps)$.
\end{thm}

Similar to the proof in the low rate regime, we first construct a Poly-IOP, and then compile it into a legitimate IOP by constructing a suitable batched code membership test according to \cref{thm: poly iop transformation}. 

\subsection{A Poly-IOP for \cref{thm: tot rm constant rate}}
For the remainder of this section, fix $\eps$ to be the proximity parameter from \cref{thm: tot rm constant rate} and set 
\[
 T = 5\left(\frac{\lambda}{2\log(\eps) + \log(q)} \right) \quad \text{and} \quad t = \frac{2 \log(q)}{\log(1/\eps)}.
\]
Note that the lower bound on $\eps$ in \cref{thm: tot rm constant rate} implies $T = \Theta(\lambda/\log(q))$, and it will be helpful to keep this in mind throughout the section. We start with a Poly-IOP for testing proximity to the constant rate Reed-Muller code.
\begin{lemma}\label{lm: poly iop tot rm constant rate}
    Keeping the notation of \cref{thm: tot rm constant rate}, there is a Poly-IOP, $\rmp(f, 6d, q, \lambda, \eps)$, that has the following guarantee:
\begin{itemize}
        \item {\bf Input:} the input is an oracle function $f: \Ff^3_q \to \Ff_{q'}$.
        \item {\bf Completeness:} if $f \in \RM_{q'}[6d, \Ff^3_q]$, then the honest prover makes the verifier accept with probability $1$.
        \item {\bf Round-by-round soundness:} $2^{-\lambda}$.
        \item {\bf Initial state:} the initial state is doomed if and only if $\agr(f, \RM_{q'}[6d, \Ff^3_q]) \leq \eps$.
        \item {\bf Number of polynomials:} $T+1$.
        \item {\bf Input query complexity:} $ O\left(\frac{\lambda}{\log(1/\eps)}\right)$.
        \item {\bf Proof query complexity:} $T \cdot (t+3)$.
    \end{itemize}
\end{lemma}
\begin{proof}
We begin by formally presenting the Poly-IOP:
\begin{algorithm}[H] \label{iop: trm const rate poly iop}
     \caption{Total Degree Reed-Muller Poly-IOP: $\rmp(f, 6d, \lambda, \eps)$}  \label{poly iop: total rm constant rate}
     \begin{algorithmic}[1]
          \STATE \textbf{V:} The verifier chooses $T$ planes, $P_1, \ldots, P_T \subseteq \Ff_q^3$ and sends these to the prover. 
    \STATE \textbf{P:} The prover sends $\wh{Q}_1, \ldots, \wh{Q}_T \in \Ff_{q'}^{6d}[x,y]$. In the honest case $\wh{Q}_i = f|_{P_i}$. 
    \STATE \textbf{V:} The verifier chooses $T$ lines $L_1, \ldots, L_T \in \Ff_{q}^2$ and on each line, $L_i$, chooses $t = 2\log(q)/\log(1/\eps)$ points, $z_{i,1},\ldots, z_{i,t} \in L_i$. The verifier checks if
    \[
    \wh{Q}_i(z_{i,j}) = f|_{P_i}(z_{i,j}) \quad \forall i \in [T], j \in [t].
    \]
    If any check fails, the verifier rejects. Otherwise, the verifier accepts.
     \end{algorithmic}
\end{algorithm}
   The input query complexity, proof query complexity, and round complexity are straightforward to check. For the completeness, each $f|_{P_i}$ is a degree $6d$ oracle over $\Ff_q^2$, so the honest prover can respond with $\wh{Q}_i$ which is the low degree extension of $f|_{P_i}$ to $\Ff_{q'}^2$ for each $i \in [T]$. It is clear that the verifier will then accept with probability $1$ in step 3. The rest 
   of the argument it devoted to the round-by-round soundness analysis.

We define state functions after each round of interaction and analyze the soundness of each round. After the first round of interaction, i.e.\ after step 1, the state function is as follows.
\begin{state} \label{state: tot rm poly 1}
The state is doomed if and only if for at least $\frac{T}{2}$ of the planes $P_i$ we have 
\[
\agr(f|_{P_i}, \RM_{q'}[6d, \Ff_q^2]) \leq 1.2 \eps.
\]
\end{state}
The soundness of the first round is established by the following claim.
\begin{claim}\label{claim:aux_rbr_1}
    If $\agr(f, \RM_{q'}[6d, \Ff_q^3]) \leq \eps$, then \cref{state: tot rm poly 1} is doomed with probability at least $1-2^{-\lambda}$
\end{claim}
\begin{proof}
    By \cref{thm: plane vs point strong} we have that,
    \[
    \Pr_{P \subseteq \Ff_q^3}[\agr(f|_P, \agr_{q'}[6d, \Ff_q^2]) \geq 1.2\eps] \leq \frac{100}{\eps^2 q}.
    \]
    It follows that \cref{state: tot rm poly 1} is doomed with probability at least,
    \[
   1- \binom{T}{T/2} \cdot \left(\frac{100}{\eps^2 q}\right)^{T/2} \geq  1-\left(\frac{600}{\eps^2 q}\right)^{T/2} \geq 1-2^{-\lambda}.\qedhere
    \]
\end{proof}

After the first round of interaction, i.e.\ after step 3, the state function is as follows.
\begin{state} \label{state: tot rm poly 2}
The state is doomed if and only if for at least one $i \in [T], j \in [t]$ we have $\wh{Q}_{i}(z_{i,j}) \neq f(P_i(z_{i,j}))$.
\end{state}
The soundness of the second round is established by the following claim.
\begin{claim}\label{claim:aux_rbr_2}
    If $\agr(f, \RM_{q'}[6d, \Ff_q^3]) \leq \eps$ and \cref{state: tot rm poly 1} is doomed, then \cref{state: tot rm poly 2} is doomed with probability at least $1-2^{-\lambda}$
\end{claim}
\begin{proof}
 We have that \cref{state: tot rm poly 1} is doomed, and assume without loss of generality that it is $P_1, \ldots, P_{T/2}$ which satisfy $\agr(f|_{P_i}, \RM_{q'}[6d, \Ff_q^2]) \leq 1.2\eps$. Then, for each $i \in [T/2]$, we have
 \[
 \Pr_{z \in \Ff_{q^2}}[\wh{Q}_{i}(z) = f(P_i(z))] \leq 1.2\eps.
 \]
 By \cref{lm: spectral sampling} applied to the line-point inclusion graph, we have that with probability at least $1-\frac{100}{\eps^2 q}$ over a line $L \subseteq \Ff_q^2$ the following holds:
 \[
 \Pr_{z \in L}[\wh{Q}_{i}(z) = f(P_i(z))] \leq 1.3\eps.
 \]
 Thus, for each $i \in [T/2]$
 \[
 \Pr_{L_i, z_{i,1}, \ldots, z_{i, t}}[\wh{Q}_{i}(z_{i,j}) = f(P_i(z_{i,j})) \; \forall j \in [t]] \leq \frac{100}{\eps^2 q} + (1.3\eps)^t \leq \frac{101}{\eps^2 q}.
 \]
 It follows that \cref{state: tot rm poly 2} is doomed with probability at least 
 \[
 1 - \left(\frac{101}{\eps^2 q}\right)^{T/2} \geq 1 - 2^{-\lambda}\qedhere
 \]
\end{proof}
Together,~\cref{claim:aux_rbr_1} and~\cref{claim:aux_rbr_2} give the round-by-round soundness of~\cref{poly iop: total rm constant rate}.
\end{proof}

\subsection{Compiling \cref{poly iop: total rm constant rate} and Proof of \cref{thm: tot rm constant rate}}
In this section we compile the result of~\cref{lm: poly iop tot rm constant rate} into an IOP as in~\cref{thm: tot rm constant rate}. As mentioned, we will test the bivariate polynomials $\wh{Q}_i$ over the domain $\Fqe^2$ so that the code we are testing proximity to, $\RM_{q'}[6d, \Fqe^2]$, has low rate.

Note that we do not require any subfield relations between the fields $\Ff_q$ and $\Fqe$. We only require that they are both subfields of the much larger $\Ff_{q'}$. Therefore there is still a degree preserving identification between $f|_{P_i}: \Ff_q^2 \to \Ff_{q'}$ to some $Q_i: \Fqe^2 \to \Ff_{q'}$. Specifically, starting from a degree $6d$ bivariate function over $\Ff_q^2$, $f|_{P_i}$, we can first take its low degree extension to, say $\wh{F}_i \in \Ff_{q'}^{\leq 6d}[x_1, x_2]$, and then the only $Q_i: \Fqe^2 \to \Ff_{q'}$ which makes sense is $Q_i = \wh{F}_i|_{\Fqe^2}$.

We now proceed with the description of the batched IOPP. Define the agreement parameter 
\[
\ec = 0.9\et,
\]
where recall $\et$ is as in~\eqref{eq: eps test definition}. Note that $\eps_{\comp} = \qenc^{-\Omega(1)}$. We keep $q, q', d, \lambda, \eps$ as before. To compile $\rmp(f, 6d, \lambda, \eps)$, we need to construct a batched code membership IOPP according to \cref{thm: poly iop transformation}. Specifically, for each $i \in [T]$, we have a side condition function $h_i: \Side_i \to \Ff_{q'}$ for the queries made to $\wh{Q}_i$ in \cref{poly iop: total rm constant rate}, and we have to test simultaneously the proximity of $T$ functions to the codes $\RM_{q'}[6d, \Fqe^2 \; | \; h_i]$. This is done in the batched code membership IOPP of \cref{lm: batch rm}. It is intended for the input $Q_i$'s to be obtained from $f|_{P_i}$ in the Poly-IOP as described above. 

We point out that in the initial state of \cref{lm: batch rm} below, we make the assumption that $T/2$ of the input functions are far from the code with side conditions. In contrast, all prior batched IOPPs only assumed that one of the input functions is far in the doomed case. This stronger assumption is necessary for us to achieve the stated input query complexity in \cref{lm: batch rm} as well as the stated proof query complexity in \cref{thm: tot rm constant rate}. Unfortunately, the doomed initial state here is not black-box compatible with the compilation requirements stated in \cref{thm: poly iop transformation}. However, unraveling the proof therein, one sees that \cref{lm: batch rm} is still sufficient for compiling the Poly-IOP from \cref{lm: poly iop tot rm constant rate} and described in \cref{poly iop: total rm constant rate}. We go over this analysis in more detail in the proof of \cref{thm: tot rm constant rate}.

\begin{lemma} \label{lm: batch rm}
    There is an IOPP, which we call $\trmb$, that has the following guarantees:
    \begin{itemize}
        \item {\bf Input:} functions $Q_i: \Fqe^2 \to \Ff_{q'}$ and side condition functions $h_i: \Side_i \to \Ff_{q'}$ as described, for $i \in [T]$.
        \item {\bf Completeness:} if $Q_i \in \RM_{q'}[6d, \Fqe^2 \; | \; h_i]$ for all $i \in [T]$, then the honest prover makes the verifier accept with probability $1$.
        \item {\bf Round-by-round soundness:} $2^{-\lambda}$.
        \item {\bf Initial state:} the initial state is doomed if and only if for at least $T/2$ of the $i \in [T]$ we have
        \[
        \agr(Q_i, \RM_{q'}[6d, \Fqe^2 \; | \; h_i]) \leq \ec.
        \]
        \item {\bf Round complexity:} $O(\log \log d)$.
        \item {\bf Input query complexity:} $O\left(\frac{\lambda^2}{\log^2(\qenc)}\right)$.
        \item {\bf Proof query complexity:} $O\left(\frac{\lambda^2}{\log^2(\qenc)}\log\log d\right)$.
        \item {\bf Length:} $O(\qenc^2 \cdot \log \log d)$.
    \end{itemize}
\end{lemma}

To construct $\trmb$, we follow a similar strategy to the batched IOPPs of \cref{sec: batch low rate}. The main difference is that in the current case, the smallest subcube containing $\Side_i$, has size $t^2 = O\left(\log(q)^2\right)$. Thus, the verifier would have to make $O(\log(q)^2)$ queries to read off the side condition function over a subcube, which is unaffordable to us. Instead, we will take advantage of the fact that by design, most of the side condition points are on the same affine line.

For each $i \in [T]$, let $\mc{T}_i: \Fqe^2 \to \Fqe^2$ be an invertible affine transformation which sends the $x$-axis to $L_i$. That is, $L_i \subseteq \{\mc{T}_i(\alpha, 0) \; | \; \alpha \in \Fqe \}$. Then note that $\mc{T}_i(\Side_i)$ is contained in a subcube $A_i \times B_i$ with $|A_i| = t+1$ and $|B_i| = 2$. This subcube has size $O(t)$ rather than $O(t^2)$. Consider the function $h'_i: \mc{T}_i^{-1}(\Side_i) \to \Ff_{q'}$ given by $h'_i(x) = h_i(\mc{T}_i(x))$. It is not hard to check that for an oracle function $g$ we have, 
\begin{equation} \label{eq: rmb equiv 1}
g \in \RM_{q'}[6d, \Fqe^2 \; | \; h'_i] \; \text{if and only if} \; g \circ \T_i \in \RM_{q'}[6d, \Fqe^2 \; | \; h_i],
\end{equation}
and 
\begin{equation}\label{eq: rmb equiv 2}
\agr(g,  \RM_{q'}[6d, \Fqe^2 \; | \; h'_i]) = \agr(g \circ \T_i, \RM_{q'}[6d, \Fqe^2 \; | \; h_i]).
\end{equation}

Therefore, at the start of the protocol, the verifier will send the prover invertible affine transformations $\mc{T}_i$ as above, and the parties will perform the equivalent task of testing proximity of $Q'_i = Q \circ \T^{-1}_i$ to  $\RM_{q'}[6d, \Fqe^2 \; | \; h_i])$. From here, the batched IOPP can proceed as in \cref{sec: batch low rate}.
\begin{proof}[Proof of \cref{lm: batch rm}]
We begin with a formal description of the batched IOPP:
\begin{algorithm}[H]\caption{Batched IOP for Total Degree Reed-Muller: $\texttt{tRMBatch}$} \label{iop: batch trm}
    \begin{algorithmic}[1]
        \STATE \textbf{V:} For each $i \in [T]$, the verifier sends an invertible affine transformation $\mc{T}_i: \Fqe^2 \to \Fqe^2$ which maps the $x$-axis to the line $L_i$ from $\trmb(f, 6d, \lambda, \eps)$.
        \STATE From now on let $h'_i$ be as described above and let $Q'_i = Q_i \circ \mc{T}^{-1}_i$.
         \STATE \textbf{P:} For each $i \in [T]$, the prover sends $h''_i: A'_i \times B'_i \to \Ff_{q'}$ as above that agrees with the side conditions, where $A'_i \times B'_i$ is the smallest subcube containing $\mc{T}^{-1}_i(\Side_i)$. The prover also sends a function $g_{i,1}: \Fqe^2 \to \Ff_{q'}$ for each $i \in [T]$.
        \STATE For each $i \in [T]$, let $\tilde{h}'_i: \Fqe^2 \to \Ff_{q'}$ be the degree $(t, 1)$ function over $\Fqe^2$ obtained by restricting the degree $(t,1)$ extension of $h''_i$ over $\Ff_{q'}^2$ to the domain $\Fqe^2$. Let
    \[
    g_{i,2} = \quo(Q'_i - V_{A'_i} \cdot g_{i,1} - \tilde{h}'_i, B'_i, \Fill_i).
    \]
    Recall that in the honest case we have, 
     \begin{equation} \label{eq: batch rm eq}
     Q'_i =   V_{A'_i} \cdot g_{i,1} + V_{B'_i} \cdot g_{i,2} + \tilde{h}'_i,
    \end{equation}
    $g_{i,1} \in \RM_{q'}[6d - |A'_i|, \Fqe^2]$, and $g_{i,2} \in \RM_{q'}[6d - |B'_i|, \Fqe^2]$.
    \STATE \textbf{V:} For each $i \in [T]$. The verifier 
    generates  $F_i = \combine_{6d}(g_{i,1}, g_{i,2})$ according to \cref{lm: combine multiple multivariate total degree}. The verifier also chooses $10$ lines, $L_{i,1}, \ldots, L_{i,10} \subseteq \Fqe^2$, and $\xi_{i,j} \in \Ff_{q'}$ according to the proximity generator for Reed-Solomon codes in \cref{thm: prox gen RS with corr} and sends these to the prover. The verifier sets
    \[
    G = \sum_{i = 1}^T \sum_{j=1}^{10} \xi_{i,j} \cdot F_i|_{L_{i,j}}.
    \]
    \STATE \textbf{P + V:} Both parties run $\rsi(G, 6d, \lambda, \eps_{{\sf test}})$ from \cref{thm: RS low rate}.
    \end{algorithmic}
\end{algorithm}
    The round complexity, input query complexity, and proof query complexity are clear. We discuss completeness and round-by-round soundness below.
    
    \paragraph{Completeness:} In the completeness case, the prover sends $h''_i$ which agrees with $Q'_i$ on each $A'_i \times B'_i$. As we are in the completeness case, each $Q'_i \in \RM_{q'}[6d, \Fqe^2 \; | \; h''_i]$, so by~\cref{lm: multivariate side condition vanish decomp} the honest prover may send $g_{i,1} \in \RM_{q'}[6d-|A'_i|, \Fqe^2], g_{i,2} \in \RM_{q'}[6d - |B'_i|, \Fqe^2]$ which satisfy~\eqref{eq: batch rm eq}. In this case, $F$ in step $5$ has degree at most $6d$, and the verifier accepts in $\rsi(G, 6d, \lambda, \eps_{k-1})$ with probability $1$ by the completeness of $\rsi$.

    \paragraph{Round-by-round soundness:} At the start of the protocol, there is one (trivial) round of interaction where the verifier sends the invertible transformations in step $1$. The state afterwards is as follows.
    \begin{state} \label{state: rmb state 1}
        The state is doomed if and only if at least $T/2$ of the $i \in [T]$ we have 
        \[
        \agr(Q'_i, \RM_{q'}[6d, \Fqe^2 \; | \; h'_i]) \leq \ec.
        \]
    \end{state}
    By~\eqref{eq: rmb equiv 2} we have that as long if the initial state is doomed, then \cref{state: rmb state 1} is doomed. This handles the soundness of the first round of interaction.
    
    After the first round of interaction (step 2), the state function as follows.
    \begin{state}
    The state is doomed if and only if 
    \[
    \agr(G, \RS_{q'}[6d, \Fqe]) \leq \eps_{{\sf test}}.
    \]
    \end{state}
    As the previous state is doomed, there is at least $T/2$ of the $i \in [T]$ satisfy
    \[
        \agr(Q'_i, \RM_{q'}[6d, \Fqe^2 \; | \; h'_i]) \leq \ec.
    \]
    Without loss of generality, let these indices be $i = 1, \ldots, T/2$. 
    
    For each $i \in [T/2]$, it follows by \cref{lm: total deg bivariate side condition decomp+combine} that for any $i \in [T/2]$, $g_{i,1}$ and $\Fill_i$ which the prover sends, we will have
    \[
    \agr(F_i, \RM_{q'}[6d, \Fqe^2]) \leq 1.04\ec,
    \]
    with probability at least $1 - 2^{-\lambda + 10}$. If this is the case, then \cref{thm: line vs point strong} implies that 
    \[
    \Pr_{L \subseteq \Fqe^2}[\agr(F_i|_L, \RS_{q'}[6d, \Fqe]) \geq 1.06 \eps_{\comp}] \leq \ec.
    \]
    Thus, over the random lines $L_{i, j}$ for $i \in [T/2], j \in [10]$, we get that with probability at least $1 - \eps_{\comp}^{5T} \geq 1 - 2^{-\lambda + 10}$, at least one of the lines $L_{i, j}$ satisfies, $\agr(F_i|_{L_{i,j}}, \RS_{q'}[d, \Fqe]) \leq 1.06 \eps_{\comp}$. Assuming this is the case, we have
    \[
     \agr(G, \RS_{q'}[6d, \Fqe]) \leq 1.06 \eps_{\comp} \leq \eps_{{\sf test}},
    \]
    with probability at least $1-  2^{-\lambda - 10}$ by \cref{thm: prox gen RS with corr}. Union bounding over the three mentioned events establishes soundness for this round.

    The soundness for the remainder of the IOPP follows from that of $\rsi$.
 \end{proof}
\begin{proof}[Proof of \cref{thm: tot rm constant rate}]
The desired IOPP is obtained by compiling the Poly-IOP in \cref{poly iop: total rm constant rate} using the batched IOPP in \cref{iop: batch trm}. Specifically, it is $\compile(\rmp(f, 6d, \lambda, \eps), \trmb(\eps_{\comp}))$ from~\cref{thm: poly iop transformation}.

For the completeness, the prover provides evaluations functions $Q_i: \Fqe^2 \to \Ff_{q'}$, which become the input to the batched IOPP, \cref{iop: batch trm}. In the honest case, each input $Q_i: \Fqe^2 \to \Ff_{q'}$ in $\trmb$ is the unique degree $6d$ function whose low degree extension to $\Ff_{q'}^2$ is the same as the low degree extension of $f|_{P_i}$ to $\Ff_{q'}^2$, where $P_i$ is the plane chosen in the \cref{iop: trm const rate poly iop}. The input query complexity, proof query complexity, and round complexity follow directly from \cref{thm: poly iop transformation} combined with \cref{lm: poly iop tot rm constant rate} and \cref{lm: batch rm}.

For the round-by-round soundness however, we cannot apply \cref{thm: poly iop transformation} in a black box manner. This is because \cref{thm: poly iop transformation} would require a batched IOPP where the initial state is doomed when at least one of the input functions $Q_i$ satisfies $\agr(Q_i, \RM_{q'}[6d, \Fqe^2 \; | \; h_i]) \leq \ec$. In contrast, the batched IOPP we construct in \cref{lm: poly iop tot rm constant rate} has a stronger requirement for the initial state to be doomed. There, the initial state is only doomed when at least $T/2$ of the $T$ input functions  $Q_i$ satisfy $\agr(Q_i, \RM_{q'}[6d, \Fqe^2 \; | \; h_i]) \leq \ec$. Fortunately, the weaker round-by-round soundness of \cref{iop: batch trm} shown in \cref{lm: batch rm} is still sufficient to get $2^{-\lambda}$ round-by-round soundness in the compiled IOP, $\compile(\rmp(f, 6d, \lambda, \eps), \trmb(\eps_{\comp}))$. 

To see this, we must open up the proof of round-by-round soundness in \cref{thm: poly iop transformation} and \cref{lm: poly iop tot rm constant rate}. Suppose that we start from the doomed initial state in $\rmp(f, 6d, \lambda, \eps)$. The analysis is the same up to step $3$ of \cref{thm: poly iop transformation}, where we instead perform a slightly stronger analysis. This stronger analysis shows that if the state is doomed prior to step $3$, then with probability at least $1-2^{-\lambda }$, the verifier will not only reject, but will also find at least $T/2$ of the indices $i$ where at least one of the checks fails. That is, there is at least one $j \in [t]$ such that
\begin{equation} \label{eq: stronger rbr neq}    
\wh{Q}_i(z_{i,j}) \neq f(P_i(z_{i,j})).
\end{equation}

We can then carry this condition over to the round-by-round soundness analysis of the compiled IOP, $\compile(\rmp(f, 6d, \lambda, \eps), \trmb(\eps_{\comp}))$. Specifically, each such $i$ translates to one $Q_i$ satisfying  $\agr(Q_i, \RM_{q'}[6d, \Fqe^2 \; | \; h_i]) \leq \ec$ in step (c) of the Decision and Test Phase of $\compile(\rmp(f, 6d, \lambda, \eps), \trmb(\eps_{\comp}))$. As a result, it is sufficient to give a batched IOPP where the initial state is doomed when least $T/2$ of the $T$ input functions $Q_i$ satisfy $\agr(Q_i, \RM_{q'}[6d, \Fqe^2 \; | \; h_i]) \leq \ec$. Such an IOPP is achieved in \cref{lm: batch rm}, and thus the final IOPP, $\compile(\rmp(f, 6d, \lambda, \eps), \trmb(\eps_{\comp}))$, has the desired round-by-round soundness.
\end{proof}

\section{IOP For $\rics$} \label{sec: iop for rics}
In this section we prove the main theorem which gives an IOP for $\rics$ with round-by-round soundness, thereby establishing~\cref{thm: main rbr},~\cref{thm: main rbr set} and~\cref{thm: main standard}. 
Towards this end we first recall the definition of the NP-Complete language $\rics$.

\begin{definition}[Rank-One Constraint Satisfaction]
An instance of {\sf R1CS} consists of three matrices $A, B, C \in \Ff_{q'}^{n \times n}$, with membership given as follows 
\begin{itemize}
    \item \textbf{Yes case:} $\exists v \in \Ff_{q'}^n$ such that $Av \; \odot \; Bv = Cv$ and $v_1 = 1$.
    \item \textbf{No case:} $\forall v \in \Ff_{q'}^n$ such that $v_1 = 1$, we have $Av \; \odot \; Bv \neq Cv$.
\end{itemize}
Here $v_1$ refers to the first coordinate of $v$, and $\odot$ refers to the coordinate-wise product operation.
\end{definition}

 Fix $n$ to be the size of the $\rics$ instance and let $\lambda$ be a security parameter. To construct the IOP, we will assume that we have the following ingredients.

\begin{itemize}
    \item The finite field $\Ff_{q'}$ is of characteristic $2$ with size 
    \[
    q' \geq 2^{\lambda + 10} \cdot n^{14/15},
    \]
    \item A multiplicative subgroup $H$ of $\Ff_{q'}$ with size 
    \[
    |H| = n^{1/3} - 1.
    \]
    \item Subfields $\Ff_q, \Fqe$ of $\Ff_{q'}$ with sizes satisfying
    \[
    q, \qenc \geq 6 \cdot (|H|-1),
    \]
    and
    \[
    q' \geq \frac{2^{\lambda + 10} \cdot \max(\qenc, q)^4}{n^{2/3}}
    \]
\end{itemize}
The motivation for introducing separate fields $q$ and $\qenc$ is towards the linear length application in \cref{thm: main rbr set}. For this setting it is important that $\qenc$ is polynomially larger than $n^{1/3}$, e.g. $\qenc = \Theta\left(n^{2/5}\right)$, so we encourage the reader to have this setting in mind. That being said, \cref{thm: main rbr} holds for any $q$ and $\qenc$ as above, so for the purpose of generality the remainder of the section proceeds without making this assumption.

\subsection{Multivariate Sumcheck}
Towards our IOP for $\rics$ we need a Poly-IOP for the sumcheck protocol. In this protocol, the verifier is given as input a polynomial $\wh{f} \in \Ff_{q'}^{\leq 6d}[x,y,z]$, for some degree parameter $6d$, and wishes to check whether $\sum_{b \in H^3} \wh{f}(b) = \alpha$ for some target sum $\alpha$ and 
a multiplicative subgroup $H$.

\begin{lemma} \label{lm: sumcheck poly iop}
    There is a Poly-IOP with the following guarantees:
    \begin{itemize}
        \item {\bf Input:} a polynomial $\wh{f} \in \Ff_{q'}^{\leq 6d}[x,y,z]$, a multiplicative subgroup $H \subseteq \Ff_{q'}$ of size $d+1$, and a target sum $\alpha$.
        \item {\bf Completeness:} If $\sum_{b \in H^3} \wh{f}(b) = \alpha$  then the honest prover makes the verifier accept with probability $1$.
        \item {\bf Round-by-round soundness:} $2^{-\lambda}$.
        \item {\bf Initial state:} the initial state is doomed if and only if $\sum_{b \in H^3} \wh{f}(b) \neq \alpha$.
        \item {\bf Number of polynomials:} $O(1)$.
        \item {\bf Round complexity:} $O(1)$.
        \item {\bf Input query complexity:} $1$
        \item {\bf Proof query complexity:} $O(1)$.
    \end{itemize}
\end{lemma}
\begin{proof}
    We defer the proof to \cref{app: rics proofs}.
\end{proof}

\subsection{Poly-IOP for $\rics$ with Round-by-Round Soundness}

We now describe our Poly-IOP for $\rics$, which we later compile into a proper IOP. Let us set $d = |H| - 1$ henceforth. In our description of the IOP, we assume there is an identification between $H^3$ and $[n]$. We will think of an assignment to the given $\rics$ instance, $v \in \Ff_{q'}^n$, as a map $v: [n] \to \Ff_{q'}$, which by identification can be thought of as a function $v: H^3 \to \Ff_{q'}$. We will repeatedly view vectors of length $n$ 
as functions over $H^3$ henceforth. In particular, for given a matrix $M \in \Ff_{q'}^{n \times n}$, we consider the vector $Mv$ and think of it as a function from $H^3$ to $\Ff_{q'}$.

\begin{thm} \label{thm: rics poly iop}
    There is a Poly-IOP for $\rics$, which we call $\ricsp$, that has the following guarantees:
    \begin{itemize}
        \item {\bf Input:} an instance of $\rics$.
        \item {\bf Completeness:} If the instance is satisfiable, then the honest prover makes the verifier accept with probability $1$.
        \item {\bf Round-by-round soundness:} $2^{-\lambda}$.
        \item {\bf Initial state:} The initial state is doomed if and only if the instance is not satisfiable.
        \item {\bf Number of polynomials:} $O(1)$.
        \item {\bf Round complexity:} $O(1)$.
        \item {\bf Proof query complexity:} $O(1)$.
    \end{itemize}
\end{thm}
\begin{proof}
The poly-IOP is described in~\Cref{poly iop: rics}.
\begin{algorithm}[]
\caption{Poly-IOP for {\sf R1CS}: $\ricsp$} \label{poly iop: rics}
    \begin{algorithmic} [1]
     \STATE \textbf{P:} For each $M \in \{A,B,C\}$, the prover sends the following 
     \begin{itemize}
         \item $\widehat{f}_M \in \Ff_{q'}^{\leq 3d}[x,y,z]$. In the honest case the prover has an assignment to the variables, namely $v$, in mind and each $\widehat{f}_M$ is the low degree extension of $Mv: H^3 \xrightarrow{} \Ff_{q'}$. 
         \item $\widehat{h}_{1}, \wh{h}_2, \wh{h}_3 \in \Ff_{q'}^{\leq 3d-1}[x,y,z]$. Define $\widehat{h} \in \Ff_{q'}^{\leq 3d}[x,y,z]$ by
    \[
    \widehat{h} = 1 + x \cdot \wh{h}_1 + y \cdot \wh{h}_2 + z \cdot \wh{h}_3
    \]
    In the honest case $\widehat{h}$ is the low degree extension $v$. Recall that it is assumed that the assignment to the variables, $v$, has $v(0) = 1$, so the prover sends $\wh{h}_i$ as above as a decomposition of the low degree extension of $v$ -- which should have constant term $1$.

    \item $\widehat{g}_{1}, \wh{g}_2, \wh{g}_3 \in \Ff_{q'}^{\leq 6d - |H|}[x,y,z]$. In the honest case,
    \begin{equation} \label{eq: zero check}
    \wh{f}_A \cdot \wh{f}_B - \wh{f}_C = V_H(x)\cdot \wh{g}_1+ V_H(y)\cdot \wh{g}_2+  V_H(z)\cdot \wh{g}_3.
    \end{equation}
     \end{itemize}
    \STATE \textbf{V. } The verifier chooses a random $\alpha \in \Ff_{q'}$ and sends this to the prover. The verifier also chooses a random $b\in \Ff^3_{q'}$ and checks if 
    \begin{equation} \label{eq: verifier Zero check} 
    \wh{f}_A(b) \cdot \wh{f}_B(b) - \wh{f}_C(b) =\wh{V}_{H}(b_1)\cdot \wh{g}_1(b)+ \wh{V}_{H}(b_2)\cdot \wh{g}_2(b)+  \wh{V}_{H}(b_3) \cdot \wh{g}_3(b).
        \end{equation}
    If the check fails then the verifier rejects and the protocol terminates.
    \STATE Both parties set $\widehat{w}_{\alpha}\in \Ff_{q'}^{\leq (d,d,d)}[x,y,z]$ to be the low degree extension of the function from $H^3 \to \Ff_{q'}$ that maps
    \[
    i \mapsto \alpha^{i-1}, \quad \forall i \in [n] \cong H^3.
    \]
    For each $M \in \{A,B,C\}$, let $\widehat{v}_{\alpha, M}\in \Ff_{q'}^{\leq (d,d,d)}[x,y,z]$ be the low degree extension of the function over $H^3$ given by the following matrix-vector product:
    \[
    M \times (1,\alpha, \ldots, \alpha^{n-1})^T.
    \]
    In other words, $\widehat{v}_{\alpha, M}$ is the low degree extension of the function that maps
    \[
    i \mapsto \sum_{j = 1}^n M_{i,j} \cdot \alpha^{j-1}, \quad \forall i \in [n] \cong H^3.
    \]
    \STATE \textbf{P+V:} For each $M \in \{A,B,C\}$, the verifier and prover perform the multivariate sumcheck protocol to check that:
    \[
    \sum_{b \in H^3} \wh{f}_M(b) \cdot \wh{w}_{\alpha}(b) - \wh{h}(b) \cdot \wh{v}_{\alpha,M}(b) = 0.
    \]
    \end{algorithmic}
\end{algorithm}
    The number of polynomials, round complexity, and proof query complexity are clear. We discuss the completeness here and save the round-by-round soundness for the next section.

    In the completeness case there exists $v \in \Ff_{q'}^n$ satisfying 
    \begin{equation} \label{eq: rics completeness}
    Av \odot Bv = Cv
    \end{equation}
    and $v_1 = 1$. Viewing $v$ as a map from $H^3 \to \Ff_{q'}$ with $v(0) = 1$, let $\wh{h} \in \Ff_{q'}^{\leq (d,d,d)}$ be its low degree extension. The honest prover then sends $\wh{f}_A, \wh{f}_B, \wh{f}_C \in \Ff_{q'}^{\leq 3d}[x,y,z]$ which are the low degree extensions of $Av, Bv, Cv$ and $\wh{h}_1, \wh{h}_2, \wh{h}_3 \in \Ff_{q'}^{\leq 3d-1}[x,y,z]$ which satisfy,
    \[
    \wh{h} = 1 + x\cdot\wh{h}_1 + y \cdot \wh{h}_2 + z\cdot \wh{h}_3.
    \]
    
    In this case, $f_A\cdot f_B - f_C$ is indeed $0$ over $H^3$ because of~\eqref{eq: rics completeness}. Therefore, the honest prover can find $\widehat{g}_{1}, \wh{g}_2, \wh{g}_3 \in \Ff_{q'}^{\leq 6d - |H|}[x,y,z]$ satisfying
    \begin{equation} \label{eq: zero check completeness rics}
    f_A \cdot f_B - f_C = \wh{V}_H(x)\cdot \wh{g}_1+ \wh{V}_H(y)\cdot \wh{g}_2+  \wh{V}_H(z)\cdot \wh{g}_3.
    \end{equation}
    As a result, in step $2$, for every $b \in \Ff^3_{q'}$ that the verifier may choose, the equation,
    \[
     \wh{f}_A(b) \cdot \wh{f}_B(b) - \wh{f}_C(b) 
     =\wh{V}_{H}(b_1)\cdot \wh{g}_1(b)+ \wh{V}_{H}(b_2)\cdot \wh{g}_2(b)+  \wh{V}_{H}(b_3) \cdot \wh{g}_3(b)
    \]
    will be satisfied. Finally, since $\wh{f}_M$ is the honest low degree extension of $Mv$ for each $M \in \{A,B,C\}$, for whichever $\alpha \in \Ff_{q'}$ that the verifier chooses in step 2, it will be the case that 
    \[
    \wh{f}_M \cdot w_{\alpha}(b) = \wh{h} \cdot v_{\alpha, M}(b),
    \]
    for all $b \in H^3$. Thus, from step 4 onward, the verifier will accept the honest prover with probability $1$ due to the completeness of the multivariate sumcheck protocol.
\end{proof}
\subsubsection{Round-by-round Soundness for \cref{thm: rics poly iop}}
In this section we show the round-by-round soundness of~\cref{poly iop: rics}. 
    There is only one round of interaction before the parties move onto the multivariate sumcheck protocol. After the prover sends the polynomials in step 1, the state is as follows. 
    
\vspace{1cm}

\begin{state} \label{state: rics poly}
    The state is doomed if and only if at least one of the following holds:
    \begin{itemize}
    \item The verifier's check in~\eqref{eq: verifier Zero check} fails. In this case the verifier rejects and the protocol terminates. Every state afterwards is doomed. 
        \item For some $M \in \{A,B,C\}$, we have $\sum_{b \in H^3} \wh{f}_M(b) \cdot \wh{w}_{\alpha}(b) - \wh{h}(b)\cdot \wh{v}_{\alpha, M}(b) \neq 0$.
    \end{itemize}
\end{state}

For any polynomial $\wh{f} \in \Ff_{q'}[x,y,z]$, we view $\wh{f}|_{H^3}$ as a length $n$ vector over $\Ff_{q'}$, and for an $M \in \Ff_{q'}^{n \times n}$, let $M\cdot \wh{f}|_{H^3}$ be the matrix-vector product. 

\begin{claim}
    If the intial state is doomed, meaning $\{A,B,C\}$ is not a satisfiable $\rics$ instance, then for any prover message in step 1, \cref{state: rics poly} will be doomed after step $2$ with probability at least $1 - 2^{-\lambda}$.
\end{claim}
\begin{proof}
    Note that for any prover messages in step $1$, at least one of the following must hold:
    \begin{itemize}
        \item $(\wh{f}_A \cdot \wh{f}_B - \wh{f}_C)|_{H^3} \neq 0$.
        \item For some $M \in \{A,B,C\}$, $\wh{f}_M|_{H^3} \neq M \cdot \wh{h}|_{H^3}$.
    \end{itemize}
    Indeed, if not, then one can check that $\wh{h}|_{H^3}$ gives a satisfying assignment to the $\rics$ instance $\{A,B,C\}$. 

    Let us suppose that it is the first item that holds then. In this case,~\eqref{eq: verifier Zero check} fails with probability at least $1 - 6d/q'$ by the Schwartz-Zippel lemma, and we are done.

    Now suppose that it is the second item that does not hold, and fix the $M$ for which we have inequality in the second item. Let $\wh{G}_1 \in \Ff_{q'}^{\leq |H|^3 - 1}[x]$ be the degree $|H|^3 - 1$ polynomial whose coefficients are the entries of the vector $\wh{f}_M$, and let  $\wh{G}_2\in\Ff_{q'}^{\leq |H|^3-1}[x]$ be the degree $|H|^3 - 1$ polynomial whose coefficients are the entries of the vector $M \cdot \wh{h}|_{H^3}$. By assumption $\wh{G}_1 \neq \wh{G}_2$. However, note that 
    \[
     \sum_{b \in H^3} \wh{f}_M(b) \cdot \wh{w}_{\alpha}(b) = \wh{G}_1(\alpha)\qquad \text{and} \qquad
     \sum_{b \in H^3} \wh{h}(b) \cdot \wh{v}_{\alpha, M}(b) = \wh{G}_2(\alpha).
    \]
    By the Schwartz-Zippel lemma, we have that $\wh{G}_1(\alpha) \neq \wh{G}_2(\alpha)$ with probability at least 
    \[
    1- (|H|^3-1)/q' \geq 1- 2^{-\lambda},
    \]
    and in this case \cref{state: rics poly} is doomed by the second item.
\end{proof}

We quickly wrap up the soundness in step 4.  If \cref{state: rics poly} is doomed by its first item, then the remainder of the protocol is doomed and we are done. If \cref{state: rics poly} is doomed by its second item, then we are in the doomed initial state of the multivariate sumcheck. The soundness of the rounds therein follow from \cref{lm: sumcheck poly iop}.

\subsection{A Batched Code Membership IOPP} \label{sec: compile rics}
Towards compiling the IOPP from~\cref{thm: rics poly iop} to a proper IOP we need to design an appropriate batched code membership IOPP, and we do
that in this section. It is straightforward to check that in \cref{poly iop: rics} (including the execution of the sum-check protocol in step 4), the prover sends:
\begin{itemize}
    \item $9$ multivariate polynomials $\wh{G}_i \in \Ff_{q'}^{\leq d_i}[x,y,z]$, where $d_i \leq 6d$.
    \item $8$ univariate polynomials $\wh{g}_j \in \Ff_{q'}^{\leq e_j}[x]$ where each $e_j \leq 6d$.
\end{itemize}
The $9$ multivariate polynomials from item $1$ are those sent by the prover in step $1$ of \cref{poly iop: rics}. The $8$ univariate polynomials from the second item are from the sumcheck in step $4$ \cref{poly iop: rics}. 

We will encode the $\wh{G}_i$'s over $\Ff_q^3$ and the $\wh{g}_j$'s over $\Fqe$. Recall that using these separate subfields is only important towards the linear length setting in \cref{thm: main rbr set}. We fix following agreement parameters:
\[
\eps_0 = 2\cdot \left(\frac{6d}{q}\right)^{1/2} 
\qquad \eps'_0 = 2\cdot \left(\frac{6d}{\qenc}\right)^{1/2},
\qquad \et = 23\left( \frac{6d}{q} \right)^{\tl/2},
\qquad \et' = \left( \frac{12d}{\qenc} \right)^{\tl/2}.
\]
In particular, $\et$ is large enough to be an applicable proximity parameter in \cref{thm: tot rm constant rate} and $\et'$ is large enough to be an applicable proximity parameter in \cref{thm: RS low rate}.

The maximum number of times any polynomial from \cref{poly iop: rics} is queried is $2$, so each of the polynomials has a constant sized side condition function. Specifically, for each $\wh{G}_i$, $i \in [9]$ we have a side condition function $H_i: A_i \times B_i \times C_i \to \Ff_{q'}$, and for each $\wh{g}_j$, $j \in [8]$ we have a side condition function $h_j: A'_j \to \Ff_{q'}$. Here, $A_i \times B_i \times C_i$ is a $3$ by $3$ by $3$ subcube containing the points queried to $\wh{G}_i$ in \cref{poly iop: rics}, and $A'_j$ is the set of points queried to $\wh{g}_j$ in \cref{poly iop: rics}. In order to compile \cref{poly iop: rics}, we must construct a batched code membership IOPP for the codes $\RM_{q'}[d_i, \Ff_q^3 \; | \; H_i]$, and $\RS_{q'}[e_j, \Fqe \; | \; h_j]$.

In the batched code membership IOPP we will use the IOPPs $\rmi(G^\star, 6d, \lambda, \et)$ from \cref{thm: tot rm constant rate} and $\rsi(g^\star, 6d, \lambda, \et')$ from \cref{thm: RS low rate}.  We will not recall all of the parameters, but for convenience let us restate their initial states, which will be important for our round-by-round soundness analysis.
\begin{state} \label{state rm final use}
 The initial state of $\rmi(G^\star, 6d, \lambda, \et)$ is doomed if and only if the input function $G^\star: \Fqe^3 \to \Ff_{q'}$ satisfies
 \[
 \agr(G^\star, \RM_{q'}[6d, \Ff_q^3]) \leq \et.
 \]
\end{state}
 \begin{state} \label{state: rs final use}
 The initial state of $\rsi(g^\star, 6d, \lambda, \et')$ is doomed if and only if the input function $g^\star: \Ff_q \to \Ff_{q'}$ satisfies
 \[
 \agr(g^\star, \RS_{q'}[6d, \Fqe]) \leq \et'.
 \]
\end{state}
 Using these IOPPs, we construct a batched code membership IOPP satisfying the following.

\begin{lemma} \label{lm: iop rics batch}
There is a batched code-membership IOPP, which we call $\texttt{R1CSBatch}$, that satisfies the following properties:
\begin{itemize}
    \item Input: functions $G_i: \Ff_q^3 \to \Ff_{q'}$ with side condition functions $H_i: A_i \times B_i \times C_i \to \Ff_{q'}$ for $i \in [9]$, and functions $g_j: \Fqe \to \Ff_{q'}$ with side condition functions $h_j: A'_j \to \Ff_{q'}$ for $j \in [8]$.
    \item Completeness: if $G_i \in \RM_{q'}[d_i, \Ff_q^3 \; | \; H_i]$ and $g_i \in \RS_{q'}[e_j, \Fqe \; | \; h_j]$ for all $i \in [9], j \in [8]$, then the honest prover makes the verifier accept with probability $1$.
    \item Round-by-round soundness: $2^{-\lambda}$.
    \item Initial state: the initial state is doomed if and only if $
    \agr(G_i, \RM_{q'}[d_i, \Ff_q^3 \; | \; H_i]) \leq \eps_0$ 
    for some $i\in [9]$, or $\agr(g_j, \RS_{q'}[e_j, \Ff_q \; | \; h_j]) \leq \eps'_0$ for some $j \in [8]$.
    \item Round complexity: $O(\log \log d)$.
    \item Input query complexity: $O\left(\frac{\lambda}{\log(1/\eps_0)} \right)$.
    \item Proof query complexity: $O\left(\frac{\lambda}{\log(1/\eps_0)} \right) + O\left(\frac{\lambda^2}{\log^2(n)} \cdot \log \log(n)\right)$.
    \item Length: $O(q^3)$.
\end{itemize}
\end{lemma}
\begin{proof}
The IOP is described in \cref{iop: rics batch}:
    \begin{algorithm}[]
\caption{$\texttt{R1CSBatch}$} \label{iop: rics batch}    
\begin{algorithmic} [1]
    \STATE \textbf{P:} The prover sends the following functions for every $i \in [9], j \in [8]$:
    \begin{itemize}
        \item $G_{i,1}, G_{i,2}: \Ff_{q}^3 \to \Ff_{q'}$.
        \item $\Fill_i: \Ff_{q} \times \Ff_q \times C_i \to \Ff_{q'}$.
        \item $\Fill'_j: A'_i \to \Ff_{q'}$.
    \end{itemize}
    \STATE For each $i \in [9], j \in [8]$, let $\Tilde{H}_i: \Ff_q^3 \to \Ff_{q'}$ be the restriction to $\Ff^3_q$ of the low degree extension of $H_i$ to $\Ff^3_{q'}$ and let $\Tilde{h}_j: \Ff_q \to \Ff_{q'}$ be the restriction to $\Ff_q$ of the low degree extension of $h_j$ to $\Ff_{q'}$. Also let 
    \[
    G_{i,3} = \Quo_3(G_i - V_{A_i}G_{i,1} - V_{B_i}G_{i,2} - \Tilde{H}_j, C_i, \Fill_i),
    \]
    and
    \[
    g'_j = \Quo_1(g_j - \Tilde{h}_j, A'_j).
    \]
    \STATE \textbf{V:} For each $i \in [9]$, the verifier generates
    \[
    G^\star_i = \combine_{d_i}(G_{i,1}, G_{i,2}, G_{i,3}),
    \]
    with starting degrees $d_i - |A_i|, d_i - |B_i|, d_i - |C_i|$, according to \cref{def: multi combine tot}. The verifier chooses random coefficients from $\xi_1, \ldots, \xi_9 \in \Ff_{q'}$ according to the proximity generator from \cref{thm: prox gen RM without corr} and sets
    \[
    G^\star = \sum_{i=1}^9 \xi_i \cdot G^\star_i.
    \]
    The verifier also generates
    \[
    g^\star = \combine_{6d}(g'_1, \ldots, g'_j),
    \]
    with starting degrees $e_1 - |A'_1|, \ldots, e_j - |A'_j|$ according to \cref{def: uni combine}.
    \STATE \textbf{P+V:} Both parties run $\rmi(G^\star, 6d, \lambda, \et)$ from \cref{thm: tot rm constant rate} and $\rsi(g^\star, 6d, \lambda, \et')$ from \cref{thm: RS low rate}.
\end{algorithmic}
\end{algorithm}
    The round complexity, input query complexity, and proof query complexity are straightforward to see. We discuss the completeness below.
    
     In the completeness case, for each $i \in [9]$ there exist $\wh{G}_{i,1},  \wh{G}_{i,2},  \wh{G}_{i,3}$ with degrees $d_i - |A_i|, d_i - |B_i|, d_i - |C_i|$ respectively, which satisfy
    \begin{equation}  \label{eq: honest prover rics batch}
    \wh{G}_i = V_{A,1}\cdot \wh{G}_{i,1} +  V_{B,1}\cdot \wh{G}_{i,2} +  V_{C,2}\cdot \wh{G}_{i,3} + \wh{H}_i.
    \end{equation}
    Here $\wh{H}_i$ is the low degree extension of $H_i$ over $\Ff_{q'}^3$. The honest prover gives $G_{i,1}, G_{i,2}$ which are the evaluations of $\wh{G}_{i,1},  \wh{G}_{i,2}$ from~\eqref{eq: honest prover rics batch}, and $\Fill_i$ which is the evaluation of $\wh{G}_{i,3}$ from~\eqref{eq: honest prover rics batch}. For each $j \in [8]$, $g'_j$ defined in step 2 of \cref{iop: rics batch} is the evaluation of $\wh{g'}_j = \frac{\wh{g}_j - \wh{h}_j}{\wh{V}_{A'_j}}$ where $\wh{g'}_j \in \Ff^{\leq e_i - |A'_i|}_{q'}$. With probability $1$, $g^\star$ is a function over $\Ff_q$ with degree at most $6d$ and $G^\star$ is a function over $\Ff_q^3$ with degree at most $6d$. The remainder of the completeness follows from the completeness of $\rmi$ and $\rsi$.
\end{proof}
\subsubsection{Round-by-Round Soundness for \cref{lm: iop rics batch}}
There is only one round of interaction, after which the state is as follows.
\begin{state} \label{state: rics batch}
The state is doomed if and only if at least one of the following holds:
\begin{itemize}
    \item $\agr(G^\star, \RM_{q'}[6d, \Ff_q^3]) \leq \et$
    \item $\agr(g^\star, \RS_{q'}[6d, \Ff_q]) \leq \et'$
\end{itemize}
\end{state}
Soundness of this round is established by the following claim.
\begin{claim}
    If the initial state is doomed as described in \cref{lm: iop rics batch}, then \cref{state: rics batch} is doomed with probability at least $1-2^{-\lambda}$.
\end{claim}
\begin{proof}
    As the initial state is doomed, by \cref{lm: total deg trivariate side condition decomp+combine}, we know that one of the following holds:
    \begin{itemize}
    \item For at least one $i \in [9]$, 
         \[
    \agr(G_i, \RM_{q'}[d_i, \Ff_q^3 \; | \; H_i]) \leq \eps_0,
    \]
    \item For at least one $j \in [8]$,
    \[
    \agr(g_i, \RS_{q'}[e_i, \Ff_q \; | \; h_i]) \leq \eps'_0.
    \]
    \end{itemize}
    Suppose it is the first item, and let $i$ be the index for which it holds. Then for any $G_{i,1}, G_{i,2}, \Fill_i$ which the prover sends, by~\cref{lm: combine multiple multivariate total degree} we have
    \[
    \agr(G^\star_i, \RM_{q'}[d_i, \Ff_q^3]) \leq 1.04\eps_0
    \]
    with probability at least $1 - 2^{-\lambda - 10}$. Assuming that this is the case, by \cref{lm: combine multiple multivariate total degree}, we have
    \[
    \agr(G^\star, \RM_{q'}[6d, \Ff_q^3]) \leq 1.1\eps_0 \leq \et.
    \]
    with probability at least $1 - 2^{-\lambda - 10}$. Union bounding over these two events, we get the desired soundness in the case that the first item holds.

    Now suppose that it is the second item which holds and let $j$ be the index for which it holds. By~\cref{lm: uni quo soundness} we get
    \[
    \agr(g'_j, \RS_{q'}[e_i - |A'_i|, \Ff_q]) \leq \eps'_0 + \frac{|A'_i|}{q} \leq 1.04\eps'_0.
    \]
    It follows from \cref{lm: combine multiple univariate}, that with probability at least $1-2^{-\lambda-10}$ we have 
    \[
    \agr(g^\star, \RS_{q'}[6d, \Ff_q]) \leq 1.1\eps'_0 \leq \et'.
    \]
    This establishes the desired soundness in the case that the second item holds.
\end{proof}
Finally, to complete the proof of round-by-round soundness, note that if \cref{state: rics batch} is doomed, then the initial state of either $\rmi(G^\star, 6d, \lambda, \et)$ or $\rsi(g^\star, 6d, \lambda, \et')$ is doomed in step 4. The remainder of the round-by-round soundness follows follows from the round-by-round soundness of $\rmi(G^\star, 6d, \lambda, \et)$ and $\rsi(g^\star, 6d, \lambda, \et')$.

\subsection{Proof of \cref{thm: main rbr}}

The IOP satisfying \cref{thm: main rbr} is $\compile(\ricsp, \texttt{R1CSBatch})$. The agreement parameters used in compiling are $\eps_0$ for the trivariate polynomials and $\eps'_0$ for the univariate polynomials. The agreement parameters used in $\texttt{R1CSBatch}$ are $\et$ and $\et'$ as described in \cref{sec: compile rics}.

Using \cref{lm: iop rics batch}, one can check that $\texttt{R1CSBatch}$ satisfies the requirements of \cref{thm: poly iop transformation} and that the parameters of $\compile(\ricsp, \texttt{R1CSBatch})$ are indeed those described by \cref{thm: main rbr}.

\subsection{Proof of \cref{thm: main rbr set}}

We must apply \cref{thm: main rbr} with the appropriate setting of $q, \qenc,$ and $q'$. Specifically, we need $q = \Theta(n^{1/3})$ and $\qenc = \Theta(n^{2/5})$. To achieve this setting, take an integer $s$ such that $2^{3s} - 1$ is divisible by $100$, set $q = 2^{3s}$, and take $H \subseteq \Ff_{q}$ to be the multiplicative of $\Ff_q$ of size $\frac{q - 1}{100}$. Note that $H$ is a multiplicative subgroup of $\Ff_{q'}$ as well. Then, set $\qenc = 2^{5 \cdot s}$, and $q' = 2^{\lambda + c \cdot s}$ for some appropriate $s$ and large enough, but constant, $c$ which makes $\lambda + c \cdot s$ divisible by $15s$. It is clear that this setup can be conceived with arbitrarily large $s$ and thus can be used for arbitrarily large instances of $\rics$.

\subsection{Proof of \cref{thm: main standard}}

We again apply \cref{thm: main rbr} with the appropriate parameters, but we do not have to be as careful here as we are not concerned with linear length. Take $\qenc = q = 2^{2^s}$, $q' = 2^{2^{s+100}}$, and take $H$ to be the multiplicative subgroup of $\Ff_q$ of size $2^{2^{s-1}} - 1 \leq \sqrt{q}$. Once again it is clear that this setup can be conceived with arbitrarily large $s$ and thus can be used for arbitrarily large instances of $\rics$.

\section*{Acknowledgements}
We sincerely thank Eylon Yogev for telling us about IOPs and for pointing us to~\cite{acy}, for helpful comments on an earlier version of this paper and for providing us various references.

\bibliographystyle{alpha}
\bibliography{references}
\appendix
\section{The Line versus Point Test} \label{sec: line vs point}

Here we give our analysis of the line versus point test. In our setting, both the points and the lines table have inputs from $\Ff_q$, while the outputs are in a larger field $\Ff_{q'}$ which contains $\Ff_q$ as a subfield. Recall that \cite{hkss} analyze the $q = q'$ case. Their proof goes through in the more general setting that we require as well, but we go over the details here for the sake of completeness. 


Throughout this section, fix two finite fields $\Ff_q$ and $\Ff_{q'}$ such that $\Ff_q$ is a subfield of $\Ff_{q'}$, and let $d < q$ be a degree parameter. We let $\mc{L}$ denote the set of affine lines in the ambient space in either $\Ff_q^2$. We only discuss the $m = 2$ and $m = 3$ cases. Let $\mc{L}_a$ be the set of affine lines going through the point $a$. We will prove the following two theorems.

\begin{thm} \label{thm: line versus point final redo}
Suppose $f: \Ff_q^2 \to \Ff_{q'}$ and degree $d$ lines table $T$ satisfy
\[
\Pr_{L \subseteq \Ff_q^2}[f|_L = T[L]] = \eps,
\]
for some $\eps$ such that $\eps^6 \geq \sqrt{d/q}$. Then, there is a degree $d$ oracle function $Q: \Ff_q^2 \to \Ff_{q'}$ such that 
\[
\Pr_{a \in \Ff_q^2}[Q(a) = f(a)] \geq \eps - 1.01\left(\frac{d}{q}\right)^{\tl}.
\]
\end{thm}
\begin{proof}
Deferred to~\cref{sec:list_decode}, after establishing some preliminary results.
\end{proof}

Starting with \cref{thm: line versus point final redo}, we can use standard boosting techniques to obtain a similar result for all dimensions $m > 2$. This is done in \cite{hkss} and comes at the price of some loss in the exponent $\tau$ in \cref{thm: line versus point redo}. For our purposes, we only need the $m = 3$ case, so we give a self contained boosting argument, which obtains a better exponent $\tau$ for the $m = 3$ case.

\begin{thm} \label{thm: line versus point m = 3}
Suppose $f: \Ff_q^3 \to \Ff_{q'}$ and degree $d$ lines table $T$ satisfy
\[
\Pr_{L \subseteq \Ff_q^2}[f|_L = T[L]] = \eps,
\]
for some $\eps$ such that $\eps^6 \geq \sqrt{d/q}$. Then, there is a degree $d$ function $Q: \Ff_q^3 \to \Ff_{q'}$ such that 
\[
\Pr_{a \in \Ff_q^3}[Q(a) = f(a)] \geq 0.49\eps - 1.01\left(\frac{d}{q}\right)^{\tl}.
\]
\end{thm}
\begin{proof}
 Deferred to~\cref{sec:boost}.
\end{proof}

\subsection{Preliminaries for the Line versus Point Test}
We will need the following two lemmas. To present them, we must first define a couple of notions, the first of which is weighted degrees.
\begin{definition}
For a vector $w\in\mathbb{N}^{m}$ 
and a polynomial $f\colon\mathbb{F}_q^m\to\mathbb{F}_{q'}$ we define
the $w$-weighted degree of
$f$ as $\sum\limits_{i=1}^m w_i d_i$ where $d_i$ is the individual degree of $x_i$ in $f$.
\end{definition}

Next, we define discriminants of univariate polynomials.
\begin{definition}
For two polynomials $A(x) = \sum\limits_{i=0}^{a} A_i x^i$ and $B(x) = \sum\limits_{i=0}^{b} B_i x^i$, the Sylvester matrix ${\sf Sylvester}(A,B)$ of $A$ and $B$ is defined as the $(a+b)$ by $(a+b)$ matrix whose first $b$ rows are the coefficients vector of $A$ $(A_0,\ldots,A_{a},0,\ldots,0)$and its shifts $(0,A_0,\ldots,A_{a},0,\ldots,0)$ up to $(0,\ldots,0,A_0,\ldots,A_{a})$, and whose last $a$ rows are the coefficient vector of $B$, $(B_0,\ldots,B_{b},0,\ldots,0)$ and its shifts $(0,B_0,\ldots,B_{b},0,\ldots,0)$ up to $(0,\ldots,0,B_0,\ldots,B_{b})$. The resultant of $A$ and $B$ is defined as the determinant of ${\sf Sylvester}(A,B)$. In the special case that $B = \partial_{x} A$, we denote
$\Disc_{x}(A) = {\sf Sylvester}(A,B) = {\sf Sylvester}(A,\partial_x A)$.
\end{definition}

We are now ready to state the two results we need from~\cite{hkss}.
\begin{lemma} \label{lm: nonzero disc aux}
Let $S \subseteq \Ff_{q'}^m$ be a set of points and let 
\[
\mc{Q}_{S,x_1} = \{Q \in \Ff_{q'}^m[x_1,\ldots, x_m] \; : \; Q(a) = 0, \; \forall a \in S, \; \partial_{x_1}(Q) \neq 0 \}.
\]
Then for any weight vector $w \in \mathbb{N}^m$, the polynomial $Q \in \mc{Q}_{S,x_1}$ of lowest $w$-weighted degree satisfies $\Disc_{x_1}(Q) \neq 0$ and is therefore square-free. 
\end{lemma}
\begin{proof}
    The proof is given in~\cite[Lemma 2.9]{hkss}.
\end{proof}
\begin{lemma}\label{lem:bur}
    Let $A \in \Ff_{q'}^{m+1}[\vec{x}, z]$ be an $m+1$-variate polynomial and let $\alpha \in \Ff_{q'}$ be a zero of multiplicity one. Then, for every $k \geq 1$, there is a unique $m$-variate polynomial $\Phi_x(\vec{x})$ which satisfies the following:
    \begin{itemize}
        \item $A(\vec{x}, \Phi_k(\vec{x})) \equiv 0 \mod \langle x_1,\ldots, x_m \rangle^{k+1}$,
        \item $\alpha = \Phi_k(0) \equiv \Phi_k(x) \mod \langle x_1,\ldots, x_m \rangle$,
        \item $\Phi_k$ has total degree at most $k$.
    \end{itemize}
\end{lemma}
\begin{proof}
    See \cite{bur}.
\end{proof}

\subsection{The Basic Decoding Result}\label{sec:proof_of_basic_lp}

The first step towards showing \cref{thm: line versus point final redo} is to show the following, weaker version in \cref{thm: line versus point redo}. It is weaker in
the sense that the guaranteed agreement between $f$ and $Q$ 
is smaller:
\begin{thm} \label{thm: line versus point redo}
Suppose $f: \Ff_q^2 \to \Ff_{q'}$ and degree $d$ lines table $T$ satisfy
\[
\Pr_{L \subseteq \Ff_q^2}[f|_L = T[L]] = \eps \geq \Omega\left(\frac{d}{q}\right)^{\tl}.
\]
Then, there is a degree $d$ function $Q: \Ff_q^2 \to \Ff_{q'}$ such that 
\[
\Pr_{a \in \Ff_q^2}[Q(a) = f(a)] \geq \Omega\left(\eps^6\right).
\]
\end{thm}

The rest of this section is devoted to the proof of \cref{thm: line versus point redo}. 
We will need the following lemma, analogous to~\cite[Lemma 3.1]{hkss}, which gives a local to global type statement. At a high level, it says that if for many lines in $L \in \mc{L}_a$ the polynomial $A(L(t), z) \in \Ff_{q'}[t, z]$ has a local low degree root, $T[L](t)$, then $A(x_1, \ldots, x_m, z)$ actually has a global low degree root $P(x_1,\ldots, x_m)$.

\begin{lemma} \label{lm: line vs point newton}
    Let $A(x_1,\ldots, x_m, z) \in \Ff_{q'}[x_1,\ldots, x_m, z]$ have $(1,\ldots, 1, d)$-degree at most $D$. Suppose there is a point $b \in \Ff_q^2$ such that the following holds:
    \begin{itemize}
        \item $A(b,z) \in \Ff_{q'}[z]$ has no repeated roots.
        \item There is a set of affine lines through $b$, $\mc{S} \subseteq \mc{L}_b$ of relative size $\mu(\mc{S}) > \frac{D}{q}$, such that for every line $L \in \mc{S}$, we have,
        \[
        A\left(L(t), T[L](t)\right) \equiv 0,
        \]
        and $T[L](0) = \alpha$, for some $\alpha \in \Ff_{q'}$ and all $L \in \mc{S}$. 
    \end{itemize}
    Then there exists $P \in \Ff^{\leq d}_{q'}[x_1,\ldots, x_m]$ such that $A(x, P(x)) \equiv 0$. Moreover, $P|_L \equiv T[L] \quad \forall L \in \mc{S}$.
\end{lemma}
\begin{proof}
    We may assume without loss of generality that $b = 0$, as otherwise we can translate $A$. This means that $L(0) = 0$ for every $L \in \mc{S}$. By assumption that $A(0, z)$ has no repeated roots, so we may apply \cref{lem:bur} and get that for every $k \geq 0$, there is a polynomial $\Phi_k(x) \in \Ff_{q'}[x_1,\ldots, x_m]$ satisfying
    \[
    \Phi_k(0) = \alpha \quad \text{and} \quad A(\vec{x}, \Phi_k(\vec{x})) = 0 \mod \langle \vec{x} \rangle^{k+1}.
    \]
    As a consequence, 
    \[
    A(L(t), \Phi_k \circ L(t)) = 0 \mod t^k, \quad \forall L \in \mc{S}.
    \]

    Fix some $L \in \mc{S}$ and consider $A'(t,z) = A(L(t), z) \in \Ff_{q'}[t,z]$. By assumption $A'(0, \alpha) = 0$ and $A'(0,z)$ is square-free as a polynomial in $z$. Thus, by \cref{lem:bur}, there is a unique polynomial $\Phi_L(t) \in \Ff_{q'}^{\leq d}[t]$ such that $\Phi_L(0) = \alpha$ and 
    \[
    A'(t, \Phi_L(t)) = A(L(t), \Phi_L(t)) = 0 \mod t^{d+1}.
    \]
    However, note that both $T[L]$ and $\Phi_d \circ L(t)$ satisfy the above. By the uniqueness in \cref{lem:bur} this means
    \[
    T[L](t) = \Phi_d \circ L(t) \quad \forall L \in \mc{S}.
    \]

    To conclude, we must show that $B(\vec{x}) = A(\vec{x}, \Phi_d(\vec{x})) = 0$. Since $\Phi_d(\vec{x})$ has degree at most $d$ and $A$ has $(1, \ldots, 1, d)$-degree at most $D$, we have that $B(\vec{x})$ has total degree at most $D$. For every $L \in \mc{S}$, let $u$ be the direction of the line. We have, $B(u) = A(u, \Phi_d(u)) = 0$. Thus, $B$ has more than $D \cdot q^{m-1}$ zeroes and therefore $B$ is the zero polynomial.
\end{proof}

The next theorem gives an interpolation result assuming that $f$ and $T$ pass the line versus point test with probability $\eps$ as in \cref{thm: line versus point redo}. In particular, we interpolate a polynomial $A \in \Ff_{q'}[x,y,z]$ which has low weighted degree and other helpful agreement properties with $T$ and $f$. We will ultimately apply \cref{lm: line vs point newton} to this polynomial $A$ to conclude \cref{thm: line versus point redo}.

\begin{thm} \label{thm: line vs point interpolation}
    Fix $\eps$ such that $\eps \geq \left(c_1 d/q\right)^{1/c_2}$ for some constants $c_1, c_2$. Suppose $f: \Ff_{q}^2 \to \Ff_{q'}$ and lines table $T$ pass the line versus point test with probability $\eps$. Then there is a non-zero polynomial $A \in \Ff_{q'}[x,y,z]$ and a set of points $S \subseteq \Ff_{q}^2$ such that
    \begin{itemize}
        \item The $(1,1,d)$-degree of $A$ is $O(d/\eps^2)$,
        \item $\mu(S) \geq \Omega(\eps^2)$,
        \item For every $x \in S$, $\eps_x :=\Pr_{L\ni x}[T[L](x) = f(x)] \geq \frac{\eps}{2}$.
        \item For every $x \in S$, $A(x, f(x)) = 0$.
        \item $\partial_z(A)$ and $\Disc_z(A)$ are nonzero polynomials.
    \end{itemize}
\end{thm}
\begin{proof}
    We defer the proof to \cref{sec: proof of lvp interpolation}.
\end{proof}

Applying \cref{thm: line vs point interpolation}, we get that there is a set of points $S \subseteq \Ff_q^2$ of measure $\mu(S) \geq \Omega(\eps^2)$ and a nonzero polynomial $A\in \Ff_{q'}[x,y,z]$ with $(1,1,d)$-degree at most $D = O(d/\eps^2)$ such that the polynomials $\partial_z(A)$ and $\Disc_z(A)$ are nonzero, and the following holds for all $a \in S$:
\begin{itemize}
    \item $A(a, f(a)) = 0$,
    \item $\eps_a \geq \frac{\eps}{2}$
\end{itemize}
From this set $S$, we will find a point $a \in S$ which satisfies some useful properties.
\begin{claim} \label{cl: point with good lines}
    There exists a point $a \in S$ and $\eta \geq \Omega(\eps^3)$ such that the following holds:
    \begin{itemize}
        \item $\Disc_z(A)(a) \neq 0$,
        \item For $\Omega(\eps^3)$-fraction of lines $L \in \mc{L}_a$, we have $A(L(t), T[L](t)) = 0$, $T[L](a) = f(a)$, and $T[L](b) = f(b)$ for $\eta$-fraction of $b \in L$.
    \end{itemize}
\end{claim}
\begin{proof}
    Consider the bipartite graph where the left side consists of all points $\Ff_q^2$, the right side consists of all affine lines, $\mc{L}$, and the edges consist of all pairs $(x,L)$ such that $x \in L$, $x \in S$, and $T[L](x) = f(x)$. Note that by \cref{thm: line vs point interpolation}, we have that each $x \in S$ has $\eps_x \cdot (q+1)$-neighbors. 

    For each line $L$, let $\sigma_L = \Pr_{x \in L}[(x,L) \in E]$. We have 
    \[
    \E_L [\sigma_L] = \Pr_{x, L \ni x}[(x, L) \in E] \geq \Omega(\eps^3).
    \]
    By an averaging argument, for at least $\Omega(\eps^3)$-fraction of lines $L$, we have $\sigma_L \geq \Omega(\eps^3)$. Let $\mc{L}'$ denote these lines. Now we wish to find a set of points $S'$ that each have many neighbors in $\mc{L}'$. Let $E'$ be the set of pairs $(x,L)$ such that $(x, L) \in E$, and $L \in \mc{L}'$. We have
    \[
    \Pr_{x, L \ni x}[(x, L) \in E'] \geq \Omega(\eps^3).
    \]
   By another averaging argument, we get that for $\Omega(\eps^3)$-fraction of points $x$, we have $x \in S$ and $\Pr_{L \ni x}[(x, L) \in E'] \geq \Omega(\eps^3)$. We let $S'$ be the set of such $x$, so that $\mu(S') \geq \Omega
\left(\eps^3\right)$.
    
    We now claim that some $a \in S'$ is the desired point. 
   Towards the first item, one can check that $\Disc_z(A)$ is the determinant of an $O(1/\eps^2) \times O(1/\eps^2)$-matrix whose entries are polynomials in $x,y$ of degree at most $D$. Moreover, by the last item of \cref{thm: line vs point interpolation}, $\Disc_z(A)$ is nonzero. Therefore, $\Disc_z(A)$ has degree at most $O(d/\eps^2)$, and in particular,
   \[
   \Pr_{a \in \Ff_q^2}[\Disc_z(A)(a) = 0] \leq \frac{d}{q \cdot \eps^2} < \Omega(\eps^3),
   \]
   where the last inequality holds provided that $c_2 > 6$.
   It follows that $\Disc_z(A)(a) \neq 0$ for at least one $a \in S'$, and we fix such an $a$ for the rest of the argument.
   
   Towards the second item, fix $L$ such that $(a, L) \in E'$. Since $a \in S$, we have $T[L](a) = f(a)$. Now let $B(t)$ be the univariate polynomial given by $A(L(t), T[L](t))$. As the $(1,1,d)$-degree of $A$ is at most $D$, the polynomial $B$ has degree at most $D = O(d/\eps^2)$. On the other hand, note that $B(t)$ is zero at every point $x$ such that $(x, L) \in E$. Indeed, for such points we have $T[L](x) = f(x)$ and $A(x, f(x))=0$. By construction, there are at least $\Omega(\eps^3) \cdot q$ such points. Since
   \[
   D = O\left(\frac{d}{\eps^2}\right) < \Omega(\eps^3)\cdot q,
   \]
   we get that $B$ must be identically $0$. Finally, we have that $T[L](x) = f(x)$ for all $x$ such that $(x, L) \in E$, which by construction is at least $\Omega(\eps^3)$-fraction of $L$. It follows that the conditions of item 2 are satisfied for all $L$ such that $(a, L) \in E'$, and this establishes the claim.
\end{proof}
Having established \cref{cl: point with good lines}, we are ready to conclude the proof of \cref{thm: line versus point redo}
\begin{proof}[Proof of \cref{thm: line versus point redo}]
    Suppose $A \in \Ff_{q'}[x,y,z]$, $a \in \Ff_q^2$, and $S \subseteq \Ff_q^2$ are as given by \cref{cl: point with good lines}. Let $\mc{L}' \subseteq \mc{L}_a$ be the set of lines such that for each $L \in \mc{L}'$, we have, 
    \[
    A(L(t), T[L](t)) = 0 \quad \text{and} \quad T[L](a) = f(a).
    \]
    Then by \cref{lm: line vs point newton}, we have that there is a degree $d$ polynomial $Q: \Ff_q^2 \to \Ff_{q'}$ such that $Q|_L = T[L]$ for every $L \in \mc{L}'$. Since $T[L](b) = f(b)$ for $\eta$-fraction of the points $b \in L$, we get that $Q(b) = f(b)$ for 
    \[
    \Omega(\eps^3) \cdot \Omega(\eps^3) = \Omega(\eps^6),
    \]
    fraction of points $b$.
\end{proof}
\subsubsection{Proof of \cref{thm: line vs point interpolation}} \label{sec: proof of lvp interpolation}
Throughout this subsection, $f: \Ff_q^2 \to \Ff_{q'}$ and lines table $T$ pass the line versus point test with probability $\eps$. 
Our argument for \cref{thm: line vs point interpolation} is essentially the same as the argument in~\cite[Theorem 4.3]{hkss}. We start with the following claim.

\begin{claim} \label{cl: lvp claim 1}
    There are two different directions $u,v \in \Ff_q^2$ and a set $H \subseteq \Ff_q^2$ such that the following hold:
    \begin{itemize}
        \item $\mu(H) \geq \Omega(\eps^2).$
        \item For every $x \in H$, we have $\eps_x \geq \frac{\eps}{2}$, and 
        \[
        T[L(x;u)](x) =  T[L(x;v)](x) = f(x).
        \]
        In the equation above, $L(x;u)$ is the line going in direction $u$ passing through $f$ and $L(x;v)$ is defined similarly. 
    \end{itemize}
\end{claim}
\begin{proof}
    The proof is the same as in~\cite[Claim A.1]{hkss}.
\end{proof}

Henceforth, we fix $u, v$ and $H$ as in \cref{cl: lvp claim 1} and suppose without loss of generality that they are the $x$ and $y$ axes. Henceforth, for $\alpha \in \Ff_{q}$ we let $L^{(1)}_{\alpha}$ be the line in direction $u$ (or parallel to the $x$-axis) through $\alpha$, that is, and let $L^{(2)}_{\alpha}$ be the line in direction $y$ (or parallel to the $y$-axis) through $\alpha$. In the next lemma, set a constant $\gamma = \Omega(1)$ so that $2\gamma \cdot \eps^2 = \mu(H)$.

\begin{lemma} \label{lm: interp proof aux}
    For any $r$ such that $\frac{2\log q}{\gamma^2} \leq r \leq \gamma \cdot q$, there are subsets $S_1, S_2 \subseteq \Ff_q$ such that:
    \begin{itemize}
        \item $|S_1| = r$,
        \item $|S_2| \geq \gamma \cdot q$,
        \item $|(\Ff_q \times S_2) \cap H| \geq |H|/2 = \gamma \cdot q^2$,
        \item For every $b \in S_2$, we have $|\{a \in \Ff_q \; | \; (a,b) \in H \}| \geq \gamma \cdot q$,
        \item For every $b \in S_2$ we have $|\{a \in \Ff_q \; | \; (a,b) \in S_1 \}| \geq \gamma/2 \cdot |S_1|$.
    \end{itemize}
\end{lemma}
\begin{proof}
    The proof is the same as in~\cite[Lemma A.2]{hkss}.
\end{proof}

For the next lemma, denote $N_{d, D}$ to be the number of exponent vectors, $(i,j,k)$ that have $(1,1,d)$-weighted degree at most, $D$, and have $k \neq 0 \pmod 2$.

\begin{lemma} \label{lm: A lvp interpolate}
    Suppose $D$ satisfies $N_{d,D} > r(D+1)$. Then there is a nonzero $A \in \Ff_{q'}[x,y,z]$ with $(1,1,d)$-degree at most $D$ satisfying:
    \begin{itemize}
        \item $A(\alpha, y, T[L^{(2)}_{a}](y)) = 0 \quad \forall \alpha \in S_1$,
        as a univariate polynomial in $y$.
        \item $A$ is supported on monomials whose $z$-degree is either $0$ or odd.
    \end{itemize}
\end{lemma}
\begin{proof}
    The proof is the same as in~\cite[Lemma A.7]{hkss}.
\end{proof}

Now fix $A$ to be the polynomial satisfying \cref{lm: A lvp interpolate}. The next 
lemma is analogous to~\cite[Lemma A.8]{hkss}.
The proof is the same, but as the statement therein is written for $\mathbb{F}_q$-valued functions, we repeat it here for $\Ff_{q'}$-valued functions. The only difference is that we apply the Schwartz-Zippel lemma over the sub-domain $\Ff_q^2 \subseteq \Ff_{q'}^2$.

\begin{lemma}
    The polynomial $A$ also satisfies $A(\alpha,\beta, f(\alpha,\beta)) = 0$ for all $(\alpha,\beta) \in S := \{(\alpha,\beta)\in H \; | \; \beta \in S_2  \}$.
\end{lemma}
\begin{proof}
    Fix a $\beta \in S_2$, and consider the polynomial $P(x) = A(x, \beta, T[L_{\beta}^{(2)}](x))$. For all $\alpha \in S_1$ such that $(\alpha, \beta) \in H$ we 
    \[
    T[L_{\beta}^{(2)}](\alpha, \beta) = T[L_{\alpha}^{(1)}](\alpha, \beta) = f(\alpha, \beta),
    \qquad\qquad
    P(\alpha) = A(\alpha, \beta, T[L_{\beta}^{(2)}](x)) = 0.
    \]
    Thus, $P$ is equal to $0$ on the entire set $\{\alpha\; | \;(\alpha, \beta) \in H,  \alpha \in S_1 \}$, which has size at least $\gamma |S_1|/2 > D$. On the other hand $P$ has degree at most $D$, so by the Schwartz-Zippel lemma it follows that $P$ is identically $0$. Thus, 
    \[
    P(\alpha) = A(\alpha, \beta,  T[L_{\beta}^{(2)}](x)(\alpha)) = A(\alpha, \beta, f(\alpha, \beta)) = 0,
    \]
    for all $\alpha$ such that $(\alpha, \beta) \in H$. As this holds for all $\beta \in S_2$, the lemma is established.
\end{proof}
We are now ready to prove \cref{thm: line vs point interpolation} following the proof of \cite[Theorem 4.3]{hkss}.
\begin{proof}[Proof of \cref{thm: line vs point interpolation}]
 Take $H \subseteq \Ff_q^2$ to be the set of points from \cref{cl: lvp claim 1} and let $\gamma = \Omega(\eps^2)$ be such that $2\gamma\cdot q^2 = |H|$. Applying an affine transformation if necessary, we assume that the directions from \cref{cl: lvp claim 1} are the standard axes. Following \cite{hkss}, we set $r = 900d/\gamma^2$ and $D = \gamma r/ 3$ so that the following hold:
 \begin{itemize}
     \item $D > 20d$,
     \item $|N_{d,D}| \geq D^3/12d = \frac{(300^3 d^2)}{12\gamma^3}$,
     \item $r(D+1) \leq 2rD = \frac{6 \cdot (300)^2 d^2}{\gamma^3}$,
     \item $\gamma r / 2 = 450d/\gamma$
 \end{itemize}
 We can then apply \cref{lm: interp proof aux}. Let $S_1, S_2$ be the sets therein and let 
 \[
 S = \{(a,b) \in \Ff_{q}^2 \; | \; b \in S_2, (a,b) \in H \}.
 \]
 Since $\frac{D^3}{12d} > r(D+1)$ and $D < \gamma r / 2$, we have that there is a non-zero polynomial $\Tilde{A} \in \Ff_{q'}[x,y,z]$ such that $\deg_{1,1,d}(\Tilde{A}) \leq D$, all monomials appearing in $\tilde{A}$ are from the set $N_{d,D}$, and $\tilde{A}(a,b,f(a,b)) = 0$ for all $(a,b) \in S$. 
 
 To conclude, we apply~\cite[Claim A.9]{hkss}, which shows that $\tilde{A}$ must depend on the variable $z$ and take $A \in \Ff_{q'}[x,y,z]$ to be the polynomial with minimum $(1,1,d)$-degree which vanishes on the set $\{(a,b,f(a,b) \; | \; (a,b) \in S\}$ and has $\partial_z(A) \neq 0$. Note that by \cref{lm: nonzero disc aux}, we get that $\Disc_z(A)$ is nonzero, and $A$ satisfies the desired properties.
\end{proof}

\subsection{Getting Agreement with the Lines Table}  \label{sec: line vs point randomized agree proof}

In this section we show how to conclude a degree $d$ function that agrees non-trivially with the lines table. The result is given in \cref{thm: line vs point randomized agree} and is a key step in the proof of \cref{thm: line versus point final redo}. One important feature is that the agreement with the lines table is boosted in \cref{thm: line vs point randomized agree}. Indeed, one goes from $\Omega(\eps^6)$ agreement with the points table in \cref{thm: line versus point final redo} to $\Omega(\eps)$ agreement with the lines table in \cref{thm: line vs point randomized agree}.

As it will be helpful for us later on, we also allow the lines table in \cref{thm: line vs point randomized agree} to be randomized.

\begin{thm} \label{thm: line vs point randomized agree}
    Let $f: \Ff_q^2 \to \Ff_{q'}$ an oracle function and let $T$ be a randomized lines table of degree $d$ such that
    \[
    \Pr_{L, T[L], x \in L}[T[L](x) = f(x)] = \eps \geq \left(\frac{d}{q}\right)^{\tl}.
    \]
    Then there exists a degree $d$ function $F: \Ff_q^2 \to \Ff_{q'}$ satisfying,
    \[
    \Pr_{L, T[L]}[T[L] = F|_L] \geq \frac{\eps}{5}.
    \]
\end{thm}

By a randomized lines table of degree $d$, $T$, we mean that $T$ has entries indexed by affine lines $L \subseteq \Ff_q^2$, and for each affine line $L$, $T[L]$ is a randomly distributed degree $d$ function from $L \to \Ff_{q'}$. 

Towards the proof of \cref{thm: line vs point randomized agree}, we will first show a list decoding version of \cref{thm: line versus point redo}. The proof proceeds using standard techniques. Let $\eps' = \Omega(\eps^6)$. Recall that we assume $\eps' \geq 2\sqrt{d/q}$.

\begin{thm} \label{thm: line vs point list decode}
      Suppose the setup of \cref{thm: line vs point randomized agree}. Then there is a list of $M \leq \frac{5}{\eps'}$ degree $d$ functions $\{F_1,\ldots, F_M\}$ from $\Ff_q^2 \to \Ff_{q'}$, such that
     \[
     \Pr_{x \in \Ff_q^2, L \ni x}\left[T[L](x) = f(x) \land f(x) \notin \{F_i(x)\}_{i=1}^M\right] \leq \frac{\eps}{2},
     \]
     and for each $i \in [M]$,
     $\Pr_{x \in \Ff_q^2}[f(x) = F_i(x)] \geq \frac{\eps'}{4}$.
\end{thm}
\begin{proof}
      Let $\{F_1,\ldots, F_M\}$ be the list of degree $d$ polynomials that agree with $f$ on at least $\eps'/4$-fraction of the points in $\Ff_q^2$. By \cref{thm: list decoding}, 
      \[
      M \leq \frac{\eps'/4}{{\eps'}^2/16 - d/q} \leq \frac{5}{\eps'}.
      \]
      We will show that $\{F_1, \ldots, F_M \}$ is the desired list of functions. Suppose for the sake of contradiction that this is not the case. That is, suppose
      \begin{equation}   \label{eq: line point assump}
      \Pr_{\substack{x \in \Ff_q^2, L \ni x\\T[L]}}[T[L](x) = f(x) \land f(x) \notin \{F_i(x)\}_{i=1}^M] > \frac{\eps}{2}.
    \end{equation}
    For each $i \in [M]$, let $W_i = \{x \in \Ff_q^2 \; | \; f(x) = F_i(x)\}$. Now consider the 
    function $f'$ obtained by setting $f'(x) = x_1^{d+1}$ for every $x \in \cup_{i=1}^M W_i$ and keeping $f'(x) = f(x)$ for all other points $x$. Note that $f'$ has agreement at most $0.26\eps'$ with every degree $d$ function over $\Ff_q^2$. Indeed, fix an arbitrary $g: \Ff_q^2 \to \Ff_{q'}$ of degree $d$, then
      \[
      \Pr_{x \in \Ff_q^2}[g(x) = f'(x)] \leq  \Pr_{x \in \Ff_q^2}[g(x) = x_1^{d+1}] + 0.25\eps' \leq 0.25\eps' + \frac{d+1}{q} \leq 0.26\eps'.
      \]
      where we use the Schwartz-Zippel lemma in the second inequality.

      On the other hand, by~\eqref{eq: line point assump}, $T$ and $f'$ will pass with probability at least $\eps/2$ due to the points $x \notin \cup_{i=1}^M W_i$ alone, so
    \[
      \Pr_{\substack{x \in \Ff_q^2, L \ni x\\ T[L]}}[T[L](x) = f'(x)] > \frac{\eps}{2}.
      \]
    Thus, we can fix for each $L$ a choice of $T[L]$ so that the line table becomes deterministic and the above inequality still holds. \cref{thm: line versus point redo} then implies that there is a degree $d$ polynomial with agreement at least $\Omega\left(\eps^6\right)\geq \eps'$
    with $f'$, which is a contradiction.
\end{proof}
We are now ready to prove \cref{thm: line vs point randomized agree}
\begin{proof}[Proof of \cref{thm: line vs point randomized agree}]
Let $F_1, \ldots, F_M$ be the degree $d$ functions given by~\cref{thm: line vs point list decode}. For each $i$, let $W_i = \{x \in \Ff_q^2 \; | \; f(x) = F_i(x) \}$ and let $W = \bigcup_{i=1}^k W_i$. 

For an affine line $L$ and randomized lines table entry $T[L]$, we let $E_i$ denote the event that $T[L] = F_i|_L$, and let Let $E = E_1 \lor \ldots \lor E_M$. Then we can rewrite the pass probability from \cref{thm: line vs point randomized agree} as:
\begin{align*}
  \Pr_{L, x \in L, T[L]}[T[L](x) = f(x)]  &= \Pr_{L, x \in L, T[L]}[T[L](x) = f(x) \land E \land x \in W]\\
  &+ \Pr_{L, x \in L, T[L]}[T[L](x) = f(x) \land \overline{E} \land x \in W] \\
    &+  \Pr_{L, x \in L, T[L]}[T[L](x) = f(x) \land x \notin W]\\
    &= \eps.
\end{align*}

We will show that the first term must be non-trivial and then further dissect it. To this end, we upper bound the second term and third terms from the right hand side above. For the third term, we apply \cref{thm: line vs point list decode} to obtain
\[
\Pr_{L, x \in L, T[L]}[T[L](x) = f(x) \land x \notin W] \leq \frac{\eps}{2}.
\]
For the second term we have
\begin{align*}
 \Pr_{L, x \in L, T[L]}[T[L](x) = f(x) \land \overline{E} \land x \in W] 
 &= \Pr_{L, x \in L, T[L]}[T[L](x) \in \{F_1(x), \ldots, F_M(x) \} \land \overline{E}] \\
 &\leq   \sum_{i=1}^M \Pr_{L, x \in L, T[L]}[T[L](x) = F_i(x) \; | \; T[L] \neq F_i|_L] \\
 &\leq M \cdot \frac{d}{q} 
\end{align*}
In the second transition, we perform a union bound and in the third transition we use the Schwartz-Zippel lemma along with the fact that $T[L]$ and $F_i$ are distinct degree $d$ functions on the affine line $L$.

It follows that 
\[
\Pr_{L, x \in L, T[L]}\left[T[L](x) = f(x) \land E \land x \in W \right] \geq \eps - \frac{\eps}{2} - M\cdot \frac{d}{q} \geq 0.49\eps.
\]
Performing a union bound on the above probability, we get
\begin{align*}
    \sum_{i, j = 1}^k \Pr_{L, x \in L, T[L]}\left[T[L](x) = f(x) \land T[L] = F_i|_L \land f(x) = F_j(x) \right]
    \geq 0.49\eps.
\end{align*}
Fix $i \neq j$ and consider a term in the summation above. We can bound the term by 
\[
\Pr_{L, x \in L, T[L]}\left[T[L](x) = f(x) \land T[L] = F_i|_L \land f(x) = F_j(x) \right] \leq \Pr_{x} \left[F_i(x) = F_j(x) \right] \leq \frac{d}{q}.
\]
It follows that 
\[
\sum_{i=1}^M \Pr_{L, x \in L, T[L]} \left[T[L](x) = f(x) \land T[L] = F_i|_L \land f(x) = F_j(x) \right] \geq 0.49\eps - \frac{M^2d}{q} \geq 0.48\eps
\]

For each $i$, let $\mc{F}_i : G(m, 1) \xrightarrow{} [0,1]$ be the function given by $\mc{F}_i[L] = \Pr_{T[L]}[T[L] = F_i|_L]$. Suppose $x$ is sampled first, and then $L \ni x$, and finally $T[L]$. We can upper bound  the probabilities in the summation as
\[
\sum_{i = 1}^M \Pr_{x}[x \in W_i] \cdot \Pr_{x \in W_i, L \ni x, T[L]}[T[L] = F_i|_L] = \sum_{i = 1}^M \Pr_{x}[x \in W_i] \E_{x \in W_i, L \ni x}[\mc{F}_i(L)].
\]
Finally, because each $W_i$ consists of at least $\eps/4$-fraction of points, applying~\cref{lm: spectral sampling} on the line versus point inclusion graph (with the singular value estimate from~\cref{lem:spectral_val_bd}), we have
\[
\left|\E_{L}[\mc{F}_i(L)] -  \E_{x \in W_i, L \ni x}[\mc{F}_i(L)]\right| \leq \frac{q^{-1/2}}{\sqrt{\eps/4}}
\]
Putting everything together, we get that
\[
\sum_{i = 1}^M \Pr_{x}\left[x \in W_i \right] \E_{L}[\mc{F}_i(L)] \geq 0.48\eps - \frac{q^{-1/2}}{\sqrt{\eps/4}} \geq 0.47\eps.
\]
It follows that there exists some $i$ such that 
\[
\E_{L, T[L]}[\mc{F}_i(L)] \geq 0.47\eps \cdot \left(\sum_{i = 1}^M \Pr_{x}\left[x \in W_i \right] \right)^{-1},
\]
and we fix this $i$ henceforth.

To conclude we note that the extra factor on the right hand side cannot be too large because the $W_i$'s have pairwise small intersections. Indeed,
\[
1 \geq \Pr_{x}\left[x \in W \right] \geq \sum_{i = 1}^M \Pr_{x}\left[x \in W_i \right] - \sum_{i \neq j} \Pr_{x} \left[x \in W_i \cap W_j \right].
\]
For fixed, $i \neq j$, $\Pr_{x} \left[x \in W_i \cap W_j \right] \leq \frac{d}{q}$. It follows that,
\[
 \sum_{i = 1}^M \Pr_{x}\left[x \in W_i \right] \leq 1 + \frac{M^2d}{q} \leq 2.
\]
As a result, for the degree $d$ function $F_i$, we have
\[
\E_{L, T[L]}[F_i(L)]  = \Pr_{L, T[L]}[T[L] = F_i|_L] \geq \frac{\eps}{5},
\]
as desired.
\end{proof}

\subsection{Boosting agreement to $\Omega(\eps)$ and the Proof of \cref{thm: line versus point final redo}}\label{sec:list_decode}

We are now ready to finish the proof of \cref{thm: line versus point final redo}. At this point, most of the work has been done already and we conclude the proof with the following two steps:
\begin{itemize}
    \item apply \cref{thm: line vs point randomized agree} to get agreement $\Omega(\eps)$,
    \item apply~\cref{lm: spectral sampling} on the affine lines versus points graph to go from $\Omega(\eps)$-agreement with the lines table to $\Omega(\eps)$-agreement with the points table.
\end{itemize}

We now move on to the formal argument, and we begin with some set-up. Suppose we are in the setting of \cref{thm: line versus point final redo}. For each line $L$, let 
\[
\eps_L = \Pr_{x \in L}[T[L](x) = f(x)], \qquad\qquad
\mc{L}' = \{L \subseteq \Ff_q^2 \; | \; \eps_L \geq \eps - (d/q)^{\tl}\}.
\]
Then the relative size of $\mathcal{L}'$ is at least $(d/q)^{\tl}$, and
\[
\Pr_{L \in \mc{L}, x\in L}[T[L](x) = f(x) \land L \notin \mc{L}'] \leq \eps - \left(\frac{d}{q}\right)^{\tl},
\]
so
\begin{equation} \label{eq: line vs point redo aux}
\Pr_{L \in \mc{L}, x\in L}[T[L](x) = f(x) \land L \in \mc{L}'] \geq \eps - \Pr_{L \in \mc{L}, x\in L}[T[L](x) = f(x) \land L \notin \mc{L}'] \geq \left( \frac{d}{q}\right)^{\tl}.
\end{equation}
Now consider the lines table $T'$ which is constructed as follows. For $L \in \mc{L}'$, set $T'[L] = T[L]$. For $L \notin \mc{L}'$, set $T'[L]$ to be a uniformly random degree $d$ function from $L \to \Ff_{q'}$. We claim that there is a setting of $T'$ such that for any degree $d$ function, $g: \Ff_q^2 \to \Ff_{q'}$, we have
\[
\Pr_{L \notin \mc{L}'}[g|_L = T[L]] \leq \frac{1}{q}.
\]
To this end note that, for any subset $S \subseteq \mc{L} \setminus \mc{L'}$ of size $|S| = |\mc{L}|/q$, and degree $d$ function $g: \Ff_q^2 \to \Ff_{q'}$, we have
\[
\Pr_{T'}[g|_L = T'[L], \; \forall L \in S] = q'^{-(d+1) |S|}.
\]
Now we union bound over $q'^{d^2} \leq q'^{q^2}$ degree $d$ functions $g: \Ff_q^2 \to \Ff_{q'}$ and the number of possible sets $S \subseteq \mc{L} \setminus \mc{L'}$ of size $|\mc{L}|/q$. We use the following bound on the number of such $S$,
\[
\binom{|\mc{L}|}{|S|} \leq q^{2|S|}.
\]
Altogether, we get that the probability there exists a degree $d$ function $g$ agreeing with $T'$ on $|\mc{L}|/q$ of the lines $L \in \mc{L'}$ is at most
\[
q'^{-(d+1) |S|} \cdot q^{2|S|} \cdot q'^{q^2} < 1.
\]
Therefore, there exists a table $T'$ such that for any degree $d$ function $g: \Ff_q^2 \to \Ff_{q'}$, we have
\begin{equation} \label{eq: randomized line table condition}
\Pr_{L \notin \mc{L'}}[g|_L = T'[L]] \leq \frac{1}{q}.
\end{equation}
We fix this $T'$ and show how to conclude \cref{thm: line versus point final redo}.
\begin{proof}[Proof of \cref{thm: line versus point final redo}]
Fix $T'$ as described above. Then,
\[
\Pr_{L \in \mc{L}, x\in L}[T'[L](x) = f(x)] \geq \Pr_{L \in \mc{L}, x\in L}[T'[L](x) = f(x) \land L \in \mc{L}] \geq \left(\frac{d}{q}\right)^{\tl} ,
\]
where we are using~\eqref{eq: line vs point redo aux} for the second inequality. Applying \cref{thm: line vs point randomized agree}, we have that there exists a degree $d$ function $F: \Ff_q^2 \to \Ff_{q'}$ such that
\[
\Pr_{L \in \mc{L}}[F|_L = T'[L]] \geq 0.2\left(\frac{d}{q}\right)^{\tl}.
\]
By~\eqref{eq: randomized line table condition}, all but $1/q$-fraction of the lines on which $F$ and $T'$ agree must actually be in $\mc{L}'$. Thus,
\begin{equation} \label{eq: line vs point redo aux agreement}
\Pr_{L \in \mc{L'}}[F|_L = T'[L]] \geq  0.2\left(\frac{d}{q}\right)^{\tl} - \frac{1}{q} := \eta.
\end{equation}
Let $\mathcal{L}''\subseteq \mathcal{L}'$ be the set of $L$'s such that $F|_{L} = T'[L]$. Since $\eps_L \geq \eps - (d/q)^{\tl}$ for all $L \in \mc{L'}$ we 
get that $\Pr_{L\in\mathcal{L}'', x\in L}[F(x) = f(x)]\geq \eps - (d/q)^{\tl}$. 
Also, note that the relative size of $\mathcal{L}''$ is at least $\eta (d/q)^{\tl}$.
It follows from~\cref{lm: spectral sampling} (with the singular value estimate from~\cref{lem:spectral_val_bd}) that
\begin{align*} 
\Pr_{x \in \Ff_q^2}[f(x) = F(x)] \geq \eps - \left(\frac{d}{q}\right)^{\tl} - \frac{q^{-1/2}}{\sqrt{\eta (d/q)^{\tl}}} \geq \eps -  1.01\left(\frac{d}{q}\right)^{\tl}.
\end{align*}
For the last inequality, we are using~\eqref{eq: line vs point redo aux agreement} and the fact that $\tl \leq 1/4$.
\end{proof}

\subsection{The $m=3$ Case: Proof of~\cref{thm: line versus point m = 3}}\label{sec:boost}
Here we show how to derive \cref{thm: line versus point m = 3} from \cref{thm: line versus point final redo} via a boosting argument, using an argument from~\cite{RazSafra}. Suppose we are in the setting of \cref{thm: line versus point m = 3} with function $f: \Ff_q^3 \to \Ff_{q'}$ and lines table $T$ which satisfy
\[
\Pr_{L \subseteq \Ff_q^3, x \in L}[T[L](x) = f(x)] \geq \eps,
\]
for $\eps$ satisfying $\eps^6 \geq \sqrt{d/q}$. Let $\mc{P}$ denote the set of affine planes $P \subseteq \Ff_q^3$. For each plane $P \in \mc{P}$, let $\eps_P = \Pr_{L \subseteq P, x \in L}[T[L](x) = f(x)]$. Then $\E_{P \subseteq \Ff_q^3}[\eps_P] = \eps$,
and by an averaging argument, for at least $\eps/2$ fraction of planes $P \subseteq \Ff_q^3$, we have $\eps_P \geq \eps/2$. Let $\mc{P}_{\good}$ denote the set of such planes, and let us now construct a planes table $T'$ as follows.
For each plane $P$, let $T'[P] \in \Ff_{q'}^{\leq d}[x,y]$ be the degree $d$ polynomial on $P$ that has the highest agreement with $f|_P$. 

\begin{lemma}
    For $P \in \mc{P}_{\good}$, we have 
    \[
    \Pr_{x \in P}[T'[P](x) = f(x)] \geq \eps' := \frac{\eps}{2} - 1.01\left(\frac{d}{q} \right)^{\tl}.
    \]
\end{lemma}
\begin{proof}
Since $T'[P]$ is the degree $d$ polynomial over $P$ which has the highest agreement with $f|_P$, it suffices to show that there is such a polynomial with agreement at least $\frac{\eps}{2} - 1.01\left(\frac{d}{q} \right)^{\tl}$ with $f|_P$. This follows immediately from \cref{thm: line versus point final redo}. Indeed, for $P \in \mc{P}_{\good}$, we have
    \[
    \Pr_{L \subseteq P, x \in L}[T[L](x) = f(x)] \geq \frac{\eps}{2}.
    \]
Therefore $f|_P$ and the lines table over $P$ obtained by taking $T[L]$ for $L \subseteq P$ pass the line versus point test with probability $\frac{\eps}{2}$. Applying \cref{thm: line versus point final redo}, there is a degree $d$ polynomial with agreement at least $\frac{\eps}{2} - 1.01\left(\frac{d}{q} \right)^{\tl}$ with $f|_P$.
\end{proof}

Now consider the graph $G(\mc{P}, E)$, whose vertices are $\mc{P}$ and whose edges are $(P, P')$ such that
\begin{itemize}
    \item $T[P]|_{P \cap P'} = T[P']|_{P \cap P'}$,
    \item $P, P' \in \mc{P}_{\good}$.
\end{itemize}
We first show that $G$ has many edges.

\begin{lemma} \label{lm: many edges}
The graph $G$ has $\Omega(\eps^2)\cdot |\mc{P}|^2$ edges.
\end{lemma}
\begin{proof}
First one can calculate that 
\[
\Pr_{P, P' \in \mc{P}}[P \cap P' = \emptyset] \leq \frac{1}{q^2}.
\]
Indeed, we have $|\mc{P}| = \frac{q^3(q^3-1)(q^3-q)}{q^2(q^2-1)(q^2-q)} = q(q^2+q+1)$. On the other hand, $P$ and $P'$ do not intersect if and only if they are identical or parallel, so for each $P$, there are $q$ possible choices $P'$ such that $P \cap P' = \emptyset$. Therefore,
\[
\Pr_{P, P' \in \mc{P}}[P \cap P' = \emptyset] = \frac{q}{|\mc{P}|} \leq \frac{1}{q^2}.
\]
Now suppose we choose $P, P' \in \mc{P}$ with non-empty intersection and $x  \in P \cap P'$. Let $I(P, x)$ be the event that $P \in \mc{P}_{\good}$ and $T[P](x) = f(x)$. Then
\begin{align*}
    \E_{P, P', x \in P \cap P'}[I(P,x) \land I(P',x)] &= \E_{x}[\Pr_{P \ni x}[I(P,x)]^2 ] \\
    &\geq \Pr_{x, P \ni x}[I(P,x)]^2 \\
    &\geq \Omega(\eps^2).
\end{align*}
In the last transition we use the fact that $P \in \mc{P}_{\good}$ with probability $\Omega(\eps)$, and conditioned on this, $T[P](x) = f(x)$ with probability $\Omega(\eps)$. It follows that for $\Omega(\eps^2)$ pairs, of $P, P'$, we have,
\[
\E_{x \in P \cap P'}[I(P,x) \land I(P', x)] \geq \Omega(\eps^2).
\]
For such $P, P'$, it must be that $P, P' \in \mc{P}_{\good}$, and
\[
\Pr_{x \in P \cap P'}[T[P](x) = T[P'](x)] \geq \Omega(\eps^2) > \frac{d}{q}.
\]
By the Schwartz-Zippel lemma, $T[P]|_{P \cap P'} = T[P']|_{P \cap P'}$ and $(P,P')$ is an edge in $G$.
\end{proof}

\begin{lemma} \label{lm: beta transitive}
    The graph $G$ is $\frac{d+1}{q}$-transitive, in the sense that for every non-edge $(P,P')$ we have that
    $\Pr_{P''}[(P,P''), (P',P'')\in E]\leq \frac{d+1}{q}$.
\end{lemma}
\begin{proof}
    Fix $P, P'$ which are non-adjacent. If either $P \notin \mc{P}_{\good}$ or  $P' \notin \mc{P}_{\good}$, then
    \[
    \Pr_{P''}[(P, P''), (P', P'') \in E] = 0.
    \]
    Otherwise, it must be the case that $T[P]|_{P \cap P'} \neq T[P]|_{P \cap P'} $. In this case, suppose we choose $P''$ randomly. Then with probability at most $1/q$ we have that $P''$ is disjoint from the line $P \cap P'$. If $P''$ is not adjacent from this line, then the only way for $P''$ to be adjacent to both $P$ and $P'$ is to have
    \begin{equation} \label{eq: transitive 1}        
   T[P]|_{P \cap P''} =  T[P'']|_{P \cap P''} \quad \text{and} \quad  T[P']|_{P \cap P''} =  T[P'']|_{P \cap P''}.
      \end{equation}
This implies that,
    \begin{equation} \label{eq: transitive 2}
      T[P]|_{P \cap P' \cap P''}    =  T[P']|_{P \cap P' \cap P''}.
    \end{equation}
    If $P'' \supseteq P \cap P'$, then we cannot have $P''$ adjacent to both $P$ and $P'$ as we cannot have~\eqref{eq: transitive 2} hold. If $P \cap P' \cap P''$ is a point, then this point is uniformly random point in $P \cap P'$, and by the Schwartz-Zippel lemma~\eqref{eq: transitive 2} holds with probability at most $d/q$. Overall this shows that if $P, P'$ are not adjacent then
    \[
    \Pr_{P''}[(P, P'') \in E \land (P', P'') \in E] \leq \frac{d+1}{q}.\qedhere
    \]
\end{proof}

Now we can use the following fact, asserting that almost transitive graphs can be made transitive by deleting a few edges.

\begin{lemma} \label{lm: delete transitive}
    Let $H = (V, E)$ be a graph which is $\beta$-transitive. Then, one can delete delete at most $2\sqrt{\beta}|V|^2$ edges from $H$ to obtain a graph $H'$ which is transitive.
\end{lemma}
\begin{proof}
This is~\cite[Lemma 2]{RazSafra}.
\end{proof}

Combining \cref{lm: many edges}, \cref{lm: beta transitive}, and \cref{lm: delete transitive}, we get that one can delete at most $2\sqrt{\frac{d+1}{q}}|\mc{P}|^2$ edges from $G$ to obtain a graph $G'$ which is transitive. Thus $G'$ is transitive and contains at least $\left(\Omega(\eps^2) - \sqrt{\frac{d+1}{q}}\right) |\mc{P}|^2$ edges. Let $\mc{C}_1, \ldots, \mc{C}_k$ be the cliques whose union is $G'$, and let $\mc{C}_1$ be the largest clique. It follows that
\[
\Omega(\eps^2)|\mc{P}|^2 = \sum_{i=1}^k |\mc{C}_i|^2 \leq |\mc{C}_1|  \sum_{i=1}^k |\mc{C}_i| \leq |\mc{C}_1|  \cdot |\mc{P}|.
\]
Therefore, $|\mc{C}_1| \geq \Omega(\eps^2)|\mc{P}|$. The clique $\mc{C}_1$ appears in $G$ too, so $\mc{G}$ has a clique of fractional size $\Omega(\eps^2)$. Going back to the table $T'$, the clique $\mc{C}_1$ corresponds to a large number of planes $P$, such that the functions $T[P]$ for $P \in \mc{C}_1$ all mutually agree, and each $P \in \mc{P}_{\good}$. We now show that one can interpolate a degree $d$ function $Q$ such that $Q|_P = T[P]$ for all $P \in \mc{C}_1$. To do so, we first rely on the same result in the $q = q'$ case, which appears in \cite{RazSafra}.
\begin{lemma}[Lemma 4 \cite{RazSafra}] \label{lm: rs interpolate} 
    Let $\mc{C}$ be a set of affine planes $P \subseteq \Ff_{q}^3$ of fractional size $\Omega(\eps^2)$ in $\mc{P}$. Suppose for each plane $P$, there is a degree $d$ function $T[P]: P \to \Ff_q$ such that for all $P_1, P_2 \in \mc{C}$, we have $T[P_1]|_{P_1 \cap P_2} = T[P_2]|_{P_1 \cap P_2}$. Then there is a degree $d$ function $Q: \Ff_q^3 \to \Ff_{q}$ such that $Q|_P = T[P]$ for all $P \in \mc{C}$.
\end{lemma}
\begin{proof}
    This is \cite[Lemma 4]{RazSafra}.
\end{proof}

We now handle the general $q'$ case:
\begin{lemma} \label{lm: interpolate m=3}
    There is a degree $d$ function $Q: \Ff_q^3 \to \Ff_{q'}$ such that $Q|_P = T[P]$ for all $P \in \mc{C}_1$.
\end{lemma}
\begin{proof}
Note that one can view $\Ff_{q'}$ as a vector space over $\Ff_q$. As $\Ff_q$ is a subfield of $\Ff_{q'}$, there is some $k$ such that $q' = q^k$, and there exist $\Lambda_1, \ldots, \Lambda_k \in \Ff_{q'}$ such that each $\beta \in \Ff_{q'}$ has a unique decomposition as an $\Ff_q$-linear combination of $\Lambda_1, \ldots, \Lambda_k$:
\[
\beta = \sum_{i = 1}^k \gamma_i \cdot \Lambda_i.
\]
For every $\beta \in \Ff_{q'}$, let $\proj_i(\beta) = \gamma_i$ from the expansion above.

It follows that we can also view each function $T[P]: P \to \Ff_{q'}$, can also be viewed as a vector of $k$-functions from $P \to \Ff_q$. Specifically, we can write the polynomial expansion of $T[P]$ as 
\[
    T[P](x_1,x_2) = \sum_{i + j \leq d} \beta_{i,j} \cdot x_1^i x_2^j,
\]
where the $\beta_{i,j}$'s are coefficients from $\Ff_{q'}$. Then writing each $\beta_{i,j}$ as an $\Ff_q$-linear combination of $\Lambda_1, \ldots, \Lambda_k$, we get
\[
T[P](x_1, x_2) = \sum_{i + j \leq d} \sum_{\ell = 1}^k \gamma_{\ell, i, j} \cdot \Lambda_{\ell} \cdot x_1^i x_2^j = \sum_{\ell = 1}^k  \Lambda_{\ell} \sum_{i + j \leq d}\gamma_{\ell, i, j} \cdot x_1^i x_2^j,
\]
where the $\gamma_{\ell, i,j}$'s are all in $\Ff_q$. Thus, for each $P \in \mc{C}_1$ and each $\ell \in [k]$, we write
\[
T_{\ell}[P] = \sum_{i + j \leq d}\gamma_{\ell, i, j} \cdot x_1^i x_2^j,
\qquad\qquad
T[P] = \sum_{\ell = 1}^k \Lambda_{\ell} \cdot T_{\ell}[P],
\]
where each $T_{\ell}[P]: \Ff_q^2 \to \Ff_q$. One can check that $T_{\ell}[P](x) = \proj_{\ell}(T[P](x))$, and as a result, the functions $T_{\ell}[P]$ all mutually agree for $P \in \mc{C}_1$. Applying \cref{lm: rs interpolate}, for each $\ell \in [k]$, we get that there are degree $d$ functions $Q_1, \ldots, Q_{k}: \Ff_q^2 \to \Ff_q$, such that $Q_{\ell}|_P = T_{\ell}[P]$
for all $P \in \mc{C}_1$. Taking $Q = \sum_{\ell = 1}^k \Lambda_{\ell} \cdot Q_{\ell}$, we get that $Q: \Ff_q^2 \to \Ff_{q'}$ is a degree $d$ function satisfying $Q|_P = T[P]$
for all $P \in \mc{C}_1$.
\end{proof}

From here we are ready to conclude the proof of \cref{thm: line versus point m = 3}. 

\begin{proof}[Proof of \cref{thm: line versus point m = 3}]
    Take $Q$ to be the degree $d$ polynomial from \cref{lm: interpolate m=3}. Then,
    \[
    \Pr_{P \in \mc{C}_{1}, x \in P}[Q(x) = f(x)] = \Pr_{P \in \mc{C}_{1}, x \in P}[T[P](x) = f(x)] \geq \eps'.
    \]
    By \cref{lm: spectral sampling} with the planes versus points graph (with the singular value estimate from~\cref{lem:spectral_val_bd}), we have
    \[
    \Pr_{x \in \Ff_q^3}[Q(x) = f(x)] \geq \eps' - \frac{q^{-1}}{\sqrt{\Omega(\eps^2)}} \geq 0.49\eps - 1.01\left(\frac{d}{q}\right)^{\tl},
    \]
    and we are done. In the last inequality we are using the fact that $\Omega(\eps) \leq q^{-1/6}$.
\end{proof}

\subsection{Agreement with the Lines Table when $m = 3$} \label{sec: line vs point randomized agree m = 3}

In \cref{thm: line vs point randomized agree} we showed how to go from $\Omega(\eps^3)$ agreement with the points table to $\Omega(\eps)$ agreement with the lines table and obtain \cref{thm: line vs point randomized agree}. For future arguments, we will require a similar result which gives agreement with the lines table (instead of the points table) in the $m=3$ case.

\begin{thm} \label{thm: line vs point randomized agree m = 3}
    Let $f: \Ff_q^3 \to \Ff_{q'}$ an oracle function and let $T$ be a randomized lines table of degree $d$ such that
    \[
    \Pr_{L, T[L], x \in L}[T[L](x) = f(x)] = \eps \geq \left(\frac{d}{q}\right)^{\tl}.
    \]
    Then there exists a degree $d$ function $F: \Ff_q^3 \to \Ff_{q'}$ satisfying
    \[
    \Pr_{L, T[L]}[T[L] = F|_L] \geq \frac{\eps}{5}.
    \]
\end{thm}
\begin{proof}
     The proof is exactly the same as that of \cref{thm: line vs point randomized agree} except one replaces \cref{thm: line versus point redo} (which only works when $m=2$) with \cref{thm: line versus point m = 3} (which does works when $m = 3$).
\end{proof}
\section{Proofs for \cref{sec:local_test}}
In this section, we will prove the theorems from \cref{sec:local_test} 
via standard reductions from the soundness of the line versus point test in \cref{thm: line versus point final redo}. 

\subsection{Proof of \cref{thm: line vs point strong}} \label{sec: proof of lvp strong}
Fix the set up as in~\cref{thm: line vs point strong} and let $\mc{L}$ be the set of lines $L$ such that
\[
\agr(f|_L, \RS_{q'}[d, \Ff_q]) \geq 1.01\eps.
\]
Construct the randomized lines table $T$ as follows. If $L \in \mc{L}$, then set $T[L]$ to be the degree $d$ function that has the highest agreement with $f|_L$. Otherwise, set $T[L]$ to be a uniformly chosen univariate random degree $d$ function. It follows that
\[
\Pr_{L, x\in L, T[L]}[T[L](x) = f(x)] \geq \Pr_{L}[L \in \mc{L}] \cdot \eps \geq \eps^2 \geq \left(\frac{d}{q}\right)^{\tl}.
\]
Applying~\cref{thm: line vs point randomized agree} we get that there is a degree $d$ function $H: \Ff_q^2 \to \Ff_{q'}$ satisfying
\[
\Pr_{L, T[L]}[T[L] = H|_L] \geq \frac{\eps^2}{5}.
\]
However, almost all of the contribution to this probability is from $L \in \mc{L}$. Indeed, if $L \notin \mc{L}$, then $\Pr_{T[L]}[T[L] = H|_L] \leq q^{-(d+1)}$, so
\[
\Pr_{L \in \mc{L}, T[L]}[T[L] = H|_L] \geq \frac{\eps^2}{5} - \Pr_{T[L]}[T[L] = H|_L] \geq \frac{\eps^2}{5} - q^{-(d+1)} := \eps'.
\]
To conclude, we apply~\cref{lm: spectral sampling} (with spectral value estimates from~\cref{lem:spectral_val_bd}) to get that
\[
\Pr_{x \in \Ff_q^m}[f(x) = H(x)] \geq \Pr_{L \in \mc{L}, x \in L}[f(x) = H(x)] - \frac{q^{-1/2}}{\sqrt{\mu(\mc{L})}} \geq 1.01\eps -  \frac{q^{-1/2}}{\sqrt{\eps'}} \geq \eps. 
\]
\hfill $\Box$
\subsection{Proof of \cref{thm: plane vs point strong}} \label{sec: proof of pvp strong}
The proof is by contradiction. Suppose that
\[
\Pr_{P \subseteq \Ff_q^m}[\agr(f|_P, \RM_{q'}[d, \Fqt]) \geq 1.2\eps] > \frac{100}{\eps^2 q}.
\]
We will show that $\agr(f, \RM_{q'}[d, \Fqm]) > \eps$, contradicting the assumption of \cref{thm: plane vs point strong}.

For each $P \in \mc{S}$, let $T'[P]$ denote the degree $d$ function closest to $f|_P$. Let $\mathcal{S}$ be the set of planes $P$ satisfying $\agr(f|_P, \RM_{q'}[d, \Ff_q^2]) \geq 1.2\eps$, so that $\agr(T'[P], f|_P) \geq 1.2 \eps$ for all $P \in \mc{S}$. By assumption, we have $\mu(\mc{S}) > \frac{100}{\eps^2 q}$. Consider the following distribution, $\mc{D}$, which outputs an affine line $L$ and a degree $d$ function from $L \to \Ff_{q'}$:
\begin{itemize}
    \item Choose $P \in \mc{S}$ uniformly at random.
    \item Choose $L \subseteq P$ uniformly at random.
    \item Output the pair $(L, T'[P]|_L)$.
\end{itemize}

Using the degree $d$ functions $T'[P]$, for $P \in \mc{S}$, we can now generate a randomized lines table $T$ as follows. For each $L$, the degree $d$ function $T[L]$ is chosen according to the marginal distribution of $\mc{D}$ conditioned on the affine line being $L$. That is, for each line $L$ and any degree $d$ function $h: L \to \Ff_{q'}$,
\[
\Pr[T[L] = h] = \frac{\mc{D}(L, h)}{\mc{D}(L, \cdot)}.
\]
If there is no plane in $\mc{S}$ containing $L$, then $T[L]$ is uniformly random over all degree $d$ functions on $L$.

We claim that the randomized lines table $T$ passes the line versus point test with $f$ with probability close to $1.1\eps$. Note that in this test, one chooses a line $L$ uniformly at random, then $T[L]$ according to the joint distribution of $\mc{D}$, and finally a uniformly random $x \in L$. Thus, $(L, T[L])$ is almost distributed according to $\mc{D}$, save for the fact that $L$ is not quite uniformly random in $\mc{D}$. 

\begin{claim}
   \[
   \Pr_{L, x \in L, T[L]}[T[L](x) = f(x)] \geq \eps' := 1.1 \eps.
   \] 
\end{claim}
\begin{proof}
Let $F \in L_2(G(\mc{L}))$ be the function given by,
\[
F(L) = \Pr_{x \in L, T[L]}[T[L](x) = f(x)].
\]
It is clear that all values of $F$ are in the interval $[0,1]$. The probability we are interested in is then,
\[
\Pr_{L, x \in L, T[L]}[T[L](x) = f(x)] = \E_{L}[F(L)].
\]
Then,
\begin{align*} 
\E_{P \in \mc{S}, L \subseteq P}[F(L)] = \Pr_{(L,T[L]) \sim \mc{D}, x\in L}[T[L](x) = f(x)] 
= \Pr_{P \in \mc{S}, x \in P}[T'[P](x) = f(x)] 
\geq 1.2\eps.
\end{align*}
where in the expectation $L$ is chosen according to the marginal distribution of $\mc{D}$. Applying \cref{lm: spectral sampling}, we have
\[
\left| \E_{L}[F(L)] - \E_{L \sim \mc{D}}[F(L)]  \right| \leq \frac{q^{-1/2}}{\sqrt{\mu(\mc{S})}}.
\]
It follows that
\[
\Pr_{L, x \in L, T[L]}[T[L](x) = f(x)]  = \E_{L}[F(L)] \geq 1.2\eps - \frac{q^{-1/2}}{\sqrt{100\eps^{-2}q^{-1}}} = 1.1\eps.\qedhere
\]
\end{proof}
Applying \cref{thm: line vs point randomized agree m = 3}, it follows that there exists a degree $d$ function $H: \Ff_q^m \to \Ff_{q'}$ satisfying: 
\begin{equation}  \label{eq: unif line agree}
\Pr_{L, T[L]}[H|_L = T[L]] \geq \frac{\eps'}{5},
\end{equation}
where $L$ is uniform. We would like to conclude that,
\[
\Pr_{(L, T[L]) \sim \mc{D}}[H|_L = T[L]] = \Pr_{P \in S, L \subseteq P} [H|_L = T'[P]|_L]
\]
is large. To accomplish this we will apply the spectral sampling lemma again.
\begin{claim} \label{cl: agree in S}
We have,
    \[
    \Pr_{P \in S, L \subseteq P}[H|_L = T'[P]|_L]\geq \frac{\eps}{10},
    \]
and as a consequence,
\[
\Pr_{P \in S}[T'[P] = H|_P] \geq \frac{\eps}{20}.
\]
\end{claim}
\begin{proof}
Let $F': \mc{L}\to [0,1]$ be the function over affine lines of $\Ff_q^3$ given by
\[
F'(L) = \Pr_{T[L]}[H|_L = T[L]].
\] 
From~\eqref{eq: unif line agree}, we have $\E_{L}[F'(L)] = \mu(F') \geq \frac{\eps'}{5}$.
The probability of interest can be expressed as:
\[
\Pr_{P \in S, L \subseteq P}[H|_L = T'[P]|_L] = \Pr_{(L, T[L]) \sim \mc{D}}[H|_L = T[L]] = \E_{L \sim \mc{D}}[F'(L)] = \E_{P \in \mc{S}}\E_{L \subseteq P}[F'(L)].
\]
so applying \cref{lm: spectral sampling} we have
\begin{align*}
   \Pr_{P \in S, L \subseteq P}[H|_L = T'[P]|_L] \geq \mu(F') - \frac{q^{-1/2}}{\sqrt{\mu(S)}} \geq \frac{\eps'}{5} - \frac{\eps}{10} \geq \frac{\eps}{10},
\end{align*}
completing the first part of the claim. For the second part of the claim, it follows from an averaging argument that for $\frac{\eps}{20}$-fraction of the planes $P \in S$, we have,
\[
\Pr_{L \subseteq P}[H|_L = T'[P]|_L] \geq \frac{\eps}{20}.
\]
For such planes $P$, if $H|_P \neq T'[P]$, then by the Schwartz-Zippel lemma, we would have
\[
\Pr_{L \subseteq P}[H|_L = T'[P]|_L] \leq \Pr_{a,b \in P}[H|_P(a) = T'[P](a)] \leq \left(\frac{d}{q}\right)^2 < \frac{\eps}{20},
\]
so it follows that $H|_P = T'[P]$ for $\frac{\eps}{20}$-fraction of the planes $P \in S$.
\end{proof}
\begin{proof}[Proof of \cref{thm: plane vs point strong}]
Let $S'$ be the set of planes in $S$ on which $H|_P = T'[P]$. By \cref{cl: agree in S}, we have that $\mu(S') \geq \frac{\eps}{20} \cdot \mu(S) \geq \frac{5}{\eps q}$. Meanwhile, by definition of $S$, for each $P \in S' \subseteq S$, we have,
\[
\Pr_{x \in P}[H(x) = T'[P](x) = f(x)] \geq 1.2\eps,
\]
so evidently
\[
\E_{P \in S'}\Pr_{x \in P}[H(x) = f(x)] \geq 1.2\eps.
\]
Applying \cref{lm: spectral sampling} it follows that,
\begin{align*} 
\Pr_{x \in \Ff_q^3}[H(x) = f(x)] \geq 1.2\eps -  \frac{q^{-1}}{\sqrt{\mu(S')}} \geq 1.2\eps - \frac{q^{-1}}{\sqrt{5\eps^{-1} q^{-1}}} 
> \eps.
\end{align*}
This contradicts the assumption of \cref{thm: plane vs point strong} and completes the proof.
\end{proof}

\section{Proofs for \cref{sec: prox gen intro}} \label{sec: prox gen proof}
In this section we prove the theorems in \cref{sec: prox gen intro} and \cref{sec: combine intro}, and to do so we first show an auxiliary result for the total degree Reed-Muller code.

\begin{claim} \label{cl: good lines}
Suppose the setting of \cref{thm: prox gen RM without corr} and set $\delta \geq 1/\sqrt{q}$. Then for at least $1-\frac{2}{\delta^2 q}$-fraction of the affine lines $L \subseteq \Ff_{q}^m$, there is a set of points $A_L \subseteq L$ of size $\mu(A_L) \geq \eps - \delta$, such that the univariate functions $f_1|_{L}, \ldots f_k|_{L}$ have correlated agreement with degree $d$ on $A_L$. That is, for each $i \in [k]$, there is a univariate degree $d$ function $h_i: \Ff_q \to \Ff_{q'}$ satisfying
\[
f_i|_{A_L} = h_i|_{A_L}.
\]
\end{claim}
\begin{proof}
Fix $\delta \geq 1/\sqrt{q}$. For each choice of $(\xi_1, \ldots, \xi_k) \sim \prox(k, d, q, q')$ according to the proximity generator, let $B_{\xi_1, \ldots, \xi_k}$ be the set of points on which $\sum_{i=1}^k \xi_i \cdot f_i$ and its closest degree $d$ function agree. Let $S \subseteq \Ff_{q'}^k$ be the set of tuples $(\xi_1, \ldots, \xi_k)$ such that $\mu(B_{\xi_1, \ldots, \xi_k}) \geq \eps$. By assumption,
    \[
    \Pr_{\xi_1, \ldots, \xi_k \in \Ff_{q'}}[(\xi_1,\ldots, \xi_k) \in S] \geq 2k\err(d,q,q').
    \]
    Fix a $(\xi_1, \ldots, \xi_k) \in S$ and let
    \[
    \mc{L}' = 
    \left\{L \in \mc{L} \; | \; |L \cap B_{\xi_1, \ldots, \xi_k}|  \leq  (\eps - \delta)|L| 
    \right\}
    \]
    be the set of lines that have small intersection with $\B_{\xi_1,\ldots, \xi_k}$. Applying \cref{lm: spectral sampling} with the lines versus points graph, we have
    \[
  \delta \leq \left|\Pr_{L \in \Lc', x \in L}[x \in B_{\xi_1, \ldots, \xi_k}] - \Pr_{x \in \Ff_q^m}[x \in B_{\xi_1, \ldots, \xi_k}]\right| \leq  \frac{q^{-1/2}}{\sqrt{\mu_{\mc{L}}(\Lc')}}.
    \]
    Thus, $\mu_{\Lc}(\Lc') \leq \frac{1}{q\cdot \delta^2}$, and so
    \[
    \E_{\xi_1,\ldots,\xi_k \in S}\left[\Pr_{L}\left[\agr_d\left(\sum_{i=1}^k \xi_i \cdot f_i|_L\right) \geq \eps - \delta \right]\right] \geq 1 - \frac{1}{q\cdot \delta^2}. 
    \]
    Switching the order of expectations above, we get,
    \[
\E_{L}\left[\Pr_{\xi_1,\ldots,\xi_k \in S}\left[\agr_d\left(\sum_{i=1}^k \xi_i \cdot f_i|_L\right) \geq \eps - \delta \right] \right] \geq 1 - \frac{1}{\delta^2 q}. 
    \]
    By an averaging argument, it follows that for at least $1 - 2/(\delta^2 q)$ fraction of affine lines $L$, we have,
    \[
 \Pr_{\xi_1,\ldots, \xi_k \in S}\left[\agr_d\left(\sum_{i=1}^k \xi_i \cdot f_i|_L\right) \geq \eps - \delta \right] \geq 0.5.
    \]
    Call these affine lines good. It remains to show that $f_1|_L, \ldots, f_k|_L$ have correlated agreement for every good $L$. 
    
    To this end, note that \cref{thm: prox gen RS with corr} can be applied. Indeed, $0.5$-fraction of $S$ is at least $k\cdot \err(d,q,q')$ fraction of the proximity generators sampled from $\prox(k,d,q,q')$, so it follows that for every good line $L$, the univariate functions $f_1|_L, \ldots, f_k|_L$ satisfy the assumption in \cref{thm: prox gen RS with corr}:
    \[
    \Pr_{\xi_1, \ldots, \xi_k \in \Ff_{q'}}\left[\agr_d\left(\sum_{i=1}^k \xi_i \cdot f_i|_L\right) \geq \eps\right] \geq k \cdot \err(d,q,q').
    \]
    Therefore, by \cref{thm: prox gen RS with corr} for every good line $L$, there is a set of points $A_L$ of size $\mu(A_L) \geq \eps - \delta$ on which $f_1|_L, \ldots, f_k|_L$ each agree with some degree $d$ function.
\end{proof}

\subsection{Proof of \cref{thm: prox gen RM without corr}}
\begin{proof} [Proof of \cref{thm: prox gen RM without corr}] Suppose the setting of \cref{thm: prox gen RM without corr}. We apply \cref{cl: good lines} with $\delta = \frac{100}{\sqrt{q}}$, and let $\Lg$ be the set of lines given. Thus, for each $L \in \Lg$, the functions $f_1|_L, \ldots, f_k|_L$ have $\eps-\delta$ correlated agreement with degree $d$, and we let $A_L \subseteq L$ be the site of this correlated agreement. 
For each $f_i$, define a lines table $T_i$ using the degree $d$ functions that are correlated with $f_i$ on the good lines. Specifically, for each $i \in [k]$, and each good line $L$, let $T_i[L]$ be the degree $d$ function such that
\[
T_i[L]|_{A_L} = f_i|_{A_L}.
\]
For every line $L \notin \Lg$, set $T_i[L]$ to be a randomly chosen degree $d$ function. Now note that $T_i$ and $f_i$ pass the line versus point test with probability at least $0.999 \eps$. Indeed, for each $i \in [k]$
\begin{equation} \label{eq: good sample without} 
\Pr_{L \in \Lc, x \in L}[T_i[L](x) = f_i(x)] \geq \Pr_{L \in \Lc}[L \in \Lg] \cdot  \Pr_{L \in \Lg, x \in L}[T_i[L](x) = f_i(x)] \geq \left(1 - \frac{2}{\delta^2q}\right)(\eps-\delta),
\end{equation}
and this final quantity is at least $\frac{9998}{10000}(\eps - \delta) \geq 0.999 \eps$. Applying \cref{thm: line vs point}, it follows that $f_i$ has agreement at least 
\[
 0.999\eps - \left(\frac{d}{q} \right)^{\tl} \geq 0.998\eps
\]
with a degree $d$ function. Moreover, this holds for every $i \in [k]$, so we are done.
\end{proof}

\subsection{Proof of \cref{thm: prox gen RM with corr}}
\begin{proof}[Proof of \cref{thm: prox gen RM with corr}]
Suppose the setting of \cref{thm: prox gen RM with corr} and apply \cref{cl: good lines}. Call the lines satisfying \cref{cl: good lines} good, and denote the set of good lines by $\Lg$.  For each $f_i$, we will define a lines table, $T_i$, using the degree $d$ functions that are correlated with $f_i$ on the good lines. Specifically, for each $i \in [k]$, and each good line $L$, let $T_i[L]$ be the degree $d$ function agreeing with $f_i$ on $A_L$.

For every line other line, $L$ that is not good, set $T_i[L]$ to be a randomly chosen degree $d$ function. Then $T_i$ and $f_i$ pass the line versus point test with probability at least 
\[
0.98\cdot (\eps-10/\sqrt{q}) \geq 0.97 \eps.
\]
Indeed, for each $i \in [k]$
\begin{equation} \label{eq: good sample} 
\Pr_{L \in \Lc, x \in L}[T_i[L](x) = f_i(x)] \geq \Pr_{L \in \Lc}[L \in \Lg] \cdot  \Pr_{\substack{L \in \Lg\\ x \in L}}[T_i[L](x) = f_i(x)] \geq 0.998\cdot (\eps-10/\sqrt{q}) \geq 0.97 \eps.
\end{equation}

Now applying \cref{thm: line vs point list decode}, we get that for each $i \in [k]$, there is a list of $M \leq (5/(0.97k\eps))^7$ degree $d$ functions $\{F_{i,1}, \ldots F_{i, M}\}: \Ff_q^m \to \Ff_{q'}$ such that,
\begin{equation} \label{eq: list decode in prox gen}  
\Pr_{L \in \Lc, x \in L}[T_i[L](x) = f_i(x) \land T_i[L] \notin \{F_{i,j}|_L\}_{j=1}^M] \leq \frac{0.97\eps}{10k}.
\end{equation}
For each $i \in [k]$, let $\Lc_{i, j} = \{L \in \Lg \; | \; T_i[L] = F_{i,j}|_L \}$. We have
\begin{align*}
    &\Pr_{L \in \Lc, x \in L}[L \in \cap_{i}(\cup_{j}\Lc_{i,j})]\\
    &\qquad\geq  \Pr_{L \in \Lc, x \in L}[L \in \Lg \land x \in A_L] -  \Pr_{L \in \Lc, x \in L}[L \in \Lg \land x \in A_L \land 
    (\cup_i T_i[L] \notin \{F_{i, j}|_L\}_{j=1}^M)] \\
    &\qquad\geq 0.97\eps - \frac{0.97\eps}{10k}\cdot k \\
    &\qquad\geq 0.87\eps,
\end{align*}
where in the second transition we are using~\eqref{eq: good sample},~\eqref{eq: list decode in prox gen}. Since this holds for every $i$, we can find $(j_1, \ldots, j_k) \in [M]^k$ such that:
\[
\Pr_{L \in \Lc}
\left[L \in \bigcap_{i \in [k]}\Lc_{i, j_i}\right] \geq 0.87\eps \cdot \frac{1}{M^{k}} \geq \frac{1}{100q^{1/2}},
\]
where in the last transition we are using the assumption on the size of $\eps$. Let $\Lc^{\star} = \bigcap_{i \in [k]}\Lc_{i, j_i}$. For each $i \in [k]$, we have that, $f_i$ and the degree $d$ function $F_{i, j_i}$ agree on $A^\star = \bigcup_{L \in \Lc^\star} A_L$. To conclude, we use \cref{lm: spectral sampling} with the lines versus points graph to lower bound the size of $A^\star$:
\[
\mu(A^\star) \geq \eps - \frac{10}{\sqrt{q}} - \frac{q^{-1/2}}{\sqrt{\mu_{\Lc}(\Lc^\star)}} \geq \eps - \frac{10}{\sqrt{q}} -\frac{1}{10q^{1/4}} \geq 0.999\eps.
\qedhere
\]
\end{proof}

\subsection{Proof of \cref{thm: prox gen iRM with corr}} \label{sec: proof of prox gen irm with corr}
Suppose the setting of \cref{thm: prox gen iRM with corr}. Our strategy will be to first show correlated agreement with a degree $2d$ function, and then use list decoding and the Schwartz-Zippel lemma to show that this agreement must actually be with a degree $(d,d)$ function. 

To this end, let $\mc{H}$ consist of all pairs of degree $2d$ functions $(h^{(i)}_1, h^{(i)}_2)$, such that there exists a set $A_i \subseteq \Ff_q^m$ of size at least $0.99\eps$ satisfying $h^{(i)}_j|_{A_i} = f_j|_{A_i}$ for all $j \in [2]$. By \cref{thm: list decoding}, we have $|\mc{H}| = M \leq \frac{2}{\eps}$. Label the tuples in $\mc{H}$ as $(h^{(i)}_1, h^{(i)}_2)$ for $i \in [M]$. For each $i$, let $A_i$ be the set on which $f_1,f_2$ have correlated agreement with $(h^{(i)}_1, h^{(i)}_2)$ and let $A = \bigcup_{i=1}^M A_i$. If for any $i$ we have that $h^{(i)}_1$ and $h^{(i)}_2$ are both degree $(d,d)$ functions, then we are done, so suppose this is not the case.
Choosing $\xi_1,\xi_2\in \Ff_{q'}$, define the following events:
\begin{itemize}
    \item $E_1$ is the event that there exists a degree $(d, d)$ function $H$ and $V \subseteq \Ff_q^m$ satisfying $\mu(V) \geq \eps$, $\mu(V \cap A) \leq M \cdot \frac{2d}{q}$ and $\left(\xi_1 \cdot f_1 + \xi_2 \cdot f_2 \right)|_V = H|_V$.
    \item $E_2$ is the event that there exists a degree $(d, d)$ function $H$ and $V \subseteq \Ff_q^m$ satisfying $\mu(V) \geq \eps$, $\mu(V \cap A) > M \cdot \frac{2d}{q}$ and $\left(\xi_1 \cdot f_1 + \xi_2 \cdot f_2 \right)|_V = H|_V$.
\end{itemize} 
Note that 
\begin{equation} \label{eq: irm prox aux}
\Pr[E_1] + \Pr[E_2] \geq \Pr[\agr_{(d,d)}(\xi_1 \cdot f_1 + \xi_2 \cdot f_2) \geq \eps] \geq 2\cdot \err(2d,q,q'),
\end{equation}

where the second inequality is the assumption from \cref{thm: prox gen iRM with corr}. We first bound $\Pr[E_1]$.
\begin{claim} \label{cl: prox aux}
    We have $\Pr[E_1] < \err(2d,q,q')$.
\end{claim}
\begin{proof}
    Suppose for the sake of contradiction that the inequality above is false. For each $j \in [2]$, define $f'_j$ by setting $f'_j(x) = f_j(x)$ for all $x \in \Ff_q^m \setminus A$ and setting $f'_j(x) = x_1^{d+1}x_2^d$ for all $x \in A$. Then by our assumption, we have
    \begin{equation}  \label{eq: prox gen application in irm prox gen proof}
    \Pr_{\xi_1,\xi_2 \in \Ff_{q'}}\left[\agr_{2d}\left(\xi_1\cdot f'_1 + \xi_2 \cdot f'_2\right) \geq 0.999\eps \right] \geq \Pr_{\xi_1,\xi_2 \in \Ff_{q'}}[E_1] \geq \err(2d,q,q').
     \end{equation}

    Indeed, if $\xi_1 \cdot f_1 + \xi_2 \cdot f_2$ agrees with some degree $(d,d)$ function on $V$, then $\xi_1 \cdot f'_1 + \xi_2 \cdot f'_2$ agrees with the same function on $V \setminus A$. Thus, if $E_1$ occurs, then there is such $V$ satisfying
    \begin{equation*}   
    \mu(V \setminus A) \geq \mu(V) - \mu(V \cap A) \geq \eps - M \cdot \frac{2d}{q} > 0.999 \eps.
    \end{equation*}
    Applying \cref{thm: prox gen RM with corr} to \cref{eq: prox gen application in irm prox gen proof}, we get that there exists a set $U \subseteq \Ff_{q}^m$ of size at least $0.998\eps$ and a pair of degree $2d$ functions $(g_1, g_2)$ such that $g_j|_U = f'_j|_U$ for all $j \in [2]$. 
    
    However, we claim that this is a contradiction because our construction of $f'_1$ and $f'_2$ guarantees that they cannot have such high correlated agreement with degree $2d$. First, they cannot have significant correlated agreement outside of $A$. Indeed, for any pair of degree $2d$ functions $(g_1, g_2)$, we have, 
    \[
    \Pr_{x \in \Ff_{q}^m \setminus A}[g_1(x) = f'_1(x) \land g_2(x) = f'_2(x)] = \Pr_{x \in \Ff_{q}^m \setminus A}[g_1(x) = f_1(x) \land g_2(x) = f_2(x)] \leq 0.99 \eps.
    \]
    The equality is because $f'_j$ and $f_j$ agree outside of $A$ for $j \in [2]$ and the inequality is by the definition of $\mc{H}$ and $A$. On the other hand, $f'_1$ and $f'_2$ cannot have significant correlated agreement inside of $A$ either as
    \[ 
     \Pr_{x \in A}[g_1(x) = f'_1(x) \land g_2(x) = f'_2(x)] \leq \frac{2d+1}{q}
    \]
    by the Schwartz-Zippel lemma. Combining these two inequalities, we get that any set $U \subseteq \Ff_q^m$ where $f'_1$ and $f'_2$ have correlated agreement with degree $2d$ can have size at most
    \[
    \mu(U) \leq 0.99\eps + \frac{2d+1}{q} < 0.998\eps,
    \]
    contradicting our prior conclusion.
\end{proof}
~\cref{cl: prox aux} and~\eqref{eq: irm prox aux} give that $\Pr[E_2]\geq \err(2d,q,q')$. From here, we can conclude the proof of \cref{thm: prox gen iRM with corr}.

\begin{proof}[Proof of \cref{thm: prox gen iRM with corr}] 
Recall that we assume every $(h^{(i)}_1, h^{(2)}_2)\in \mc{H}$ is a pair of functions whose individual degree is not $(d,d)$. Thus, for each $(h^{(i)}_1, h^{(2)}_2)\in \mc{H}$, we have that $\xi_1 \cdot h^{(i)}_1 + \xi_2 \cdot h^{(i)}_2$ is degree $(d,d)$ with probability at most $1/q'$ when $\xi_1, \xi_2 \in \Ff_{q'}$ are chosen uniformly at random. Letting $E_3$ denote the event that this is the case for some $i \in [M]$, we get that $\Pr[E_3] \leq M/q'$. Thus, we get that 
\[
\Pr_{\xi_1, \xi_2 \in \Ff_{q'}}[E_2 \land \overline{E_3}] > \err(2d,q,q') - \frac{M}{q'} > 0,
\]
so there exist coefficients $\xi_1, \xi_2 \in \Ff_{q'}$ such that the following holds for $F = \xi_1 \cdot f_1 + \xi_2 \cdot f_2$:
\begin{itemize}
    \item There exists a degree $(d, d)$ function $H$ and $V \subseteq \Ff_q^m$ such that $\mu(V \cap A) > M\cdot \frac{2d}{q}$ and $F|_V = H|_V$.
    \item The function $\xi_1 \cdot h^{(i)}_1 + \xi_2 \cdot h^{(i)}_2$ is not degree $(d,d)$ for any $i \in [M]$. In particular, $\xi_1 \cdot h^{(i)}_1 + \xi_2 \cdot h^{(i)}_2 \neq H$.
\end{itemize}
We now reach a contradiction. First, by
a union bound $\mu(V \cap A) \leq \sum_{i=1}^M \mu(V \cap A_i)$.
For each $i\in [M]$ and $x \in V \cap A_i$ we have 
\[
\xi_1 \cdot h^{(i)}_1(x) + \xi_2 \cdot h^{(i)}_2(x) = \xi_1 \cdot f_1(x) + \xi_2 \cdot f_2(x) = F(x) = H(x),
\]
where the first equality is because $x \in A_i$, the second equality is by definition, and the third equality is because $x \in V$. As a result,
\[
\mu(V \cap A_i) \leq \Pr_{x \in \Ff_q^m}[\xi_1 \cdot h^{(i)}_1(x) + \xi_2 \cdot h^{(i)}_2(x) = H(x)] \leq \frac{2d}{q}
\]
by the Schwartz-Zippel lemma and the fact that $\xi_1 \cdot h^{(i)}_1 + \xi_2 \cdot h^{(i)}_2 \neq H$ are distinct functions with total degree at most $2d$. It follows that $\mu(V \cap A) \leq M \cdot \frac{2d}{q}$, and contradiction.
\end{proof}

\section{$\rsi(2\sqrt{d}, f, \lambda, \eps_{k-1})$ Construction: Proof of~\cref{lm: 2sqrtd rsi}} \label{sec: 2sqrtd}
We construct $\rsi(f, 2\sqrt{d}, q, \lambda, \eps_{k-1})$ according to \cref{lm: 2sqrtd rsi} here. As the IOPP is fairly simple, we give the construction of $\rsi(f, 2\sqrt{d}, q, \lambda, \eps_{k-1})$ directly, instead of constructing a Poly-IOP first and then compiling it.

\begin{algorithm}[H]\label{iop: 2sqrtd}
  \caption{Reed-Solomon IOP: $\rsi(f, 2\sqrt{d}, q, \lambda, \eps_{k-1})$}
  \begin{algorithmic}[1]
    \STATE \textbf{P:} The prover sends $f_1: \Ff_q \to \Ff_{q'}$. We will assume that $f_1(0) = f(0)$. In order for this assumption to work, the verifier queries $f(0)$ to obtain the value of $f_1(0)$.
    \STATE  Define the function $f_2: \Ff_q \to \Ff_{q'}$ as follows:
   \begin{equation*}
       f_2(x) = 
       \begin{cases}
       \frac{f(x) - f_1(x)}{x^{\sqrt{d}}} \text{ if } x \neq 0, \\
       0 \text{ if } x = 0.
       \end{cases}
   \end{equation*}
   In the honest case, $f_2$ has degree at most $\sqrt{d}$ and, 
   \[
   f(x) = f_1(x) + x^{\sqrt{d}}f_2(x), \; \forall x \in \Ff_q.
   \]
   \STATE \textbf{V:} The verifier chooses $\gamma \in \Ff_{q'}$ uniformly at random and sets $g: \Ff_q \to \Ff_{q'}$ to be
   \[
   g(x) = f_1(x) +  \gamma\cdot f_2(x).
   \]
   \STATE \textbf{P+V:} Both parties run $\rsi(g, \sqrt{d}, q, \lambda, \eps_{k-1})$.
\end{algorithmic}
\end{algorithm}

\begin{proof}[Proof of \cref{lm: 2sqrtd rsi}]
    The length, round complexity, and query complexities are straightforward to verify.  For the completeness, if $f \in \RS_{q'}[2\sqrt{d}, \Ff_q]$, then its low degree extension $\wh{f} \in \Ff_{q'}^{\leq 2\sqrt{d}}[x]$ can be decomposed as $\wh{f} = \wh{f}_1 + x^{\sqrt{d}}\wh{f}_2$
    for $\wh{f}_1, \wh{f}_2 \in \Ff_{q'}^{\leq \sqrt{d}}[x]$.
    The honest prover sends $f_1$ to be the evaluation of $\wh{f}_1$ over $\Ff_q$.  In this case, one can check that $f_2$ as defined in step 2 is the evaluation of $\wh{f}_2$ over $\Ff_q$. Thus, both $f_1, f_2 \in \RS_{q'}[\sqrt{d}, \Ff_q]$. It follows that $g \in \RS_{q'}[\sqrt{d}, \Ff_q]$ with probability $1$ and the remainder of the completeness follows from that of $\rsi(f, \sqrt{d}, q, \lambda, \eps_{k-1})$.

    For the soundness case, suppose that 
    \begin{equation} \label{eq: 2sqrt soundness assumption} 
    \agr(f, \RS_{q'}[2\sqrt{d}, \Ff_q]) \leq \eps_{k-1}.
     \end{equation}
    Note that
    \begin{equation}  \label{eq: 2sqrtd rbr}
    \Pr_{\gamma \in \Ff_{q'}}[\agr(g, \RS_{q'}[\sqrt{d}, \Ff_q]) \geq \eps_{k-1}] \leq \frac{\poly(\sqrt{d},q)}{q'} \leq 2^{-\lambda}.
    \end{equation}
    Indeed, otherwise by~\cref{lm: combine one univariate} $f_1$ and $f_2$ have $\eps_{k-1}$ correlated agreement with degree $\sqrt{d}$ functions $F_1, F_2: \Ff_q \to \Ff_{q'}$, implying $\agr(f, F_1 + x^{\sqrt{d}}\cdot F_2) \geq \eps_{k-1}$.
    Since $F_1 + x^{\sqrt{d}}\cdot F_2 \in \RS_{q'}[2\sqrt{d}, \Ff_q]$, this contradicts~\eqref{eq: 2sqrt soundness assumption}. 
    
    With~\eqref{eq: 2sqrt soundness assumption} established, the round-by-round soundness then follows from the round-by-round soundness of $\rsi(g, \sqrt{d}, q, \lambda, \eps_{k-1})$.
\end{proof}

\section{Proofs for \cref{sec: iop for rics}} \label{app: rics proofs}

\subsection{Univariate Sumcheck}
Towards the proof of \cref{lm: sumcheck poly iop}, we will first describe a univariate sumcheck Poly-IOP, and then use this in the multivariate sumcheck Poly-IOP. The following fact will be useful.
\begin{fact} \label{fact: sum}
    If $H$ is a multiplicative subgroup of $\Ff_{q'}$ and $\wh{R} \in \Ff_{q'}^{\leq |H|-1}[x]$ be a polynomial, then
    \[
    \sum_{\alpha \in H} \wh{R}(\alpha) = \wh{R}(0) \cdot |H|.
    \]
\end{fact}
\begin{proof}
Writing $H = \{w^0,\ldots,w^{|H|-1}\}$ for some $w$,
for every $1\leq j\leq |H|-1$ we have that 
\[
\sum\limits_{w\in H} w^j 
= \sum\limits_{i=0}^{|H|-1}(w^j)^{i}
=(w^j-1)^{-1}((w^j)^{|H|}-1) = 0
\]
via the geometric sum formula. Expanding $R$ into a linear combination of monomials and using this fact gives the result.
\end{proof}

\begin{lemma} \label{lm: univariate sumcheck poly iop}
    There is a Poly-IOP with the following guarantees:
    \begin{itemize}
        \item Input: A Polynomial $\wh{f} \in \Ff_{q'}^{\leq d'}[x]$, a multiplicative subgroup $H \subseteq \Ff_{q'}$ of size $d+1$, and a target sum $\alpha$.
        \item Completeness: If $\sum_{b \in H} \wh{f}(b) = \alpha$ then the honest prover makes the verifier accept with probability $1$.
        \item Round-by-Round Soundness: $\frac{d'}{q'}$.
        \item Initial State: The initial state is doomed if and only if  $\sum_{b \in H} \wh{f}(b) \neq \alpha$.
        \item Number of Polynomials: $O(1)$.
        \item Round Complexity: $O(1)$.
        \item Input Query Complexity: $\wh{f}$ is queried $1$ time.
        \item Proof Query Complexity: $O(1)$.
    \end{itemize}
\end{lemma}
\begin{proof}
    We begin with a formal description of the Poly-IOP:
    \begin{algorithm}[H]
    \caption{Univariate Sumcheck Poly-IOP} \label{poly iop: univar sumcheck}
    \begin{algorithmic}[1]
        \STATE \textbf{P:} The prover sends $\wh{F} \in \Ff_{q'}^{\leq d' - |H|}[x]$, and $\wh{R} \in \Ff_{q'}^{\leq |H|-2}[x]$. In the honest case,
        \[
        \wh{f} = \frac{\alpha}{|H|} + \wh{V}_H \cdot \wh{F} + x \cdot \wh{R}.
        \]
        \STATE \textbf{V:} The verifier chooses $\gamma \in \Ff_{q'}$ uniformly at random and checks if
        \[
        \wh{f}(\gamma) = \frac{\alpha}{|H|} + \wh{V}_H(\gamma) \cdot \wh{F}(\gamma) + \gamma \cdot \wh{R}(\gamma).
        \]
        If the check fails then the verifier rejects.
    \end{algorithmic}
\end{algorithm}
    The number of polynomials, round complexity, input query complexity, and proof query complexity are clear. 
    
    \paragraph{Completeness:} assume that $\sum_{b \in H} \wh{f}(b) = \alpha$. By~\cref{lm: side condition vanish decomp} we may write 
    \[
    \wh{f} = \wh{V}_H \cdot \wh{F} + \wh{R'},
    \]
    for $\wh{F} \in \Ff_{q'}^{\leq d' - |H|}[x]$, and $\wh{R}' \in \Ff_{q'}^{\leq |H|-1}[x]$ that agrees with $\wh{f}$ on $H$. Furthermore, by \cref{fact: sum}, it must be the case that $\wh{R}(0) = \alpha/|H|$. Therefore, the honest prover can find $\wh{F} \in \Ff_{q'}^{\leq d' - |H|}[x]$ and $\wh{R} \in \Ff_{q'}^{\leq |H|-2}[x]$ which satisfy 
    \[
    \wh{f} = \alpha/|H| + \wh{V}_H \cdot \wh{F} + x\wh{R}.
    \]
    It follows that the equation in step 2 is satisfied for every $\gamma$ and the verifier always accepts.

   \paragraph{Round-by-round soundness:} there is only one round of interaction, after which we define the state as doomed if and only if the check fails and the verifier rejects. If the initial state is doomed, then it follows that for any  $\wh{F} \in \Ff_{q'}^{\leq d' - |H|}[x]$ and $\wh{R} \in \Ff_{q'}^{\leq |H|-2}[x]$, 
    \[
      \wh{f} \neq \alpha/|H| + \wh{V}_H \cdot \wh{F} + x\wh{R}.
    \]
    This is because the right hand side is a polynomial summing to $\alpha$ over $H$ by \cref{fact: sum}. As the polynomial on the right hand side has degree $d'$, by the Schwartz-Zippel lemma, the check in step $2$ fails with probability at least $1 - d'/q'$.
\end{proof}

\subsubsection{Proof of \cref{lm: sumcheck poly iop}}
The IOP is described in \cref{poly iop: mvar sumcheck}.
\begin{algorithm}[H]
    \caption{Multivariate Sumcheck Poly-IOP} \label{poly iop: mvar sumcheck}
    \begin{algorithmic}[1]
    \STATE \textbf{P: } The prover sends $\wh{F}_1 \in \Ff^{\leq 6d}_{q'}[x]$. In the honest case,
    \[
    \wh{F}_1(x) = \sum_{\gamma_2, \gamma_3 \in H} \wh{f}(x, \gamma_2, \gamma_3),
    \]
    and 
    \begin{equation} \label{eq: sumcheck equation 1}   
    \sum_{\gamma \in H} \wh{F}_1(\gamma) = b.
    \end{equation}
    \STATE \textbf{V: } The verifier chooses $\zeta_1 \in \Ff_{q'}$ uniformly at random and sends it to the prover. The verifier queries $\wh{F}_1(\zeta_1)$.
    \STATE \textbf{P: } The prover sends $\wh{F}_2 \in \Ff_{q'}^{\leq 6d}[x]$. In the honest case,
    \[
    \wh{F}_2(x) = \sum_{\gamma_3 \in H} \wh{f}(\zeta_1, x, \gamma_3),
    \]
    and
    \begin{equation}\label{eq: sumcheck equation 2}
            \sum_{\gamma \in H} \wh{F}_2(\gamma) = \wh{F}_1(\zeta_1).
    \end{equation}
    \STATE \textbf{V: } The verifier chooses $\zeta_2 \in \Ff_{q'}$ uniformly at random and queries $\wh{F}_2(\zeta_2)$.
    \STATE \textbf{P + V: } Both parties run the univariate sumcheck protocol to check~\eqref{eq: sumcheck equation 1},~\eqref{eq: sumcheck equation 2}, and
     \begin{equation} \label{eq: sumcheck equation 3}  
     \sum_{\gamma_3 \in H} \wh{f}(\zeta_1, \zeta_2, \gamma_3) = \wh{F}_2(\zeta_2),
    \end{equation}
    hold. 
     \STATE \textbf{V:} The verifier accepts if they reach this stage. That is, they accept if they have not rejected yet during steps 2,5, or 7.
    \end{algorithmic}
\end{algorithm}
The number of polynomials, round complexity, input query complexity, and proof query complexity are clear. We discuss the completeness and Round-by-Round soundness below.
\paragraph{Completeness.} In the completeness case, the prover sends $\wh{F}_1(x) = \sum_{\gamma_2, \gamma_3 \in H} \wh{f}(x, \gamma_2, \gamma_3)$. Note that indeed $\wh{F}_1 \in \Ff_{q'}^{\leq 6d}[x,y]$ and that~\eqref{eq: sumcheck equation 1} is satisfied. As a result, the univariate sumcheck in step 2 passes with probability $1$ by the completeness in \cref{lm: univariate sumcheck poly iop}. In step $4$, the prover sends $\wh{F}_2 = \sum_{\gamma_3 \in H}\wh{f}(\zeta_1, x, \gamma_3)$. Once again $\wh{F}_2 \in \Ff_{q'}^{\leq 6d}[x]$ and Equation~\eqref{eq: sumcheck equation 2} is satisfied, so step 5 always passes by the completeness in \cref{lm: univariate sumcheck poly iop}. We also have $\sum_{\gamma_3 \in H} \wh{f}(\zeta_1, \zeta_2, \gamma_3) = \wh{F}_2(\zeta_2)$ in this case, so the sumcheck in step 7 is once passes by the completeness in \cref{lm: univariate sumcheck poly iop} and the verifier accepts at the end.

\paragraph{Round-by-Round Soundness.}
After the first round of interaction (step 2), the state function is as follows.

\begin{state} \label{state: sumcheck state 1}
The state is doomed if and only if at least one of the following holds
\begin{itemize}
    \item $\wh{F}_1(\zeta_1) \neq \sum_{\gamma_2, \gamma_3}\wh{f}(\zeta_1, \gamma_2, \gamma_3)$.
    \item $\sum_{\gamma \in H} \wh{F}_1(\gamma) \neq b$.
\end{itemize}
\end{state}

After the second round of interaction (step 4), the state function is as follows.
\begin{state}\label{state: sumcheck state 2}
The state is doomed if and only if at least one of the following holds
\begin{itemize}
    \item $\wh{F}_2(\zeta_2) \neq \sum_{\gamma_3}\wh{f}(\zeta_1, \zeta_2, \gamma_3)$.
    \item $\sum_{\gamma \in H} \wh{F}_2(\gamma) \neq \wh{F}_1(\zeta_1)$.
    \item $\sum_{\gamma \in H} \wh{F}_1(\gamma) \neq b$.
\end{itemize}
\end{state}

\begin{proof}
    Suppose the initial state is doomed, meaning $\sum_{\gamma_1, \gamma_2, \gamma_3 \in H} \wh{f}(\gamma_1, \gamma_2, \gamma_3) \neq b$. If the prover sends $\wh{F}_1 \in \Ff_{q'}^{\leq 6d}[x]$ not satisfying~\eqref{eq: sumcheck equation 1}, then \cref{state: sumcheck state 1} is doomed, and soundness of the first round is done. Otherwise,~\eqref{eq: sumcheck equation 1} is satisfied and it must be the case that 
    \[
    \wh{F}_1(x) \neq \sum_{\gamma_2, \gamma_3 \in H} \wh{f}(x, \gamma_2, \gamma_3).
    \]
    Note that both sides of the above equation are degree at most $6d$ polynomials in $x$. By the Schwartz-Zippel lemma, we have 
    \[
    \wh{F}_1(\zeta_1) \neq \sum_{\gamma_2, \gamma_3\in H}\wh{f}(\zeta_1, \gamma_2, \gamma_3)
    \]
    with probability at least $1 - 6d/q' \geq 1-2^{-\lambda}$. This shows soundness for the first round of interaction.

    For the second round of interaction, suppose that \cref{state: sumcheck state 1} is doomed. If it is due to the second item in \cref{state: sumcheck state 1}, then \cref{state: sumcheck state 2} is automatically doomed by item 3 of  \cref{state: sumcheck state 2}. If the prover sends $\wh{F}_2$ not satisfying~\eqref{eq: sumcheck equation 2}, then  \cref{state: sumcheck state 2} is automatically doomed by item 2 of \cref{state: sumcheck state 2}. Otherwise, it must be the case that 
  \[
    \wh{F}_2(x) \neq \sum_{\gamma_3 \in H} \wh{f}(\zeta_1, x, \gamma_3),
    \]
    in which case item 1 of \cref{state: sumcheck state 2} holds with probability at least $1-6d/q' \geq 2^{-\lambda}$. This establishes soundness of the second round of interaction.

    For the remainder of the Round-by-Round soundness, note that going into step 5, the state is doomed if and only if at least one of \eqref{eq: sumcheck equation 1}, \eqref{eq: sumcheck equation 2}, or \eqref{eq: sumcheck equation 3} is not satisfied. In this case the Round-by-Round soundness of \cref{lm: univariate sumcheck poly iop} takes care of the soundness of the remaining rounds.  
\end{proof}

\end{document}